\documentclass[]{article}
\usepackage{authblk}
\usepackage{xcolor}
\usepackage{amssymb}
\usepackage{amsmath}
\usepackage{amsthm}
\usepackage{mathtools}
\usepackage{changepage}
\usepackage{color}
\usepackage{upgreek}
\usepackage{dsfont}
\usepackage{cite}
\usepackage[hidelinks]{hyperref}
\usepackage{qcircuit}
\usepackage{enumitem}
\usepackage{cancel}
\usepackage{circledsteps}
\usepackage{subcaption}
\usepackage{float}
\usepackage{mathrsfs}
\usepackage{array}
\usepackage{footnote}
\makesavenoteenv{tabular}
\makesavenoteenv{table}
\usepackage{savesym}
\savesymbol{checkmark}
\usepackage{dingbat}
\usepackage{tikz}

\usepackage[format=hang]{caption}

\usepackage{graphicx}
\usetikzlibrary{matrix}
\usetikzlibrary{shapes,arrows}
\usepackage{enumitem}

\usepackage{orcidlink}
\usepackage{pgfplots}
\pgfplotsset{compat=1.14}

\usepackage[noend]{algpseudocode}
\usepackage{algorithm}

\newcommand{\expp}[1]{\exp\left(#1\right)}
\newcommand{\sinp}[1]{\sin\left(#1\right)}
\newcommand{\cosp}[1]{\cos\left(#1\right)}
\newcommand{\sinpt}[1]{\sin^2\left(#1\right)}
\newcommand{\cospt}[1]{\cos^2\left(#1\right)}
\newcommand{\csp}[1]{\cosp{#1} \sinp{#1}}

\newcommand{\sgn}[1]{\mathrm{sgn}\left(#1\right)}

\newcommand{\ket}[1]{\left| #1 \right\rangle}
\newcommand{\bra}[1]{\left \langle #1 \right |}

\newcommand{\kbra}[2]{|#1\rangle\!\langle #2|}
\def\wt{\operatorname{wt}}
\def\e{\ensuremath{\mathrm{e}}}
\def\i{\ensuremath{\mathrm{i}}}
\def\d{\ensuremath{\mathrm{d}}}
\def\AF{\textup{AF}}
\def\AB{\textup{AB}}

\DeclareMathAlphabet{\mathsfit}{T1}{\sfdefault}{\mddefault}{\sldefault}
\SetMathAlphabet{\mathsfit}{bold}{T1}{\sfdefault}{\bfdefault}{\sldefault}
\def\Var{\mathrm{Var}}
\def\tr{\mathrm{tr}}

\def\Pr{\mathrm{Pr}}

\DeclareMathAlphabet{\mathpzc}{OT1}{pzc}{m}{it}

\def\qA{\mathpzc{q}}
\def\rA{\mathpzc{r}}

\theoremstyle{plain}
\newtheorem{thm}{Theorem}
\newtheorem{lemma}[thm]{Lemma}
\newtheorem{proposition}[thm]{Proposition}

\theoremstyle{definition}
\newtheorem{definition}[thm]{Definition}
\newtheorem*{remark*}{Remark}

\usepackage{multirow}
\usepackage{array}

\DeclareMathOperator*{\argmax}{arg\,max}

\usepackage{fancyhdr}
\title{\bf Foundations for Bayesian inference with engineered likelihood functions for robust amplitude estimation

}
\author[1,2]{Dax Enshan Koh\orcidlink{0000-0002-8968-591X}\thanks{\texttt{dax\_koh@ihpc.a-star.edu.sg}}}
\author[3]{Guoming Wang\orcidlink{0000-0002-0768-2644}\thanks{\texttt{guoming.wang@zapatacomputing.com}}}
\author[1]{Peter D.~Johnson\orcidlink{0000-0003-3470-1246}\thanks{\texttt{peter@zapatacomputing.com}}}
\author[1]{Yudong Cao\orcidlink{0000-0003-0201-3764 }\thanks{\texttt{yudong@zapatacomputing.com}}}

\affil[1]{\normalsize Zapata Computing~Inc., 100 Federal Street, Boston, Massachusetts 02110, USA}

\affil[2]{\normalsize Institute of High Performance Computing,
Agency for Science, Technology and Research (A*STAR), 1
Fusionopolis Way, \#16-16 Connexis, Singapore 138632, Singapore}

\affil[3]{\normalsize Zapata Computing~Inc., 325 Front Street West, Suite 300, Toronto,
ON M5V 2Y1, Canada}

\date{}
\begin{document}
\maketitle

\begin{abstract}

We present mathematical and conceptual foundations for the task of robust amplitude estimation using engineered likelihood functions (ELFs), a framework introduced in Wang et al.~[PRX Quantum 2, 010346 (2021)] that uses Bayesian inference to enhance the rate of information gain in quantum sampling. These ELFs, which are obtained by choosing tunable parameters in a parametrized quantum circuit to minimize the expected posterior variance of an estimated parameter, play an important role in estimating the expectation values of quantum observables. We give a thorough characterization and analysis of likelihood functions arising from certain classes of quantum circuits and combine this with the tools of Bayesian inference to give a procedure for picking optimal ELF tunable parameters. Finally, we present numerical results to demonstrate the performance of ELFs.

\end{abstract}

\section{Introduction}
\label{sec:introduction}

Likelihood functions play a fundamental role in various quantum algorithms ranging from quantum channel parameter estimation \cite{4655455}, quantum metrology \cite{giovannetti2006quantum}, to quantum phase estimation \cite{svore2013faster,wiebe2016efficient,Zintchenko2016} and amplitude estimation \cite{suzuki2020amplitude}. The ability to realize likelihood functions that are otherwise infeasible with only classical resources allows one to take advantage of quantum mechanics to significantly accelerate sampling and measurement processes. The problems encountered in many of these settings \cite{4655455,giovannetti2006quantum,suzuki2020amplitude} can be likened to the problem of estimating the bias $q$ of a coin, where $q\in[-1,1]$ is an unknown parameter. In the classical coin flip setting, we draw samples from the likelihood function $p(d|q)=\frac 12[1+(-1)^dq]$, where $d\in\{0,1\}$ is the outcome of the coin toss. In the quantum setting, as will become clear, we are able to draw samples from a modified likelihood function of the form $p(d|q)=\frac 12[1+(-1)^df(q)]$ where $f(q)$ is a function that depends on the specifics of the quantum scheme. This opens up fundamentally new opportunities for accelerating parameter estimation\footnote{It is also possible to classically generate likelihood functions that have a non-linear dependence on $q$ by combining multiple flips. However, one can demonstrate that the class of functions $f$ possible with quantum schemes are clearly hard to produce classically.}. 

Here we focus on the task of estimating the overlap $\Pi$ between two (pure) quantum states that differ by a unitary transformation $U$, i.e.~$\Pi = \bra A U \ket A$ for some states $\ket A$ and $U\ket A$. This task encompasses a broad variety of quantum algorithms and techniques, including the SWAP test \cite{PhysRevLett.122.040504} and variants for estimating general state overlaps, expectation estimation \cite{knill2007optimal,wang2019accelerated} in the special case where $U$ is also a Hermitian operator, and phase estimation \cite{quant-ph/9511026,svore2013faster,wiebe2016efficient,O_Brien_2019} in the special case where $U=\e^{-\i Ht}$ for some Hermitian operator $H$ and both states are identical and an eigenvector\footnote{There are also cases where the eigenvector requirement can be dropped for phase estimation but these do not fit the overlap estimation paradigm in a straightforward way.} of $U$. Though some of these works may not explicitly refer to the output distribution $p(d|\Pi)$  of the measurement outcome $d$ as a likelihood function, it can be argued that the quantum advantage of these schemes derive from the ability to realize likelihood functions that are beyond classical possibilities.

In this work, we are motivated by likelihood functions as a way of shedding light on an origin of quantum advantage that has perhaps been previously underappreciated. Our specific approach is to investigate the inferential power of not only the likelihood functions that have commonly arisen in the literature (which we call ``Chebyshev likelihood functions" (CLF) in this paper), but also a broad class of tunable likelihood functions that can be generated by quantum mechanical means (namely the ``quantum-generated likelihood functions"). 
Our results show that the quantum-generated likelihood functions allow for more effective information gain than CLFs during inference, leading to further quantum speedup than previous CLF-based schemes. In addition, we develop an alternative scheme for amplitude estimation based on quantum-generated likelihood functions that does not involve ancilla qubits. A similar scheme has also been considered recently \cite{suzuki2020amplitude} using CLFs.

Our work provides the mathematical and conceptual foundations for the algorithms presented in Wang et al.\ \cite{ELFPaper}, where efficient schemes based on engineering likelihood functions were introduced to estimate the expectation values of quantum observables. These engineered likelihood functions (ELFs) are obtained by choosing tunable parameters in a parameterized quantum circuit that minimize the expected posterior variance of an estimated parameter. The algorithm in \cite{ELFPaper} utilizes a multi-round Bayesian updating scheme to reduce the expected posterior variance of the parameter of interest. In this manuscript, we go beyond \cite{ELFPaper} in two ways. First, we provide a thorough characterization and analysis of likelihood functions arising from certain classes of quantum circuits that forms the basis for the approximate optimization methods used in \cite{ELFPaper}. In particular, we derive explicit series expansions formulas for these functions that are useful for proving properties of these functions. Second, unlike \cite{ELFPaper} which uses a Gaussian distribution to represent one's prior knowledge of the unknown parameter, our analysis is carried out for arbitrary prior distributions. While \cite{ELFPaper} considers the more realistic case of noisy ELFs, this paper concentrates on the idealized case where the ELFs are noiseless. It turns out that studying this idealized case already suffices for gaining useful insights on the behavior and properties of these ELFs.

The rest of the paper is organized as follows. In Section \ref{sec:quantumGeneratedLikelihoodFunctions}, we introduce the concept of quantum-generated likelihood functions through the lens of amplitude estimation. We describe two schemes, called the ancilla-based and ancilla-free schemes, and derive analytical expressions for their likelihood functions as cosine polynomials and trigono-multivariate polynomial functions. In Section \ref{sec:bayesianinference}, we 
define our objective function, namely the expected posterior variance, and derive simplified expressions for it as we make various assumptions. In Section \ref{sec:optimization},
we combine the results of Section \ref{sec:quantumGeneratedLikelihoodFunctions} and \ref{sec:bayesianinference} to show how the tunable parameters in the likelihood functions can be optimized to reduce the expected posterior variance. We also present some numerical results based on solving this optimization problem that demonstrate the performance of various ELFs.  For a glossary of some of the definitions and notation used in this paper, see Appendix \ref{sec:notation}.

\section{Quantum-generated likelihood functions}
\label{sec:quantumGeneratedLikelihoodFunctions}

\subsection{Formulation of problem}
\label{sec:formulation}

In this section, we introduce the formulation of the problem considered in \cite{ELFPaper}. For positive integers $n \in \mathbb Z^+$, let $A$ be an $n$-qubit unitary operator and $P$ be an $n$-qubit Hermitian unitary operator, i.e.
\begin{align}
    A^\dag A = I, \quad P=P^\dag, \quad P^2 = I.
\end{align}
Write $\ket{A} = A\ket {0^n}$, where $\ket {0^n}$ is the all-zero computational basis state on $n$ qubits. The computational task that we consider in this paper is the estimation of the (arccosine of the) expectation value
\begin{align}
\theta = \arccos(\bra{A} P \ket{A}) \in [0,\pi].
\label{eq:thetaDef}
\end{align}

For simplicity, we assume that $\theta \notin \{0,\pi\}$, i.e. 
\begin{align}
 |\bra{A} P \ket{A}| \neq 1.   
 \label{eq:assumptionTheta}
\end{align}

Consider the subspace $\mathscr S=\mathrm{span}\{\ket{A}, P\ket{A}\}$. Note that the assumption in  Eq.~\eqref{eq:assumptionTheta} implies that $\mathscr S$ is a two-dimensional subspace. This allows us to define $\ket{A^\bot}$ as the unique state in $\mathscr S$ that is orthogonal to $\ket{A}$, i.e.
\begin{align}
    \ket{A^\bot} = \frac{P\ket A - \bra A P \ket A \ket A}{\sqrt{1-\bra A P \ket A^2}}.
\end{align}
By construction, $\mathscr B = \{\ket{A},\ket{A^\bot}\}$ forms an orthonormal basis for $\mathscr S$. Henceforth, we shall label $\ket{A}$ and $|A^\bot\rangle$ as $\ket{\bar 0} = \begin{pmatrix} 1 \\ 0 \end{pmatrix}$ and $\ket{\bar 1} = \begin{pmatrix} 0 \\ 1 \end{pmatrix}$ respectively and refer to $\mathscr S$ as the \textit{logical space}, which is to be contrasted with the \textit{physical space} describing the $n$-qubit system. We shall extend the bar notation to certain operators on $\mathscr S$. For example, we will denote the Pauli matrices with respect to the basis $\mathscr B$ by
\begin{align}
    \bar I = 
    \begin{pmatrix}
    1 & 0 \\ 0 & 1
    \end{pmatrix},
    \quad
    \bar X = 
    \begin{pmatrix}
    0 & 1 \\ 1 & 0
    \end{pmatrix},
    \quad
    \bar Y = 
    \begin{pmatrix}
    0 & -\i \\ \i & 0
    \end{pmatrix},
    \quad
    \bar Z = 
    \begin{pmatrix}
    1 & 0 \\ 0 & -1
    \end{pmatrix}.
\end{align}
In the basis $\mathscr B$, $P$ may be written as
\begin{align}
P &= \cos(\theta) \bar Z + \sin(\theta) \bar X =
\begin{pmatrix}
\cos \theta & \sin\theta \\
\sin\theta & -\cos\theta
\end{pmatrix}.
\label{eq:Pmatrix}
\end{align}
To indicate explicit dependence of $P$ on $\theta$, we will write $P=P(\theta)$.

Next, we define the following two types of unitary operators:
\begin{align}
U(\theta; \alpha) &= \expp{-\i \alpha P(\theta)} \nonumber \\
&= \cosp{\alpha} I - \i \sinp{\alpha} P(\theta) \nonumber \\
&= \cosp{\alpha} \bar I - \i \sinp{\alpha} (\cosp{\theta} \bar Z + \sinp{\theta} \bar X),
\label{eq:ualpha}
\end{align}
and
\begin{align}
V(\beta) &= \expp{-\i \beta (2\kbra {\bar 0}{\bar 0}-I)} \nonumber \\
&= \expp{-\i \beta \bar Z} \nonumber \\
&= \cosp{\beta} \bar I - \i \sinp{\beta} \bar Z, 
\label{eq:vbeta}
\end{align}
where $\alpha, \beta \in \mathbb{R}$. 

In constructing our estimation algorithm, we assume that we are able to perform the following primitive operations. First, we assume that we are able to prepare computational basis states $\ket {0^n}$ and apply $A$ to them to obtain $\ket {\bar 0}$. Next, we assume that we are allowed to apply the following operations on the physical qubits: $U(\theta;\alpha)$ for any angle $\alpha \in \mathbb R$ and $V(\beta)$ for any $\beta \in \mathbb R$, as well as controlled versions of these operations, namely
$\textup{c-}U(\theta;\alpha) = \kbra{0}{0} \otimes I + \kbra{1}{1} \otimes U(\theta;\alpha)$ and $\textup{c-}V(\beta) = \kbra{0}{0} \otimes I + \kbra{1}{1} \otimes V(\beta)$. \footnote{The implementation of $U(\alpha) \coloneqq U(\theta;\alpha)$ and $ V(\beta)$ depends on $P \coloneqq P(\theta)$ and $A$ respectively. Specifically, $V(\beta)$ can be implemented based on the equation
$V(\beta)=A^\dagger \expp{-\i \beta (2\kbra {0^n}{0^n}-I)} A$, which requires one use of $A$ and $A^\dagger$, and $\expp{-\i \beta (2\kbra {0^n}{0^n}-I)}$ can be implemented with $O(n)$ single- and two-qubit gates with the help of $n-1$ ancilla qubits. As to $U(\alpha)$, suppose $P$ has the spectral decomposition $P=Q^\dagger \Lambda Q$, where $Q$ is unitary and $\Lambda=\sum_{j \in S} \ket{j}\bra{j} - \sum_{j \in \{0,1\}^n \setminus S} \ket{j} \bra{j}$ for some $S \subset \{0,1\}^n$. Then 
$U(\alpha)$ can be implemented based on the equation
$U(\alpha)=Q^\dagger \expp{-\i \alpha \Lambda} Q$, which requires one use of $Q$ and $Q^\dagger$, and the implementation of $\expp{-\i \alpha \Lambda}$ depends on $\Lambda$. In many applications, $P$ is a Pauli operator, and in this case, $Q$ is a Clifford operator, and we may assume without loss of generality that $\Lambda=Z \otimes I \otimes I \dots \otimes I$, which means that $\expp{-\i \alpha \Lambda}$ can be simply implemented as a $Rz$ gate.}
Furthermore, we assume that we can apply the single-qubit Hadamard gate $H = \frac{1}{\sqrt 2}
\left(\begin{smallmatrix}
1 & 1 \\ 1 & -1
\end{smallmatrix}
\right)
$ 
to the physical qubits.
Finally, we assume that we are allowed to perform computational basis measurements as well as measurements of the operator $P$ (here, a measurement of $P$ corresponds to the projection-valued measure $\{ \frac{I+P}{2},\frac{I-P}{2}\}$).

\subsection{Schemes for expectation estimation}
\label{sec:schemes}

We now describe two different schemes for expectation estimation that we will focus on in this paper. Each of these schemes, which will use only the primitive operations that were listed in Section \ref{sec:formulation}, depends on some tunable parameters, which we collect in a single vector $\vec x = (x_1, x_2,\ldots, x_{2L}) \in \mathbb R^{2L}$. Here, the variable $L\in \mathbb Z^+$ represents half the number of tunable parameters in each scheme. The vector $\vec x$ is tunable in the sense that we will later (specifically, in Section \ref{sec:optimization}) tune it to optimize some objective function. For this section, however, it would suffice to treat $\vec x$ as fixed.

The first scheme, called the \textit{ancilla-free} (AF) scheme proceeds as follows:

\begin{algorithm}[H]
\caption{Ancilla-free (AF) scheme}
\begin{algorithmic}[1]
\State Start with the $n$-qubit state $\ket{\bar 0}$
\For {$i=1,\ldots,L$}
\State Apply the operator $U(\theta;x_{2i-1})$
\State Apply the operator
    $V(x_{2i})$
\EndFor
\State Perform a measurement corresponding to the Hermitian operator $P$, i.e. perform the projective measurement $(\frac{I+P}2,\frac{I-P}2)$ with outcomes $(0,1)$
\end{algorithmic}
\end{algorithm}

The second scheme, called the \textit{ancilla-based} (AB) scheme, uses the Hadamard test and proceeds as follows:

\begin{algorithm}[H]
\caption{Ancilla-based (AB) scheme}
\begin{algorithmic}[1]
\State Start with the input $\ket {\bar 0}$ in the $n$-qubit data register and $\ket 0$ in the single-qubit ancilla register
\State Apply the Hadamard gate $H$ to the ancilla register
\For {$i=1,\ldots,L$}
\State \parbox[t]{\dimexpr\textwidth-\leftmargin-\labelsep-\labelwidth}{
Apply the operator  $\textup{c-}U(\theta;x_{2i-1}) = \kbra{0}{0} \otimes I + \kbra{1}{1} \otimes U(\theta;x_{2i-1})$, with the control being the ancilla register and the target being the data register
  \strut}
\State \parbox[t]{\dimexpr\textwidth-\leftmargin-\labelsep-\labelwidth}{
Apply the operator  $\textup{c-}V(x_{2i}) = \kbra{0}{0} \otimes I + \kbra{1}{1} \otimes V(x_{2i})$, with the control being the ancilla register and the target being the data register
  \strut}
\EndFor
\State Apply the Hadamard gate $H$ to the ancilla register
\State Perform a computational-basis measurement on the ancilla register to obtain outcomes $(0,1)$
\end{algorithmic}
\end{algorithm}

Circuit diagrams illustrating these schemes are given in Figure \ref{fig:elfcircuit}, where we have
introduced the operator
\begin{align}
Q(\theta; \vec x)=V(x_{2L})U(\theta; x_{2L-1}) V(x_{2L-2}) U(\theta; x_{2L-3})\dots V(x_4)U(\theta; x_3)V(x_2)U(\theta; x_1),
\label{eq:qalphabeta}
\end{align}
to describe the overall unitary operator resulting from $2L-1$ alternations of the rotations $U(\theta;\cdot)$ and $V(\cdot)$. For each of these schemes, we denote the single-bit measurement outcome by $d \in \{0,1\}$.

\begin{figure}[!htbp]
Ancilla-free circuit:\vspace{-0.5cm}
\begin{align*}
\Qcircuit @C=1em @R=0.3em @!R {
\lstick{\ket {\bar 0}} & \qw & \gate{Q(\theta;\vec x)} & \qw & \measureD{P} & \cw & \rstick{d\in\{0,1\}}
}
\end{align*}
\vspace{0.1cm}

Ancilla-based circuit:
\vspace{-0.5cm}
\begin{align*}
\Qcircuit @C=1em @R=0.3em @!R {
\lstick{\ket 0} & \gate H & \ctrl 1 & \gate H & \meter & \rstick{d\in\{0,1\}} \cw \\
\lstick{\ket {\bar 0}} & \qw & \gate{Q(\theta;\vec x)} & \qw & \qw & \qw
}
\end{align*}
where
\begin{align*}
\Qcircuit @C=1em @R=0.3em @!R { 
& \gate{Q(\theta;\vec x)} & \qw 
& = & & \gate{U(\theta;x_1)} & \gate{V(x_2)} & \gate{U(\theta;x_3)} & \gate{V(x_4)} & \qw & \cdots & & \gate{U(\theta;x_{2L-1})} & \gate{V(x_{2L})} & \qw
}
\end{align*}
\caption{Quantum circuits for both the ancilla-free and ancilla-based schemes. The output of each of these circuits is a single bit, denoted by $d\in \{0,1\}$.}
\label{fig:elfcircuit}
\end{figure}

Each of these schemes is associated with a likelihood function\footnote{
In probability theory, if $X$ is a discrete random variable with probability mass function $p(\cdot;\theta)$ that depends on the parameter $\theta$, the likelihood function given outcome $x$ of $X$ is the function $\theta \mapsto L(\theta;x) := p(x;\theta)$.}, which will play a central role in this paper. Specifically, the likelihood function associated with the scheme $\mathcal A$, where $\mathcal A =$ AF, AB (which stand for ancilla-free and ancilla-based respectively), and tunable parameters $\vec x$ is defined to be the likelihood of 
the random variable 
$\theta$ (which encodes information about the expectation value $\bra A P \ket A$) given the measurement outcome $d$. In other words, this likelihood function, denoted $\mathcal L^\mathcal A(\theta;d,\vec x) = \Pr(d|\theta;\vec x)$, is the probability of obtaining the outcome $d$ in scheme $\mathcal A$ given the unknown parameter $\theta$ and tunable parameters $\vec x$. Since $d$ takes only the values $0$ and $1$, the likelihood function can be written as
\begin{align}
    \mathcal L^\mathcal A(\theta;d,\vec x) =\frac 12\left[
    1+(-1)^d \Lambda^\mathcal A(\theta;\vec x)
    \right]
    \label{eq:likelihoodInTermsOfBiasSchemes}
\end{align}
for some function $\Lambda^\mathcal A(\theta;\vec x)$. We shall call the function $\Lambda^\mathcal A$ the \textit{bias} associated with scheme $\mathcal A$.

In the next proposition, we derive expressions for the biases $\Lambda^\mathcal A(\theta;\vec x)$, which will be useful for the rest of this paper. For this purpose, we first extend the definition in Eq.~\eqref{eq:qalphabeta} slightly so that the second argument of $Q$ can be a vector of arbitrary nonzero length. Let $\alpha \in \mathbb Z^+$ be a positive integer. For $\theta \in \mathbb R$ and an arbitrary vector $\vec z = (z_1,\ldots, z_\alpha) \in \mathbb R^\alpha$, define
\begin{align}
    Q(\theta;\vec z) := R(z_\alpha)\ldots V(z_4) U(\theta;z_3) V(z_2) U(\theta;z_1),
    \label{eq:Q_as_a_product_of_Us_and_Vs}
\end{align}
where
\begin{align}
    R(\ \boldsymbol\cdot\ ) = 
    \begin{cases}
    U(\theta;\ \boldsymbol\cdot\ ), & \alpha \mbox{ odd},\\
    V(\ \boldsymbol\cdot\ ), & \alpha \mbox{ even}.
    \end{cases}
\end{align}
Next, define
\begin{align}
    Q_{00}[\vec z](\theta):=\bra {\bar 0} Q(\theta;\vec z) \ket {\bar 0}
    \label{eq:defQ00}
\end{align}
to be the expectation value of $Q(\theta;\vec z)$ in the state $\ket {\bar 0}$. These allow us to state the following proposition.
\begin{proposition}
Let $\theta \in \mathbb R$ and $\vec x \in \mathbb R^{2L}$. The ancilla-free and ancilla-based biases are given by
\begin{align}
    \Lambda^{\AF}(\theta;\vec x) &= \i Q_{00}\left[\vec x,\tfrac \pi 2, -\vec x^R\right](\theta)
\label{eq:delta_in_terms_of_Q00_AF}
\\
    \Lambda^{\AB}(\theta;\vec x) &= \operatorname{Re} Q_{00}[\vec x](\theta),
\label{eq:delta_in_terms_of_Q00_AB}
\end{align}
where $\vec x^R$ denotes the reverse of $\vec x$.
\label{prop:expressionsForBiases}
\end{proposition}
\begin{proof}
In the ancilla-free case, the state just before the measurement is $Q(\theta;\vec x)\ket {\bar 0}$, whose density operator can be written as
\begin{align}
    \varrho = Q(\theta;\vec x)\kbra{\bar 0}{\bar 0} Q(\theta;\vec x)^\dag  = Q(\theta;\vec x)\left(\frac{\bar I+\bar Z}{2}\right) Q(\theta;\vec x)^\dag = \frac 12(1+Q(\theta;\vec x)\bar ZQ(\theta;\vec x)^\dag).
\end{align}
Hence, the probability of obtaining the outcome $d\in \{0,1\}$ when $\varrho$ is measured is
\begin{align}
    \mathcal L^{\AF}(\theta;d,\vec x) &= \tr\left[ \frac 12\left(1+(-1)^d P(\theta)\right) \varrho \right] \nonumber\\
    &= \frac 12\left[1+(-1)^d \tr(P(\theta) \varrho) \right].
    \label{eq:problik}
\end{align}
where we have used $P=P(\theta)$ defined in Eq.~\eqref{eq:Pmatrix}.

Comparing this with Eq.~\eqref{eq:likelihoodInTermsOfBiasSchemes}, we get the following expression for the ancilla-free bias
\begin{align}
    \Lambda^\AF(\theta;\vec x) &= \tr(P(\theta) \varrho) \nonumber\\
    &=
    \tr\left[P(\theta)\frac 12\left(I+Q(\theta;\vec x)\bar{Z}Q(\theta;\vec x)^\dag\right)\right]
    \nonumber\\
    &= \frac 12 \tr(Q(\theta;\vec x)^\dag P(\theta) Q(\theta;\vec x) \bar{Z}) \nonumber\\
    &=
    \frac 12 \tr\left[Q(\theta;\vec x)^\dag P(\theta) Q(\theta;\vec x) (2 \kbra {\bar 0}{\bar 0} -I)\right] \nonumber
    \\
    &= \bra{\bar 0} Q(\theta;\vec x)^{\dagger}P(\theta)Q(\theta;\vec x) \ket{\bar 0}
    \nonumber\\
    &= 
\bra{\bar 0}
U(\theta,-x_1)V(-x_2)\ldots
U(\theta,-x_{2L-1})V(-x_{2L}) \ \i U(\theta,\pi/2)
\nonumber\\
& \qquad \times V(x_{2L})U(\theta;x_{2L-1})\ldots
V(x_2) U(\theta;x_1)
\ket{\bar 0} \nonumber\\
&= \i Q_{00}\left[x_1,x_2,\ldots, x_{2L-1}, x_{2L},\tfrac \pi 2, -x_{2L}, -x_{2L-1},\ldots, -x_2, -x_1\right](\theta)
\nonumber\\
&= \i Q_{00}\left[\vec x,\tfrac \pi 2, -\vec x^R\right](\theta),
    \label{eq:delta_af}
\end{align}
where the third line follows from the fact that the matrix \eqref{eq:Pmatrix} representing $P(\theta)$ is traceless, the fifth line follows from 
\begin{align}
\tr(Q(\theta;\vec x)^\dag P(\theta) Q(\theta;\vec x)) &= \tr(Q(\theta;\vec x)Q(\theta;\vec x)^\dag P(\theta)) = \tr(P(\theta)) = 0
\end{align}
and the sixth line follows from the fact that $P(\theta) = \i U(\theta, \pi/2)$.

In the ancilla-based case, the state just before measurement is
\begin{align}
    \varsigma &= (H\otimes I) \textup{c-}Q(\theta;\vec x)(H\otimes I) \ket 0 \ket{\bar 0} \nonumber\\
    &= \frac 12 \ket 0 (I+Q(\theta;\vec x))\ket{\bar 0}
    +
    \frac 12 \ket 1 (I-Q(\theta;\vec x))\ket{\bar 0},
\end{align}
where the second line follows from a straightforward calculation. Hence, the probability of obtaining the outcome $d \in \{0,1\}$ when $\varsigma$ is measured is
\begin{align}
    \mathcal L^{\AB}(\theta;d,\vec x) &= \left\Vert \frac 12 \left(I+(-1)^d Q(\theta;\vec x)\right) \ket {\bar 0}\right\Vert^2
    \nonumber\\
    &= \frac 12\left[1+(-1)^d
    \operatorname{Re} \bra{\bar 0} Q(\theta;\vec x)\ket{\bar 0}) \right].
\end{align}

Comparing this with Eq.~\eqref{eq:likelihoodInTermsOfBiasSchemes}, we obtain
\begin{align}
    \Lambda^{\AB}(\theta;\vec x) = \operatorname{Re} \bra{\bar 0} Q(\theta;\vec x)\ket{\bar 0} 
    =
    \operatorname{Re} Q_{00}[\vec x](\theta).
\end{align}
\end{proof}

In Appendix \ref{sec:seriesExpansions}, we develop analytical tools to study the biases $\Lambda^\mathcal A$, which will be useful for understanding their properties.



\section{The expected posterior variance: a Bayesian perspective}
\label{sec:bayesianinference}

In this section, we develop tools for understanding the expected posterior variance, which will be useful in Section \ref{sec:optimization}
when we apply them to the quantum-generated likelihood functions that we introduced in Section \ref{sec:quantumGeneratedLikelihoodFunctions}. The treatment in this section, which can be read independently of Section \ref{sec:quantumGeneratedLikelihoodFunctions}, is fairly general and applicable to a broad range of likelihood functions.

\subsection{General form}
\label{subsec:expvar}

Suppose that we have a (continuous) prior distribution $p(\theta) = p_{\uptheta}(\theta)$ that reflects our current state of knowledge about the value of some parameter $\theta \in \Theta$, where $\Theta \subseteq \mathbb R$. To update our knowledge about $\theta$, we may design an experiment that yields measurement outcomes $d \in \Omega$, where the set of possible measurement outcomes $\Omega \subseteq \mathbb N$ is assumed to be a discrete set. The probability of obtaining an outcome $d$ given the unknown parameter $\theta$ is captured by the (continuous) likelihood function $\mathcal{L}(\theta ; d)=\mathbb{P}_{\mathrm d| \uptheta}(d | \theta)$. Our knowledge of $\theta$ after performing Bayesian updating is described by the (continuous) posterior  distribution
\begin{align}
\label{eq:posteriorDistribution}
p_{\uptheta | \mathrm d}(\theta | d) = \frac{\mathcal L(\theta ; d) p(\theta)}{\mathbb{P}_{\mathrm d}(d)}
\end{align}
where
\begin{align}
    \mathbb{P}_{\mathrm d}(d) = \int \d \theta \  \mathcal{L}(\theta ; d)p(\theta)
    \label{eq:marginalDistrd}
\end{align}
is the marginal distribution of $d$. For Eq.~\eqref{eq:posteriorDistribution} to be well-defined, the outcome $d$ must have a nonzero probability of occurring, i.e. $d \in \Omega' := \{d \in \Omega: \mathbb P_{\mathrm d}(d) \neq 0 \}$.

Note that different experiments could give rise to different likelihood functions, and hence, different amounts of information gain about $\theta$. Our goal is to engineer likelihood functions that allow us to minimize the \textit{expected posterior variance}, defined as
\begin{align} \label{eq:expvar}
    \mathbb{E}_{\mathrm d} \Var_{\uptheta|\mathrm d}(\theta|d) = \sum_{d\in \Omega'} \mathbb P_{\mathrm d}(d) \Var_{\uptheta|\mathrm d}(\theta |d).
\end{align}
In the above formula, $\Var_{\uptheta | \mathrm d}(\theta|d)$ is the variance of the posterior distribution upon obtaining outcome $d$, i.e.
\begin{align}
    \Var_{\uptheta | \mathrm d}(\theta|d) =
    \mathbb{E}_{\uptheta | \mathrm d}[\theta^2 | d] - \left(
    \mathbb{E}_{\uptheta | \mathrm d}[\theta | d]
    \right)^2,
\end{align}
where
\begin{align}
    \mathbb{E}_{\uptheta | \mathrm d}[\theta | d] &=\int \d \theta \  \theta \  p_{\uptheta | \mathrm d}(\theta | d) \\
    \label{eq:secondMomentPosterior}
    \mathbb{E}_{\uptheta | \mathrm d}[\theta^2 | d] &=\int \d \theta \  \theta^2 \  p_{\uptheta | \mathrm d}(\theta | d)
\end{align}
are the first and second moments of $p_{\uptheta | \mathrm d}(\boldsymbol{\cdot} | d)$ respectively.

The expression for the expected posterior variance as given by Eqs.~\eqref{eq:expvar}--\eqref{eq:secondMomentPosterior} involves the computation of several integrals, which may be computationally intensive in general. Hence, our goal in this section is to provide various simplifications of the formula \eqref{eq:expvar} as we specialize to specific experiments that are relevant to this work. To this end, our first step is to show that the expected posterior variance can be written in terms of the prior mean
\begin{align} \label{eq:priorMean}
    \mu = \int \d \theta \  \theta p(\theta)
\end{align}
and prior variance
\begin{align} \label{eq:priorVariance}
    \sigma^2 = \int \d \theta \  (\theta-\mu)^2 p(\theta)
\end{align}
as follows:
\begin{thm}
The expected posterior variance is given by
\begin{align}
    \mathbb{E}_{\mathrm d} \operatorname{Var}_{\uptheta | \mathrm d}(\theta | d)=\sigma^{2}+\mu^{2}-\sum_{d \in \Omega^{\prime}} \frac{I_{1}(d)^{2}}{I_{0}(d)}
    \label{eq:expPostVarGeneral}
\end{align}
where
\begin{align}
    I_{k}(d)=\int \d \theta \ \theta^{k} \mathcal L(\theta ; d) p(\theta)
    \label{eq:I_kdef}
\end{align}
is the $k$th moment of the function $\mathcal L(\, \boldsymbol{\cdot} \, ;d)p(\boldsymbol{\cdot})$.
\end{thm}

\begin{remark*}
Alternatively, by using Eq.~\eqref{eq:posteriorDistribution}, the function $I_k$ defined in Eq.~\eqref{eq:I_kdef} can be written as
\begin{align}
    I_k(d) = \mathbb P_{\mathrm d}(d) \langle \theta^k \rangle_d
\end{align}
where
\begin{align}
    \langle \theta^k \rangle_d = \int \d \theta \ \theta^{k} p_{\uptheta | \mathrm d}(\theta | d)
\end{align}
is the $k$th moment of the probability distribution $p_{\uptheta | \mathrm d}( \cdot | d)$. In particular, $I_0(d) = \mathbb{P}_{\mathrm d}(d)$ is the marginal probability of $d$, and $I_1/\mathbb{P}_{\mathrm d}(d) = \langle \theta \rangle_d$ is the mean of $\theta$ given $\mathrm d = d$.
\end{remark*}

\begin{proof}
\begin{align*}
\mathbb{E}_{\mathrm d} \operatorname{Var}_{\uptheta | \mathrm d}(\theta | d) &= \sum_{d \in \Omega'} \mathbb{P}_{
\mathrm d}(d) \operatorname{Var}_{\uptheta | \mathrm d}(\theta | d) \\
&=\sum_{d \in \Omega'} \mathbb{P}_{\mathrm d}(d)\left\{\int \d \theta  \ \theta^{2} p_{\uptheta | \mathrm d}(\theta | d)-\left[\int \d \theta \ \theta p_{\uptheta | \mathrm d}(\theta | d)\right]^{2}\right\} \\
&=\sum_{d \in \Omega'} \mathbb{P}_{\mathrm d}(d) \int \d \theta \ \theta^{2} \frac{\mathcal L(\theta; d)
p(\theta)}{\mathbb{P}_{\mathrm d}(d)} - 
\sum_{d \in \Omega' } \mathbb{P}_{\mathrm d}(d)
\left(\int \d \theta \ \theta \frac{\mathcal L(\theta; d) p(\theta)}{\mathbb P_{\mathrm d}(d)}\right)^2 \\
&=\sum_{d \in \Omega'} \int \d \theta \  \theta^{2} \mathcal L(\theta; d) p(\theta)-\sum_{d \in \Omega'} \frac{\left[\int \d \theta \ \theta \mathcal L(\theta ; d) p(\theta)\right]^{2}
}{\int \d \theta \ \mathcal L(\theta; d) p(\theta)} \\
&=\int \d \theta \  \theta^{2} \underbrace{\sum_{d \in \Omega'} \mathcal L(\theta; d)}_{=1} p(\theta)-\sum_{d \in \Omega'} \frac{I_{1}(d)^{2}}{I_{0}(d)} \\
& =\sigma^{2}+\mu^{2}-\sum_{d \in \Omega'} \frac{I_{1}(d)^{2}}{I_{0}(d)}.
\end{align*}
where we used the fact that 
$$\sigma^2 = \left(\int \d \theta \ \theta^2 p(\theta)\right) - 
\left(\int \d \theta \ \theta p(\theta)\right)^2 = \left(\int \d \theta \ \theta^2 p(\theta)\right) - 
\mu^2.$$
in the last line.

\end{proof}

Note that Eq.~\eqref{eq:expPostVarGeneral} can also be expressed in the following form:
\begin{align}
\mathbb{E}_{\mathrm d} \operatorname{Var}_{\uptheta | \mathrm d}(\theta | d) = \sigma^2(1-\sigma^2 \mathcal V)
\label{eq:defVfactor}
\end{align}
where
\begin{align}
\mathcal V = \frac 1{\sigma^4}\left[ \sum_{d\in \Omega'} \frac{I_1(d)^2}{I_0(d)} - \mu^2 \right].
\label{eq:defVfactor2}
\end{align}

We shall call the symbol $\mathcal V$ defined by Eq.~\eqref{eq:defVfactor2} the \textit{variance reduction factor}. The motivation for introducing $\mathcal V$ is that working with $\mathcal V$ turns out to be more convenient than working directly with the expected posterior variance when we specialize to specific domains $\Omega$, prior distributions $p(\theta)$ and likelihood functions $\mathcal L(\theta ; d)$. We refer the reader to Figure \ref{fig:vFactor} for a summary of the specializations that we will consider in Sections \ref{subsec:twoOutcome}--\ref{sec:limBehavior} as well as the results we obtained.

\tikzstyle{block} = [rectangle, rounded corners, minimum width=3.5cm, minimum height=1.5cm,text centered, draw=red!55,fill=red!5,inner sep=0.2cm]
\tikzstyle{blueblock} = [rectangle, rounded corners, minimum width=3.5cm, minimum height=1.5cm,text centered, draw=blue!55,fill=blue!5,inner sep=0.2cm]
\tikzstyle{tightblock} = [rectangle, text centered, draw=blue!55,fill=blue!5,inner sep=0.09cm]
\tikzstyle{redblock} = [rectangle, rounded corners, minimum width=3.5cm, minimum height=1.5cm,text centered, draw=orange!55,fill=orange!5,inner sep=0.2cm]
\tikzstyle{arrow} = [->,>=stealth]

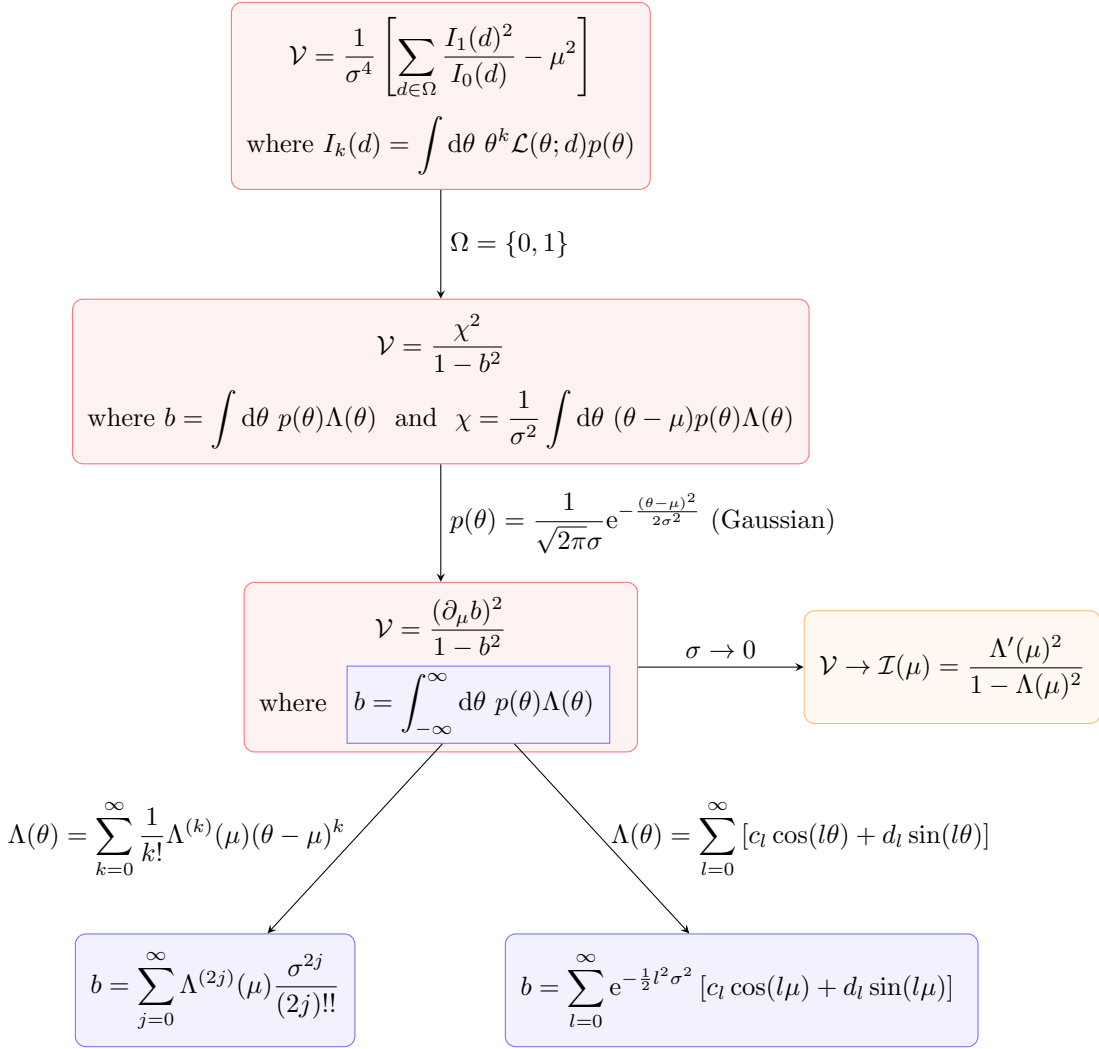
\begin{figure}[!ht]
    \centering
    \begin{tikzpicture}
    [node distance = 3.8cm, auto]
   
    \node [block, align=center] (0) {$\displaystyle{
    \mathcal V = \frac 1{\sigma^4}\left[ \sum_{d\in \Omega} \frac{I_1(d)^2}{I_0(d)} - \mu^2 \right]}
    $  \\[0.18cm] 
   where $\displaystyle{I_{k}(d)=\int \d \theta \ \theta^{k} \mathcal L(\theta ; d) p(\theta)}$ };
   
    \node [block, align=center, below of=0] (1) {$\displaystyle{\mathcal V = \frac{\chi^2}{1-b^2}}$  \\[0.18cm] 
   where $\displaystyle{b = \int \d \theta \ p(\theta) \Lambda(\theta)}$ \ and \  $\displaystyle{ \chi = \frac 1{\sigma^2} \int \d \theta \ (\theta-\mu) p(\theta) \Lambda(\theta)}$
   };

    \node [block, align=center, below of=1] (2) {$\displaystyle{\mathcal V = \frac{(\partial_\mu b)^2}{1-b^2}}$  \\[0.18cm] 
   where $\phantom{\displaystyle{b = \int_{-\infty}^\infty \d \theta \ p(\theta) \Lambda(\theta)}}$ \hspace{0.5cm}
   };

    \node [tightblock, align=center, below of=1, yshift = -0.5cm, xshift=0.5cm] (2p) { 
   $\displaystyle{ b = \int_{-\infty}^\infty \d \theta \ p(\theta) \Lambda(\theta)}$
   };
   
   \node [redblock, align=center, right of=2, xshift=3cm] (2R) {$\displaystyle{\mathcal{V}\rightarrow\mathcal I(\mu) = 
    \frac{\Lambda'(\mu)^2}{1-\Lambda(\mu)^2}}$};
   
    \node [blueblock, align=center, below of=2p, xshift=-3.5cm] (3L) {$\displaystyle{b=\sum_{j=0}^{\infty} \Lambda^{(2 j)}(\mu) \frac{\sigma^{2 j}}{(2 j) !!}}$};   

    \node [blueblock, align=center, below of=2p, xshift=3.5cm] (3R) {$\displaystyle{b=\sum_{l=0}^{\infty} \e^{-\frac{1}{2} l^{2} \sigma^{2}}\left[c_{l} \cos (l \mu)+d_{l} \sin (l \mu)\right]}$ 
    };  

   \draw [arrow] (0) -- node[anchor=west] {$\Omega = \{0,1\}$} (1);

   \draw [arrow] (1) -- node[anchor=west] {$\displaystyle{p(\theta) = \frac{1}{\sqrt{2 \pi} \sigma} \e^{-\frac{(\theta-\mu)^{2}}{2 \sigma^{2}}}}$ (Gaussian)} (2);

   \draw [arrow] (2p) -- node[anchor=east] {$\displaystyle{\Lambda(\theta) = \sum_{k=0}^\infty \frac 1{k!} \Lambda^{(k)}(\mu) (\theta-\mu)^k}$} (3L);

   \draw [arrow] (2p) -- node[anchor=west] {$\displaystyle{    \Lambda(\theta) = \sum_{l=0}^\infty \left[c_l \cos(l\theta) +d_l \sin(l\theta)\right]}$} (3R);

   \draw [arrow] (2) -- node[anchor=south] {$\sigma \rightarrow 0$} (2R);
   
\end{tikzpicture}
\caption{Flowchart showing expressions for the variance reduction factor $\mathcal V$ under various assumptions.
At the top of the flowchart is the most general expression for $\mathcal V$. As we specialize to specific $\Omega$, $p(\theta)$ and $\Lambda(\theta)$ by moving down the flowchart, the expression for $\mathcal V$ simplifies. Here, (i) $\Omega$ represents the set of measurement outcome values, (ii) $p(\theta)$ is the prior distribution, (iii) $\mathcal L(\theta;d) = \mathbb P_{\mathrm{d}|\uptheta}(d|\theta)$ is the likelihood function, (iv) $\mu= \int \d \theta \  \theta p(\theta)$ is the prior mean, and (v) $\sigma^2 = \int \d \theta \  (\theta-\mu)^2 p(\theta)$ is the prior variance. Also, $b$ denotes the expected bias and $\chi$ denotes the chi function. If $\Omega = \{0,1\}$, let $\Lambda(\theta) = 2\mathcal L(\theta;0)-1$ be the bias. In this flowchart, we assume that $b\neq 1$, $|\Lambda(\mu)| \neq 1$, and $I_0(d)\neq 0$ for all $d\in \Omega$.
}
\label{fig:vFactor}
\end{figure}

\subsection{Two-outcome likelihood functions}
\label{subsec:twoOutcome}

In the rest of this paper, we specialize to the case when $\Omega = \{0,1\}$. In this case, we show that the expression for the expected posterior variance simplifies as follows.

\begin{thm}
Assuming that $\Omega = \{0,1\}$, the expected posterior variance may be written as
\begin{align}
\label{eq:expVar_in_terms_of_Ik}
    \mathbb{E}_{\mathrm d} \operatorname{Var}_{\uptheta | \mathrm d}(\theta | d)=
    \begin{cases}
        \sigma^2 &  I_0 \in \{0,1\} \\
        \sigma^2 - \frac{(I_1-\mu I_0)^2}{I_0(1-I_0)} & I_0 \notin \{0,1\} 
    \end{cases}
\end{align}
where for $k \in \{0, 1\}$,
\begin{align}
\label{eq:Ik_integral}
    I_{k}:= I_{k}(0)=\int \d \theta \ \theta^{k} \mathcal L(\theta ; 0) p(\theta)
\end{align}
is the $k$th moment of the function $\mathcal L(\, \boldsymbol{\cdot} \, ;0)p(\boldsymbol{\cdot})$.
\end{thm}

\begin{proof}
Consider
\begin{align}
\sum_{d \in \{0,1\} } I_{k}(d) &= \int \d \theta \ \theta^{k} \underbrace{\sum_{d \in \{0,1\} } \mathcal L(\theta ; d)}_{=1} p(\theta) \\
&=\int \d \theta \ \theta^{k} p(\theta).
\end{align}
Equivalently,
\begin{align}
    I_0(0)+I_0(1) &= 1, \nonumber\\
    I_1(0)+I_1(1) &= \mu.
    \label{eq:I10I11mu}
\end{align}

We first consider the case when $I_0=I_0(0)=1$. In this case, 
\begin{align}
    0= I_0(1) = \int \d \theta \ \mathcal L(\theta;1)p(\theta).
\end{align} Since $\mathcal L(\theta;1)$ and $p(\theta)$ are both continuous and nonnegative, it follows that
\begin{align}
    \mathcal L(\theta;1) p(\theta) = 0.
    \label{eq:Lp0}
\end{align}
This implies that 
\begin{align}
    I_1(0) &= \int \d \theta \ \theta \mathcal L(\theta;0) p(\theta)\nonumber\\
    &= \int \d \theta \ \theta \left[ 1-\mathcal L(\theta;1)\right] p(\theta)
    \nonumber\\
&= \mu - \int \d \theta \ \theta \underbrace{\mathcal L(\theta;1) p(\theta)}_{=0, \ \mbox{\scriptsize by Eq.~\eqref{eq:Lp0}}
} \nonumber\\
    &= \mu.
\end{align} Since $\Omega' = \{0\}$, it follows that
\begin{align}
    \mathbb{E}_{\mathrm d} \operatorname{Var}_{\uptheta | \mathrm d}(\theta | d) &= \sigma^2+\mu^2 - \frac{I_1(0)^2}{I_0(0)} \nonumber\\
    &=\sigma^2 +\mu^2 - \mu^2/1 \nonumber\\
    &=\sigma^2.
\end{align}
By symmetry, the case when $I_0=0$ is similar, and also gives $\mathbb{E}_{\mathrm d} \operatorname{Var}_{\uptheta | \mathrm d}(\theta | d)=\sigma^2.$

Finally, we consider the case when $I_0\notin\{0,1\}$. In this case,
\begin{align}
    \mathbb{E}_{\mathrm d} \operatorname{Var}_{\uptheta | \mathrm d}(\theta | d) &= \sigma^2+\mu^2 - \sum_{d\in \{0,1\}} \frac{I_1(d)^2}{I_0(d)}. 
\end{align}

The last term of the above expression can be simplified as
\begin{align}
\sum_{d\in \{0,1\}} \frac{I_1(d)^2}{I_0(d)} &=
\frac{I_{1}(0)^{2}}{I_{0}(0)}+\frac{I_{1}(1)^{2}}{I_{0}(1)}\nonumber\\
&=\frac{I_{1}^{2}}{I_{0}}+\frac{\left(\mu-I_{1}\right)^{2}}{1-I_{0}}, \qquad \mbox{by Eq.~\eqref{eq:I10I11mu}}\nonumber\\
&=\frac{I_{1}^{2}\left(1-I_{0}\right)+I_{0}\left(\mu^{2}-2 \mu I_{1}+I_{1}^{2}\right)}{I_{0}\left(1-I_{0}\right)}\nonumber\\
&=\frac{I_{1}^{2}-I_{0} I_{1}^{2}+\mu^{2} I_{0}-2 \mu I_{0} I_{1}+I_{0} I_{1}^{2}}{I_{0}\left(1-I_{0}\right)}\nonumber\\
&=\frac{\mu^{2}-2 I_{1} \mu+I_{1}^{2} / I_{0}}{1-I_{0}}.
\end{align}

Hence,
\begin{align}
\mathbb{E}_{d} \operatorname{Var}_{\theta | d}(\theta | d) &=\sigma^{2}+\mu^{2}-\frac{\mu^{2}-2 I_{1} \mu+I_{1}^{2} / I_{0}}{1-I_{0}} \nonumber\\
&=\sigma^{2}-\frac{- \mu^{2}+\mu^{2} I_{0}+\mu^{2}-2 I_{1} \mu+I_{1}^{2} / I_{0}}{1-I_{0}} \nonumber\\
&=\sigma^{2}-\frac{\mu^{2} I_{0}^{2}-2 I_{1} I_{0} \mu+I_{1}^{2}}{I_{0}\left(1-I_{0}\right)} \nonumber\\
&=\sigma^{2}-\frac{\left(I_{1}-\mu I_{0}\right)^{2}}{I_{0}\left(1-I_{0}\right)}.
\end{align}
\end{proof}

It turns out that Eq.~\eqref{eq:expVar_in_terms_of_Ik} can be simplified if we expressed the expected posterior variance in terms of the bias, which we now define. For a two-outcome likelihood function $\mathcal L(\theta;d)$, where $d\in \Omega = \{0,1\}$, the \textit{bias} is given by
\begin{align}
\label{eq:deltadef}
    \Lambda(\theta) = 2 \mathcal L(\theta;0)-1.
\end{align}
This gives the following expression for the likelihood function in terms of the bias
\begin{align}
    \mathcal L(\theta;d) = \frac 12\left[1+(-1)^d \Lambda(\theta)\right].
    \label{eq:likelihoodInTermsOfBias}
\end{align}

Let us now define
\begin{align}
    b &= \int \d \theta \ p(\theta) \Lambda(\theta),
    \label{eq:biasDef}
    \\
    \chi &= \frac 1{\sigma^2} \int \d \theta \ (\theta-\mu) p(\theta) \Lambda(\theta).
    \label{eq:chiDef}
\end{align}
We shall call $b$ the \textit{expected bias} and $\chi$ the \textit{chi function}. It is straightforward to check that $b$ and $\chi$ can be expressed in terms of $I_0$ and $I_1$, which were defined in
Eq.~\eqref{eq:Ik_integral}, as follows:
\begin{align}
    b &= 2 I_0 -1, \\
    \chi &= \frac 2{\sigma^2}(I_1-\mu I_0).
\end{align}
Substituting these expressions into Eq.~\eqref{eq:expVar_in_terms_of_Ik} gives
\begin{align}
    \mathbb{E}_{\mathrm d} \operatorname{Var}_{\uptheta | \mathrm d}(\theta | d)= \sigma^2\left(1-\sigma^2 \mathcal V\right),
    \label{eq:expVar_in_terms_of_V}
\end{align}
with the variance reduction factor \eqref{eq:defVfactor} given by 
\begin{align}
\label{eq:varReductionFactor}
    \mathcal V = \begin{cases}
        \frac{\chi^2}{1-b^2}, & |b| < 1, \\[4pt]
        0, & |b|=1.
    \end{cases}
\end{align}

\subsection{Two-outcome likelihood functions with a Gaussian prior}
\label{subsec:two-outcome-lfs-gauss}

\subsubsection{Variance reduction factor with a Gaussian prior}
\label{subsec:two-outcome-lfs-gauss-prelim}
In this section, we fix the prior distribution to be the Gaussian distribution
\begin{align}
\label{eq:gaussianDistribution}
    p(\theta) = p(\theta ; \mu, \sigma) := \frac{1}{\sqrt{2 \pi} \sigma} \e^{-\frac{(\theta-\mu)^{2}}{2 \sigma^{2}}}
\end{align}
with prior mean $\mu\in \mathbb R$ and prior variance $\sigma^2 \in \mathbb R^+$. As before, we denote the bias by $\Lambda(\theta)$. The expected bias and the chi function are given by
\begin{align}
\label{eq:expectedBiasGaussian}
    b = b(\mu,\sigma) &:= \int_{-\infty}^\infty \d \theta \ p(\theta;\mu,\sigma) \Lambda(\theta), \\
    \label{eq:chiFunctionGaussian}
    \chi = \chi(\mu,\sigma) &:= \frac 1{\sigma^2} \int_{-\infty}^\infty \d \theta \ (\theta-\mu) p(\theta;\mu,\sigma) \Lambda(\theta).
\end{align}

The Gaussian prior has the following nice property\footnote{
In fact, the Gaussian prior is the only (continuous) prior with this property. To see this, note that equating $\chi(\mu,\sigma) = \partial_\mu b(\mu,\sigma)$ gives the differential equation
$\tfrac 1{\sigma^2} (\theta-\mu)p(\theta;\mu,\sigma) = \partial_\mu p(\theta;\mu,\sigma)$.
The unique solution to this equation with the normalization boundary condition 
$\int_{-\infty}^\infty \d \theta \ p(\theta;\mu,\sigma) = 1$ is the Gaussian distribution $p(\theta;\mu,\sigma)$ as given by Eq.~\eqref{eq:gaussianDistribution}.
}: differentiating the expected bias with respect to the prior mean gives the chi function, i.e.
\begin{align}
\label{eq:chiDerivative_b}
    \chi(\mu,\sigma) = \frac{\partial}{\partial \mu} b(\mu,\sigma).
\end{align}

Next, note that $b(\mu,\sigma)<1$ if and only if the bias $\Lambda(\theta)$ is neither the constant 1 function nor the constant -1 function (write this as $\Lambda \notin \{-1,1\}$). This, together with Eq.~\eqref{eq:chiDerivative_b}, gives the following expression for the variance reduction factor defined in Eq.~\eqref{eq:varReductionFactor}:
\begin{align}
    \mathcal V = \mathcal V(\mu,\sigma):= \frac{\partial_\mu b(\mu,\sigma)^2}{1-b(\mu,\sigma)^2} \mathds 1_{\Lambda \notin \{\pm 1\}}
    \label{eq:gaussreductionfactor}
\end{align}
where $\mathds 1_{\Lambda \notin \{\pm 1\}}$ denotes the indicator function which is equal to 1 when $\Lambda \notin \{\pm 1\}$ and 0 otherwise. 


So far, we have assumed that the bias $\Lambda(\theta)$ is arbitrary. In the rest of this section, we will derive series expansions for $b(\mu,\sigma)$ and $\chi(\mu,\sigma)$ when $\Lambda(\theta)$ is expanded as a (Taylor or Fourier) series.

\subsubsection{Bias as a Taylor series}
\label{sec:biasTaylorSeries}

We will now derive series expansions for $b(\mu,\sigma)$ and $\chi(\mu,\sigma)$ when $\Lambda(\theta)$ is written as a Taylor series at $\mu$, i.e.
\begin{align}
\label{eq:biasAsTaylorSeries}
    \Lambda(\theta) = \sum_{k=0}^\infty \Lambda_k (\theta-\mu)^k, \quad\mbox{where }\Lambda_k = \frac 1{k!}\Lambda^{(k)}(\mu).
\end{align}

Substituting this into Eq.~\eqref{eq:expectedBiasGaussian} gives
\begin{align} \label{eq:evaluationOfb}
b\left(\mu, \sigma\right) &=\int_{-\infty}^{\infty} \d \theta \ p(\theta ; \mu, \sigma) \Lambda(\theta) \nonumber\\
&=\int_{-\infty}^{\infty} \d \theta \ \frac{1}{\sqrt{2 \pi} \sigma} \e^{-\frac{(\theta-\mu)^{2}}{2 \sigma^{2}}} \sum_{k=0}^{\infty} \Lambda_{k}(\theta-\mu)^{k} \nonumber\\
&=\sum_{k=0}^{\infty} \Lambda_{k} \int_{-\infty}^{\infty} \d \theta \  \frac{1}{\sqrt{2 \pi} \sigma} \e^{-\frac{(\theta-\mu)^{2}}{2 \sigma^{2}}}(\theta-\mu)^{k} \nonumber\\
&=\sum_{k=0}^{\infty} \Lambda_{k}(k-1) ! ! \sigma^{k} \mathds{1}_{k \in 2 \mathbb Z}
\end{align}
where in the last line we used the following identity: for $\sigma>0$ and $n\in \mathbb N$,
\begin{align}
\int_{-\infty}^{\infty} \d \theta \ \frac{1}{\sqrt{2 \pi} \sigma} \e^{-\frac{(\theta-\mu)^{2}}{2 \sigma^{2}}}(\theta-\mu)^{n}
= (n-1) ! ! \sigma^{n} \mathds{1}_{n \in 2 \mathbb Z},  
\end{align}
where 
\begin{align}
n!!=\prod_{k=0}^{\lceil n/2 \rceil-1}(n-2k) = n(n-2)(n-4)\ldots
\end{align}
(with the empty product equal to 1) denotes the double factorial of $n$. 

Substituting $\Lambda_k = \frac 1{k!} \Lambda^{(k)}(\mu)$ into Eq.~\eqref{eq:evaluationOfb} gives
\begin{align}
b(\mu, \sigma) &=\sum_{k=0}^{\infty} \frac{1}{k !} \Lambda^{(k)}(\mu)(k-1) ! ! \sigma^{k} \mathds 1_{k \in 2 \mathbb Z} \nonumber\\
&=\sum_{k=0}^{\infty} \frac{1}{k !!} \Lambda^{(k)}(\mu) \sigma^{k} \mathds 1_{k \in 2 \mathbb Z}, \quad \text { where we used } \frac{(k-1)!!}{k!} = \frac 1{k!!} \nonumber\\
&=\sum_{j=0}^{\infty} \Lambda^{(2 j)}(\mu) \frac{\sigma^{2 j}}{(2 j) !!}.
\label{eq:bmusigma}
\end{align}

Differentiating this expression with respect to $\mu$ gives
\begin{align}
    \chi(\mu,\sigma) = \sum_{j=0}^{\infty} \Lambda^{(2 j+1)}(\mu) \frac{\sigma^{2 j}}{(2 j)!!}.
\label{eq:chimusigma}
\end{align}

Equivalently, the derivatives of the expected bias and the chi function are given as follows: \begin{align} \left.\frac{\partial^k}{\partial \sigma^k} b(\mu,\sigma)\right|_{\sigma=0} &= \Lambda^{(k)}(\mu) (k-1)!! \mathds 1_{k \in 2\mathbb Z}, \label{eq:Maclaurin_b} \\ \left.\frac{\partial^k}{\partial \sigma^k} \chi(\mu,\sigma)\right|_{\sigma=0} &= \Lambda^{(k+1)}(\mu) (k-1)!! \mathds 1_{k \in 2\mathbb Z}. \label{eq:Maclaurin_chi} \end{align}

\subsubsection{Bias as a trigonometric Fourier series}
\label{sec:biasTrigoSeries}

Next, we will derive series expansions for $b(\mu,\sigma)$ and $\chi(\mu,\sigma)$ when $\Lambda(\theta)$ is written as a trigonometric Fourier  series with period $2\pi$, i.e.
\begin{align}
\label{eq:biasAsCosineSeries}
    \Lambda(\theta) = \sum_{l=0}^\infty \left[c_l \cos(l\theta) +d_l \sin(l\theta)\right],
\end{align}
where
\begin{align}
c_{l}&=\frac{1}{\left(1+\delta_{l}\right) \pi} \int_{-\pi}^{\pi} \d \theta \ \Lambda(\theta) \cos (l \theta),  \\
d_{l}&=\frac{1}{\pi} \int_{-\pi}^{\pi} \d \theta \
\Lambda(\theta) \sin (l \theta) ,
\end{align}
where $\delta_l$ is the Kronecker delta which is equal to 1 if $l=0$ and 0 otherwise.

Substituting Eq.~\eqref{eq:biasAsCosineSeries} into Eq.~\eqref{eq:expectedBiasGaussian} gives
\begin{align}
b\left(\mu, \sigma\right) &=\int_{-\infty}^{\infty} \d \theta \ p(\theta ; \mu, \sigma) \Lambda(\theta) \nonumber\\
&=\int_{-\infty}^{\infty} \d \theta \  \frac{1}{\sqrt{2 \pi} \sigma} \e^{-\frac{(\theta-\mu)^{2}}{2 \sigma^{2}}} \sum_{l=0}^{\infty}\left[c_{l} \cos (l \theta)+d_{l} \sin (l \theta)\right] \nonumber\\
&=\sum_{l=0}^{\infty}\left[c_{l} \int_{-\infty}^\infty \d \theta \ \frac{1}{\sqrt{2 \pi} \sigma} \e^{-\frac{(\theta-\mu)^{2}}{2 \sigma^{2}}} \cos (l \theta)
+d_{l} \int_{-\infty}^{\infty} \d \theta \ \frac{1}{\sqrt{2 \pi} \sigma} \e^{-\frac{(\theta-\mu)^{2}}{2 \sigma^{2}}} \sin (l \theta)\right] \nonumber\\
&= 
\label{eq:bAsCosSeries}
\sum_{l=0}^{\infty} \e^{-\frac{1}{2} l^{2} \sigma^{2}}\left[c_{l} \cos (l \mu)+d_{l} \sin (l \mu)\right],
\end{align}
where in the last line we used the following identities (which follow from\cite{weisstein2004gaussian}): for $\sigma>0$ and $\mu,l \in \mathbb R$,
\begin{align}
\int_{-\infty}^{\infty} \d \theta \  \frac{1}{\sqrt{2 \pi} \sigma} \e^{-\frac{(\theta- \mu)^2}{2 \sigma^{2}}} \cos (l \theta) = \e^{-\frac{1}{2} l^{2} \sigma^{2}} \cos (l \mu), \\
\int_{-\infty}^{\infty} \d \theta \  \frac{1}{\sqrt{2 \pi} \sigma} \e^{-\frac{(\theta- \mu)^2}{2 \sigma^{2}}} \sin (l \theta) = \e^{-\frac{1}{2} l^{2} \sigma^{2}} \sin (l \mu).
\end{align}
Differentiating Eq.~\eqref{eq:bAsCosSeries} with respect to $\mu$ gives
\begin{align}
    \chi(\mu,\sigma) = 
    \sum_{l=1}^{\infty} l \e^{-\frac{1}{2} l^{2} \sigma^{2}} \left[d_{l} \cos (l\mu) -c_{l} \sin (l \mu) \right].
    \label{eq:chiAsCosSeries}
\end{align}

\subsection{Limiting behavior of the expected posterior variance for small prior variance
}
\label{sec:limitingBehavior}

\subsubsection{Connection to Fisher information}
\label{sec:limBehavior}

Throughout this section, we will assume that the prior distribution $p(\theta;\mu,\sigma)$ is Gaussian and given by Eq.~\eqref{eq:gaussianDistribution}, and that the bias $\Lambda(\theta)$ can be expressed as the Taylor series given by Eq.~\eqref{eq:biasAsTaylorSeries}. We will consider the limiting behavior of the expected posterior variance as $\sigma$ vanishes and show the relationship between this quantity and Fisher information.


The Fisher information is commonly used to capture the power of a likelihood function in the estimation process \cite{fisher1925theory}.
The Fisher information is defined as
\begin{align}
    \mathcal{I}(\theta)=\mathbb{E}_d \left(\frac{\partial}{\partial \theta}\log \mathcal{L}(\theta; d)\right)^2.
\end{align}
Larger Fisher information indicates that the data is expected to be more informative of the value of the unknown parameter $\theta$. We show that the Fisher information evaluated at the mean of the prior distribution is closely related to the variance reduction factor $\mathcal{V}$.

In the case of the two-outcome likelihood function, the Fisher information evaluates to
\begin{align}
    \mathcal{I}(\theta)&= \mathcal{L}(\theta;0)\left(\frac{\mathcal{L}'(\theta;0)}{\mathcal{L}(\theta;0)}\right)^2+\mathcal{L}(\theta;1)\left(\frac{\mathcal{L}'(\theta;1)}{\mathcal{L}(\theta;1)}\right)^2\nonumber \\
    &= \frac{\left(\mathcal{L}'(\theta;0)\right)^2}{\mathcal{L}(\theta;0)(1-\mathcal{L}(\theta;0))}.
\end{align}
In terms of the bias of the likelihood function \eqref{eq:deltadef}, we can express the Fisher information as
\begin{align}
    \mathcal{I}(\theta)&=\frac{\Lambda'(\theta)^2}{1-\Lambda(\theta)^2}.
    \label{eq:FisherInformation}
\end{align}

We now show that $\mathcal V$ can be approximated by the Fisher information when $\sigma$ is small. 
Expanding the expression of the expected bias given by Eq.~\eqref{eq:bmusigma} gives
\begin{align}
b(\mu,\sigma) &=    \Lambda(\mu)+\frac{1}{2} \Lambda^{(2)}(\mu) \sigma^{2}+\frac{1}{8} \Lambda^{(4)}(\mu) \sigma^{4}+O(\sigma^{6})
\\
&\rightarrow \Lambda(\mu) \quad \mbox{ as }\sigma \rightarrow 0.
\label{eq:bToDeltaAsSigmaToZero}
\end{align}

Expanding the expression of the chi function given by Eq.~\eqref{eq:bmusigma} gives
\begin{align}
\chi(\mu,\sigma) &=    \Lambda'(\mu)+\frac{1}{2} \Lambda^{(3)}(\mu) \sigma^{2}+\frac{1}{8} \Lambda^{(5)}(\mu) \sigma^{4}+O(\sigma^{6})
\\
&\rightarrow \Lambda'(\mu) \quad \mbox{ as }\sigma \rightarrow 0.
\label{eq:chiToDeltaAsSigmaToZero}
\end{align}
Hence, as long as\footnote{When $|\Lambda(\mu)| \neq 1$ does not hold, Eq.~\eqref{eq:V0Imu} may not either. In this case, L'H\^{o}pital's rule may be used to evaluate the limit 
$\lim_{\sigma \rightarrow 0} \mathcal V(\mu,\sigma)$. We study this case in detail in Section \ref{app:limiting_behavior}.
} $|\Lambda(\mu)| \neq 1$ (which implies that $\Lambda \notin \{\pm 1\}$), the variance reduction factor \eqref{eq:gaussreductionfactor} as $\sigma$ goes to zero is
equal to the Fisher information at $\theta=\mu$,
\begin{align}
    \lim_{\sigma \rightarrow 0} \mathcal V(\mu,\sigma) = 
    \frac{\Lambda'(\mu)^2}{1-\Lambda(\mu)^2}=\mathcal{I}(\mu).
    \label{eq:V0Imu}
\end{align}
Therefore, using Eqs.~\eqref{eq:expVar_in_terms_of_V}, the expected posterior variance when $\sigma$ is small is approximately linear in the Fisher information:
\begin{align}
    \mathbb{E}_{\mathrm d} \operatorname{Var}_{\uptheta | \mathrm d}(\theta | d) &\approx \sigma^2\left(1- \frac{\sigma^2\Lambda'(\mu)^2}{1-\Lambda(\mu)^2}\right)
    \nonumber\\
    &=
    \sigma^2\left(1-\sigma^2\mathcal{I}(\mu)\right),\quad \mbox{for } \sigma\approx 0.
    \label{eq:exppostvar_fi}
\end{align}

\subsubsection{Applying L'H\^{o}pital's rule
}
\label{app:limiting_behavior}

Take $\Omega =\{0,1\}$ and the prior distribution to be the Gaussian distribution $p(\theta;\mu,\sigma)$ given by Eq.~\eqref{eq:gaussianDistribution}. In Section \ref{sec:limBehavior}, we showed that if $|\Lambda(\mu)| \neq 1$, the variance reduction factor in the limit when $\sigma \rightarrow 0$ (by Eq.~\eqref{eq:gaussreductionfactor})
\begin{align}
    \lim_{\sigma\rightarrow 0}\mathcal V(\mu,\sigma)= \lim_{\sigma\rightarrow 0}\frac{\chi(\mu,\sigma)^2}{1-b(\mu,\sigma)^2} \mathds 1_{\Lambda \notin \{\pm 1\}}
\end{align}
is equal to the Fisher information at the prior mean $\mu$ (see Eq.~\eqref{eq:V0Imu}). In this section, we shall explore the case when $|\Lambda(\mu)|=1$. In this case, L'H\^{o}pital's rule may be used to evaluate the limit 
$\lim_{\sigma \rightarrow 0} \mathcal V(\mu,\sigma)$ as follows:
\begin{align}
    \lim_{\sigma\rightarrow 0}\mathcal V(\mu,\sigma)= \frac{\left.\frac{\partial^w}{\partial \sigma^w}\chi(\mu,\sigma)^2\right|_{\sigma=0}}{\left.\frac{\partial^w}{\partial \sigma^w}\left[1-b(\mu,\sigma)^2\right]\right|_{\sigma=0}} \mathds 1_{\Lambda \notin \{\pm 1\}}
    \label{eq:lhopitalrule}
\end{align}
where $w$ is the minimum $w'\in \mathbb N$ such that
\begin{align}
    \left.\frac{\partial^{w'}}{\partial \sigma^{w'}}\chi(\mu,\sigma)^2\right|_{\sigma=0} \neq 0 \quad\mbox{ or }\quad
    \left.\frac{\partial^{w'}}{\partial \sigma^{w'}}\left[1-b(\mu,\sigma)^2\right]\right|_{\sigma=0} \neq 0.
    \label{eq:notzeroor}
\end{align}
Note that $w$ could be calculated using Eq.~\eqref{eq:nDerivativeChi} and \eqref{eq:nDerivativeb} from  the following proposition.
\begin{proposition}
\begin{align}
    \left.\frac{\partial^{n}}{\partial \sigma^{n}}\chi(\mu,\sigma)^2\right|_{\sigma=0} &= 
    \mathds 1_{n\in 2\mathbb Z} \sum_{k=0}^n \frac{n!}{k!! (n-k)!!} \Lambda^{(n-k+1)}(\mu) \Lambda^{(k+1)}(\mu) \mathds 1_{k\in 2\mathbb Z}
    \label{eq:nDerivativeChi}
    \\
    \left.\frac{\partial^{n}}{\partial \sigma^{n}}\left[1-b(\mu,\sigma)^2\right]\right|_{\sigma=0} 
    &= \delta_n - \mathds 1_{n\in 2\mathbb Z} \sum_{k=0}^n \frac{n!}{k!! (n-k)!!} \Lambda^{(n-k)}(\mu) \Lambda^{(k)}(\mu) \mathds 1_{k\in 2\mathbb Z}.
    \label{eq:nDerivativeb}
\end{align}
\label{prop:derivatives_b_chi}
\end{proposition}
\begin{proof}
We first state a useful identity:
\begin{align}
    \frac{\partial^n}{\partial x^n} f(x)^2 = \sum_{k=0}^n \binom{n}{k} f^{(n-k)}(x) f^{(k)}(x),
    \label{eq:derivativeSquare}
\end{align}
where $f^{(l)}(x) = \frac{\partial^l}{\partial x^l}f(x)$ denotes the $l$th derivative with respect to $x$. 

To prove Eq.~\eqref{eq:nDerivativeChi}, we use Eqs.~\eqref{eq:Maclaurin_chi} and \eqref{eq:derivativeSquare} to obtain
\begin{align}
    \left.\frac{\partial^{n}}{\partial \sigma^{n}}\chi(\mu,\sigma)^2\right|_{\sigma=0} &=
    \sum_{k=0}^n \binom{n}{k} \left(\left.\frac{\partial^{n-k} }{\partial \sigma^{n-k}}\chi(\mu,\sigma) \right|_{\sigma=0}\right)\left(\left.\frac{\partial^{k}}{\partial \sigma^{k}}\chi(\mu,\sigma) \right|_{\sigma=0}\right) \nonumber\\
    &=\sum_{k=0}^n \binom{n}{k} \Lambda^{(n-k+1)}(\mu)(n-k-1)!! \mathds 1_{n-k \in 2\mathbb Z} \Lambda^{(k+1)}(\mu) (k-1)!! \mathds 1_{k\in 2\mathbb Z}
    \nonumber\\
    &= \mathds 1_{n\in 2\mathbb Z} \sum_{k=0}^n \frac{n!}{k!! (n-k)!!} \Lambda^{(n-k+1)}(\mu) \Lambda^{(k+1)}(\mu) \mathds 1_{k\in 2\mathbb Z}.
\end{align}
To prove Eq.~\eqref{eq:nDerivativeb}, we use Eqs.~\eqref{eq:Maclaurin_b} and \eqref{eq:derivativeSquare} to obtain
\begin{align}
    \left.\frac{\partial^{n}}{\partial \sigma^{n}}\left[1-b(\mu,\sigma)^2\right]\right|_{\sigma=0} &=
    \delta_{n} -
    \sum_{k=0}^n \binom{n}{k} \left(\left.\frac{\partial^{n-k} }{\partial \sigma^{n-k}}b(\mu,\sigma) \right|_{\sigma=0}\right)\left(\left.\frac{\partial^{k}}{\partial \sigma^{k}}b(\mu,\sigma) \right|_{\sigma=0}\right) \nonumber\\
    &=\delta_n - \sum_{k=0}^n \binom{n}{k} \Lambda^{(n-k)}(\mu)(n-k-1)!! \mathds 1_{n-k \in 2\mathbb Z} \Lambda^{(k)}(\mu) (k-1)!! \mathds 1_{k\in 2\mathbb Z}
    \nonumber\\
    &= \delta_n - \mathds 1_{n\in 2\mathbb Z} \sum_{k=0}^n \frac{n!}{k!! (n-k)!!} \Lambda^{(n-k)}(\mu) \Lambda^{(k)}(\mu) \mathds 1_{k\in 2\mathbb Z}.
\end{align}

\end{proof}

\renewcommand{\arraystretch}{1.5}
\begin{table}
\centering
\caption{Table of $\left.\frac{\partial^{n}}{\partial \sigma^{n}}\chi(\mu,\sigma)^2\right|_{\sigma=0}$ and
$\left.\frac{\partial^{n}}{\partial \sigma^{n}}\left[1-b(\mu,\sigma)^2\right]\right|_{\sigma=0} $ values for odd $n$ and $n=0,2,4,6$ calculated using Eqs.~\eqref{eq:nDerivativeChi} and \eqref{eq:nDerivativeb}. These expressions can be used to determine the minimum $w'$ for which Eq.~\eqref{eq:lhopitalrule}
holds, which in turn can be used to calculate Eq.~\eqref{eq:notzeroor}.    
}
\begin{tabular}{c|c|c}
$n$ \qquad &  $\displaystyle{
\left.\frac{\partial^{n}}{\partial \sigma^{n}}\chi(\mu,\sigma)^2\right|_{\sigma=0}
}
$
&
$\displaystyle{
\left.\frac{\partial^{n}}{\partial \sigma^{n}}\left[1-b(\mu,\sigma)^2\right]\right|_{\sigma=0}
}$
\\[0.3cm] \hline
odd   & 0                                       & 0                                        \\
0   & $\Lambda'^2$    & $1-\Lambda^2$             \\
2   & $2 \Lambda' \Lambda^{(3)}$    & $-2 \Lambda \Lambda^{(2)}$             \\ 
4   & $6 \Lambda^{(3)2} + 6 \Lambda' \Lambda^{(5)}$    & $-6 \Lambda^{(2)2} - 6 \Lambda \Lambda^{(4)}$             \\
6   & $90 \Lambda^{(3)} \Lambda^{(5)} + 30 \Lambda' \Lambda^{(7)}$    & $-90 \Lambda^{(2)} \Lambda^{(4)} - 30 \Lambda \Lambda^{(6)}$    
\end{tabular}
\label{tab:nDerivativechib}
\end{table}
\renewcommand{\arraystretch}{1}

Note that Proposition \ref{prop:derivatives_b_chi} implies that both $\left.\frac{\partial^{n}}{\partial \sigma^{n}}\chi(\mu,\sigma)^2\right|_{\sigma=0}$ and 
$\left.\frac{\partial^{n}}{\partial \sigma^{n}}\left[1-b(\mu,\sigma)^2\right]\right|_{\sigma=0}$ vanish when $n$ is odd, which implies that the statement \eqref{eq:notzeroor}
is necessarily false for these values of $n$. Hence, the integer $w$ in Eq.~\eqref{eq:lhopitalrule} must be even. A table of 
$\left.\frac{\partial^{n}}{\partial \sigma^{n}}\chi(\mu,\sigma)^2\right|_{\sigma=0}$ and 
$\left.\frac{\partial^{n}}{\partial \sigma^{n}}\left[1-b(\mu,\sigma)^2\right]\right|_{\sigma=0}$ values for some small values of $n$ (namely, $n=0,2,4,6$) is given in Table \ref{tab:nDerivativechib}.

As a consequence of Table \ref{tab:nDerivativechib} and Eqs.~\eqref{eq:lhopitalrule} and \eqref{eq:notzeroor}, we obtain, for example, the following limiting behaviors of $\mathcal V(\mu,\sigma)$ as $\sigma\rightarrow 0$:
\begin{itemize}
\item If $|\Lambda(\mu)|=1$, $\Lambda'(\mu)=0$, and $\Lambda''(\mu) \neq 0$, then
    \begin{align}
\lim_{\sigma\rightarrow 0}\mathcal V(\mu,\sigma)=0.
\label{eq:limitVasSigma0One}
\end{align}
\item If $|\Lambda(\mu)|=1$, $ \Lambda'(\mu)=\Lambda''(\mu) = 0$, and ($\Lambda'''(\mu) \neq 0 $ or $ -\Lambda''^2(\mu)\mp \Lambda''''(\mu) \neq 0$), then
\begin{align}
\lim_{\sigma\rightarrow 0}\mathcal V(\mu,\sigma)=\frac{\Lambda'''(\mu)^2}{- \Lambda''^2(\mu) \mp \Lambda''''(\mu)}.
\end{align}
\item If $|\Lambda(\mu)|=1$, $ \Lambda'(\mu) = \Lambda''(\mu) = \Lambda'''(\mu) = 0 $, $ -\Lambda''^2(\mu)\mp \Lambda''''(\mu) = 0$,
and ($
3 \Lambda'''(\mu) \Lambda''''(\mu) \neq 0
$ or $
-3 \Lambda''(\mu) \Lambda''''(\mu) \mp 30 \Lambda^{(6)}(\mu) \neq 0
$), then
\begin{align}
\lim_{\sigma\rightarrow 0}\mathcal V(\mu,\sigma)=\frac{
3\Lambda'''(\mu)\Lambda'''''(\mu)
}{
- 3 \Lambda''(\mu) \Lambda''''(\mu) \mp 30 \Lambda^{(6)}(\mu)}.
\end{align}

\end{itemize}

\section{Engineered likelihood functions}
\label{sec:optimization}

In this section, we apply the tools that we developed in Section \ref{sec:bayesianinference} to the quantum-generated likelihood functions from Section \ref{sec:quantumGeneratedLikelihoodFunctions}. The problem that we wish to solve may be phrased as an optimization problem, which we will state in Section \ref{sec:minExpPosVar} (see Eq.~\eqref{eq:optProblem}). In Section \ref{sec:numerResults}, we numerically solve this optimization problem to compare the performance of various engineered likelihood functions with each other and with the fixed-angle Chebyshev likelihood functions.

\subsection{Minimizing the expected posterior variance}
\label{sec:minExpPosVar}


Our goal is to \textit{engineer} quantum likelihood functions by choosing appropriate tunable parameters $\vec x \in \mathbb R^{2L}$ in the parametrized quantum-generated likelihood functions given by Eq.~\eqref{eq:likelihoodInTermsOfBiasSchemes}. Specifically, we will choose these tunable parameters $\vec x$ to minimize
the expected posterior variance \eqref{eq:expvar} of the unknown parameter
$\theta = \arccos(\bra{\bar 0} P \ket{\bar 0})$ given by Eq.~\eqref{eq:thetaDef}. The likelihood functions that arise from such a minimization are called \textit{engineered likelihood functions}.

We will take the prior distribution to be the Gaussian distribution with probability density function given by $p(\theta;\mu,\sigma)$ (see Eq.~\eqref{eq:gaussianDistribution}) and the likelihood function to be the quantum-generated likelihood function     $\mathcal L^\mathcal A(\theta;d,\vec x) =\frac 12\left[
1+(-1)^d \Lambda^\mathcal A(\theta;\vec x)
\right]$ given by Eq.~\eqref{eq:likelihoodInTermsOfBiasSchemes}. We will consider both the ancilla-free scheme ($\mathcal A = \AF$) and the ancilla-based scheme ($\mathcal A = \AB$).

To explicitly indicate dependence on the prior mean $\mu$, the prior variance $\sigma$, the tunable parameters $\vec x\in \mathbb R^{2L}$ and the scheme $\mathcal A$, we shall denote the expected bias \eqref{eq:expectedBiasGaussian}, the chi function \eqref{eq:chiFunctionGaussian} and the variance reduction factor \eqref{eq:gaussreductionfactor} by
\begin{align}
    b^{\mathcal A}(\mu,\sigma;\vec x) &= \int_{-\infty}^\infty \d \theta \ p(\theta;\mu,\sigma) \Lambda^\mathcal A(\theta;\vec x) 
    \label{eq:bchiV1}
    \\
    \chi^{\mathcal A}(\mu,\sigma;\vec x) &= \frac 1{\sigma^2} \int_{-\infty}^\infty \d \theta \ (\theta-\mu) p(\theta;\mu,\sigma) \Lambda^\mathcal A(\theta;\vec x) = \partial_\mu b^{\mathcal A}(\mu,\sigma;\vec x)
    \label{eq:bchiV2}
    \\
    \mathcal V^{\mathcal A}(\mu,\sigma;\vec x) &= \frac{\chi^{\mathcal A}(\mu,\sigma;\vec x)^2}{1-b^{\mathcal A}(\mu,\sigma;\vec x)^2} \mathds 1_{\Lambda^\mathcal A(\theta;\vec x) \notin \{\pm 1\}}
    \label{eq:bchiV3}
\end{align}
respectively. By Eq.~\eqref{eq:expVar_in_terms_of_V}, the expected posterior variance may be expressed in terms of the variance reduction factor as
\begin{align}
    \mathbb{E}_{\mathrm d} \operatorname{Var}_{\uptheta | \mathrm d}(\theta | d; \mu, \sigma, \vec x, \mathcal A)= \sigma^2\left[1-\sigma^2 \mathcal V^{\mathcal A}(\mu,\sigma;\vec x) \right].
    \label{eq:expVar_in_terms_of_V_quantum}
\end{align}

Next, we find series expansions for the functions \eqref{eq:bchiV1}--\eqref{eq:bchiV3}. By Theorems \ref{thm:expanpansionOfDeltaAF} and \ref{thm:expanpansionOfDeltaAB}, the biases $\Lambda^{\mathcal A}(\theta;\vec x)$, for $\mathcal A \in\{\AF,\AB\}$, can be written as the cosine polynomials 
\begin{align} \Lambda^{\mathcal A}(\theta;\vec x)=
        \sum_{l=0}^{\qA(\mathcal A)} \mu_l^{\mathcal A}(\vec x) \cos(l\theta)
        \label{eq:cosineSeries}
\end{align}
where 
\begin{align}
\qA(\mathcal A)
=
\begin{cases}
2L+1, & \mathcal A = \AF \\
L, & \mathcal A = \AB,
\end{cases}
\label{eq:qmathcalA}
\end{align}
and $\mu_l^{\mathcal A}(\vec x)$'s are given by Eq.~\eqref{eq:cosPolyCoef} and \eqref{eq:cosPolyCoefAB}.

The series \eqref{eq:cosineSeries} allows us to use the machinery introduced in Section \ref{sec:biasTrigoSeries}: by Eqs.~\eqref{eq:bAsCosSeries} and \eqref{eq:chiAsCosSeries}, the bias and the chi function can be written as the sums
\begin{align}
    b^{\mathcal A}(\mu,\sigma;\vec x) &= \sum_{l=0}^{\qA(\mathcal A)} \mu_l^{\mathcal A}(\vec x) \e^{-l^2 \sigma^2/2} \cos(l\mu), 
    \label{eq:bFourierSeries}
    \\
    \chi^{\mathcal A}(\mu,\sigma;\vec x) &= -\sum_{l=1}^{\qA(\mathcal A)} \mu_l^{\mathcal A}(\vec x) \e^{-l^2 \sigma^2/2} l \sin(l\mu).
    \label{eq:chiFourierSeries}
\end{align}
and the variance reduction factor may be written as
\begin{align}
    \mathcal V^{\mathcal A}(\mu,\sigma;\vec x)= \frac{\chi^{\mathcal A}(\mu,\sigma;\vec x)^2}{1-b^{\mathcal A}(\mu,\sigma;\vec x)^2} =  \frac{\left[\displaystyle{\sum_{l=0}^{\qA(\mathcal A)} \mu_l^{\mathcal A}(\vec x) \e^{-l^2 \sigma^2/2} \cos(l\mu)}\right]^2}{1-\left[ 
    \displaystyle{\sum_{l=1}^{\qA(\mathcal A)} \mu_l^{\mathcal A}(\vec x) \e^{-l^2 \sigma^2/2} l \sin(l\mu)}
    \right]^2}.
    \label{eq:varRedFactorelf}
\end{align}

Our goal is to find tunable parameters $\vec x = (x_1,\ldots, x_{2L}) \in \mathbb R^{2L}$ that minimize the expected posterior variance \eqref{eq:expVar_in_terms_of_V_quantum}. Due to the inverse relationship between the expected posterior variance and the variance reduction factor (see Eq.~\eqref{eq:expVar_in_terms_of_V_quantum}),
minimizing the former is equivalent to maximizing the latter.
Since $\mathcal V(\mu,\sigma; \vec x)$ is $2\pi$-periodic in each coordinate $x_i$, it suffices to restrict the search space of each $x_i$ to $(-\pi,\pi]$. In other words, the optimization problem we wish to solve may be stated as:
\begin{align}
\begin{tabular}{cp{7.3cm}}
$\mathtt{Input}$:
&
$(\mu, \sigma,\mathcal A)$, where $\mu\in \mathbb R, \sigma> 0$, $\mathcal A \in \{\AF,\AB\}$
\\[0.2cm]
$\mathtt{Output}$:
&
$\displaystyle \argmax_{\vec x \in (-\pi,\pi]^{2L} }\mathcal V^{\mathcal A}(\mu,\sigma; \vec x)$.
\end{tabular}
\label{eq:optProblem}
\end{align}

Note that the input prior variance $\sigma$ in the optimization problem \eqref{eq:optProblem}  is required to be positive to guarantee that Eq.~\eqref{eq:varRedFactorelf} is well-defined. 
For the case when $\sigma\rightarrow 0$, the results of Section \ref{sec:limitingBehavior} may be used. In particular, 
Eqs.~\eqref{eq:V0Imu} and \eqref{eq:limitVasSigma0One} imply the following.
\begin{itemize}
\item If $|\Lambda^{\mathcal A}(\mu;\vec x)|\neq 1$, then
\begin{align}
    \lim_{\sigma\rightarrow 0}\mathcal V^{\mathcal A}(\mu,\sigma;\vec x)
    =
    \frac{(\partial_\mu \Lambda^{\mathcal A}(\mu;\vec x))^2}{1 - \Lambda^{\mathcal A}(\mu;\vec x)}.
    \label{eq:limV0fisher}
\end{align}
\item If $|\Lambda^{\mathcal A}(\mu;\vec x)|=1$, $\partial_\mu \Lambda^{\mathcal A}(\mu;\vec x)=0$, and $\partial_\mu^2 \Lambda^{\mathcal A}(\mu;\vec x) \neq 0$, then
    \begin{align}
\lim_{\sigma\rightarrow 0}\mathcal V^{\mathcal A}(\mu,\sigma;\vec x)=0.
\label{eq:limV0general}
\end{align}
\end{itemize}


We will present some numerical results obtained by solving the optimization problem \eqref{eq:optProblem} in Section \ref{sec:numerResults}. Before we do that, we study the behavior and properties of the variance reduction factor \eqref{eq:varRedFactorelf} in the special case where the angles $\vec x$ are chosen to give rise of Chebyshev likelihood functions.

\subsection{Chebyshev variance reduction factor}
\label{sec:chebyshevVarianceReduction}

The Fourier coefficients \eqref{eq:mu_as_delta_function} allow us to compute the variance reduction factor \eqref{eq:varRedFactorelf} in the case when $\vec x=(\pi/2)^{2L}$.
By substituting Eq.~\eqref{eq:mu_as_delta_function} into Eqs.~\eqref{eq:bFourierSeries} and \eqref{eq:chiFourierSeries}, 
we find that
\begin{align}
    b^{\mathcal A}\left(\mu,\sigma;\left(\tfrac \pi 2\right)^{2L}\right) &= (-1)^{\rA}\e^{-\qA^2\sigma^2/2} \cos(\qA\mu),
    \label{eq:bmupi2L}
    \\
    \chi^{\mathcal A}\left(\mu,\sigma;\left(\tfrac \pi 2\right)^{2L}\right) &= -(-1)^{\rA}\e^{-\qA^2\sigma^2/2} \qA\sin(\qA\mu).
    \label{eq:chimupi2L}
\end{align}
where
\begin{align}
\rA =
\begin{cases}
0, & \mathcal A = \AF \\
L, & \mathcal A = \AB
\label{eq:rmathcalA}
\end{cases}
\end{align}
and $\qA=\qA(\mathcal A)$ is 
given by Eq.~\eqref{eq:qmathcalA}.

Hence, substituting Eq.~\eqref{eq:bmupi2L} and Eq.~\eqref{eq:chimupi2L} into Eq.~\eqref{eq:varRedFactorelf} gives the following expression for the variance reduction factor:
\begin{align}
    \mathcal V^{\mathcal A}\left(\mu,\sigma;
    \left(\tfrac \pi 2\right)^{2L}
    \right)
    =\frac{\qA^2 \sin^2(\qA \mu)}{\e^{\qA^2 \sigma^2}- \cos^2(\qA\mu)}.
    \label{eq:VmuCheby}
\end{align}

The following proposition lists a few useful properties of the variance reduction factor \eqref{eq:VmuCheby}.

\begin{proposition}
\label{prop:varproperties}
Let $L \in \mathbb Z^+$ and $\sigma>0$. The variance reduction factor $\mathcal V^{\mathcal A}(\mu,\sigma;    (\tfrac \pi 2)^{2L}
    )$ given by Eq.~\eqref{eq:VmuCheby} satisfies the following properties.
\begin{enumerate}
    \item Periodicity (in $\mu)$ with period $\tfrac{\pi}{\qA}$: For all $\mu \in \mathbb R$, 
    \begin{align}
    \mathcal V^{\mathcal A}\left(\mu+ \tfrac{\pi}{\qA},\sigma;
    (\tfrac \pi 2)^{2L}
    \right)=\mathcal V^{\mathcal A}\left(\mu,\sigma;
    (\tfrac \pi 2)^{2L}
    \right).
    \label{eq:minimumV}
    \end{align}
    \item Minimum: For all $\mu\in \mathbb R$, 
    \begin{align}
        \mathcal V^{\mathcal A}\left(\mu,\sigma;
    (\tfrac \pi 2)^{2L}
    \right) \geq 0
    \label{eq:maximumV}
    \end{align} with equality if and only if $\mu \in \frac{\pi}{\qA} \mathbb Z$.
    \item Maximum: For all $\mu\in \mathbb R$, 
    \begin{align}
        \mathcal V^{\mathcal A}\left(\mu,\sigma;
    (\tfrac \pi 2)^{2L}
    \right) \leq \qA^2 \e^{-\qA^2\sigma^2}
    \label{eq:upperBoundV}
    \end{align} with equality if and only if $\mu \in \frac{\pi}{2\qA} \mathbb Z_{\textup{odd}}$.
\end{enumerate}
In the above, $\qA=\qA(\mathcal A)$ is given by Eq.~\eqref{eq:qmathcalA}.
\end{proposition}

\begin{proof} \hfill
\begin{enumerate}
    \item This follows directly from Eq.~\eqref{eq:VmuCheby} and the fact that $\sin^2(\qA \mu)$ and $\cos^2(\qA \mu)$ are both periodic and have period $\tfrac {\pi}{\qA}$.
    \item The numerator $\qA^2 \sin^2(\qA \mu)$ of Eq.~\eqref{eq:VmuCheby} is clearly positive. The denominator $\e^{\qA^2 \sigma^2}- \cos^2(\qA\mu)$ of Eq.~\eqref{eq:VmuCheby} is positive since $\sigma>0$. Equality holds if and only if the numerator $\qA^2 \sin^2(\qA \mu) = 0$, which holds if and only if $\mu \in \frac{\pi}{\qA} \mathbb Z$.
    \item By dividing both the numerator and denominator of Eq.~\eqref{eq:VmuCheby} by $\sin^2(\qA\mu)$, the variance reduction factor \eqref{eq:VmuCheby} can be written as
\begin{align}
    \mathcal V^{\mathcal A}(\mu,\sigma;
    (\tfrac \pi 2)^{2L}
    )
    =\frac{\qA^2}{1+\left( \e^{\qA^2 \sigma^2}-1\right) \csc^2(\qA\mu)}
    \label{eq:VmuChebycsc}
\end{align}
whenever $\sin(\qA\mu) \neq 0$. Since $\csc^2(\qA \mu) \geq 1$,
\begin{align}
    \mathcal V^{\mathcal A}(\mu,\sigma;
    (\tfrac \pi 2)^{2L}
    ) \leq \frac{\qA^2}{1+\e^{\qA^2 \sigma^2}-1} = \qA^2 \e^{-\qA^2\sigma^2}.
\end{align}
Equality holds if and only if $\csc^2(\qA \mu) = 1$, which holds if and only if $\mu \in \frac{\pi}{2\qA} \mathbb Z_{\textup{odd}}$.
\end{enumerate}
\end{proof}

Note that Eq.~\eqref{eq:upperBoundV} can be used to give an upper bound for the variance reduction factor that is independent of $L$: since the function $\qA \mapsto \qA^2 \e^{-\qA^2 \sigma^2}$ achieves a maximum of $(\e\sigma^2)^{-1}$ at $\qA =\tfrac 1{\sigma}$, it follows that
\begin{align}
\mathcal V^{\mathcal A}\left(\mu,\sigma; 
    (\tfrac \pi 2)^{2L}
    \right) \leq \frac{1}{\e \sigma^2},
\end{align}
with equality if and only if $\qA = \frac{1}{\sigma}$ and $\mu \in \frac{\pi}{2\qA} \mathbb Z_{\textup{odd}}$.

Finally, we conclude with a proposition that characterizes the limiting behavior of the Chebyshev variance reduction factor as $\sigma \rightarrow 0$:
\begin{proposition}
Let $L \in \mathbb Z^+$ and $\mu \in \mathbb R$. Then,
\begin{align}
    \lim_{\sigma\rightarrow 0}\mathcal V^{\mathcal A}(\mu,\sigma;
    (\tfrac \pi 2)^{2L}
    )  = \qA^2 \mathds 1_{\mu \notin \tfrac{\pi}{\qA}\mathbb Z}
    \label{eq:v0cases}
\end{align}
where $\qA=\qA(\mathcal A)$ is given by Eq.~\eqref{eq:qmathcalA}.
\label{prop:V0chebyExpressions}
\end{proposition}

\begin{proof}

We first note that the bias  \eqref{eq:biasAsCosine} and its first two derivatives are given by
\begin{align}
    \Lambda^{\mathcal A}\left(\theta;(\tfrac \pi 2)^{2L}
    \right) &= (-1)^\rA \cos(\qA\theta) \nonumber\\
    (\Lambda^{\mathcal A})'\left(\theta;(\tfrac \pi 2)^{2L}
    \right) &= -(-1)^\rA\qA \sin(\qA\theta) \nonumber\\
    (\Lambda^{\mathcal A})''\left(\theta;(\tfrac \pi 2)^{2L}
    \right) &= -(-1)^\rA\qA^2 \cos(\qA\theta)
    \label{eq:derivativesOfBiasCheby}
\end{align}
where $\qA$ and $\rA$ are given by Eq.~\eqref{eq:qmathcalA} and \eqref{eq:rmathcalA} respectively.

If $\mu \notin \tfrac \pi{\qA}\mathbb Z$, then $\Lambda^{\mathcal A}(\mu;\vec x) \neq 1$. By using Eq.~\eqref{eq:limV0fisher}, we obtain
\begin{align}
    \lim_{\sigma\rightarrow 0}\mathcal V^{\mathcal A}(\mu,\sigma;
    (\tfrac \pi 2)^{2L}
    ) = 
    \frac{\qA^2 \sin^2(\qA \mu)}{1-\cos^2(\qA \mu)}
    =
    \qA^2
    .
    \label{eq:V0Limit1}
\end{align}

If $\mu \in \tfrac{\pi}{\qA}\mathbb Z$, it follows from the expressions in
Eq.~\eqref{eq:derivativesOfBiasCheby} that
\begin{align}
    \Lambda^{\mathcal A}\left(\theta;(\tfrac \pi 2)^{2L}
    \right) &\in \{1,-1\} \nonumber\\
    (\Lambda^{\mathcal A})'\left(\theta;(\tfrac \pi 2)^{2L}
    \right) &= 0 \nonumber\\
    (\Lambda^{\mathcal A})''\left(\theta;(\tfrac \pi 2)^{2L}
    \right) &\in \{\qA^2,-\qA^2\}.
    \label{eq:derivativesOfBiasChebyDeadspots}
\end{align}
Therefore, in this case, we have
\begin{align}
    |\Lambda^{\mathcal A}\left(\theta;(\tfrac \pi 2)^{2L}
    \right)| = 1, \quad
    (\Lambda^{\mathcal A})'\left(\theta;(\tfrac \pi 2)^{2L}
    \right) = 0, \quad
    (\Lambda^{\mathcal A})''\left(\theta;(\tfrac \pi 2)^{2L}
    \right) \neq 0.
\end{align}
Hence, by Eq.~\eqref{eq:limV0general}, we obtain
\begin{align}
\lim_{\sigma\rightarrow 0}\mathcal V^{\mathcal A}(\mu,\sigma;
    (\tfrac \pi 2)^{2L}
    ) = 0,
\end{align}
which completes the proof of the proposition.
\end{proof}

\subsection{Numerical simulations}
\label{sec:numerResults}

\begin{figure}[htbp]
    \centering
    \includegraphics[trim={2.94cm 4.9cm 2.07cm 4.75cm},clip]{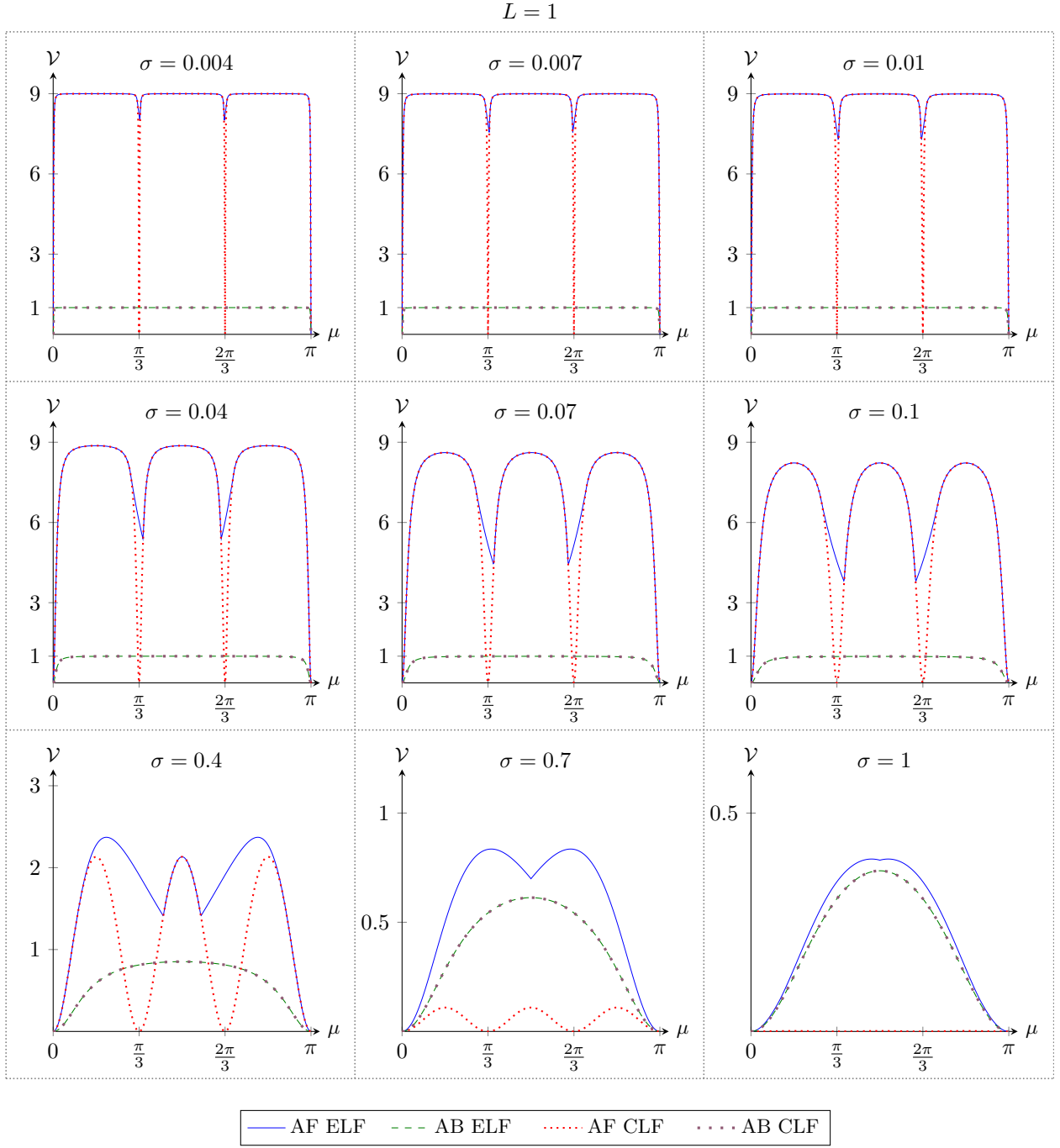}
  
    \caption{Plots of the variance reduction factor versus the prior mean $\mu$ for $L=1$ for various prior variances $\sigma$. In each plot, the (optimized) engineered likelihood function (ELF) is compared with the Chebyshev likelihood function (CLF) for both the ancilla-free (AF) and ancilla-based (AB) schemes. Note that the AB ELF and the AB CLF curves are completely identical. We give a proof of this in Appendix~\ref{app:optimalityL1}
    .
}
  \label{fig:grid_L1}
\end{figure}

\begin{figure}[htbp]
    \centering
    \includegraphics[trim={2.94cm 4.9cm 2.07cm 4.75cm},clip]{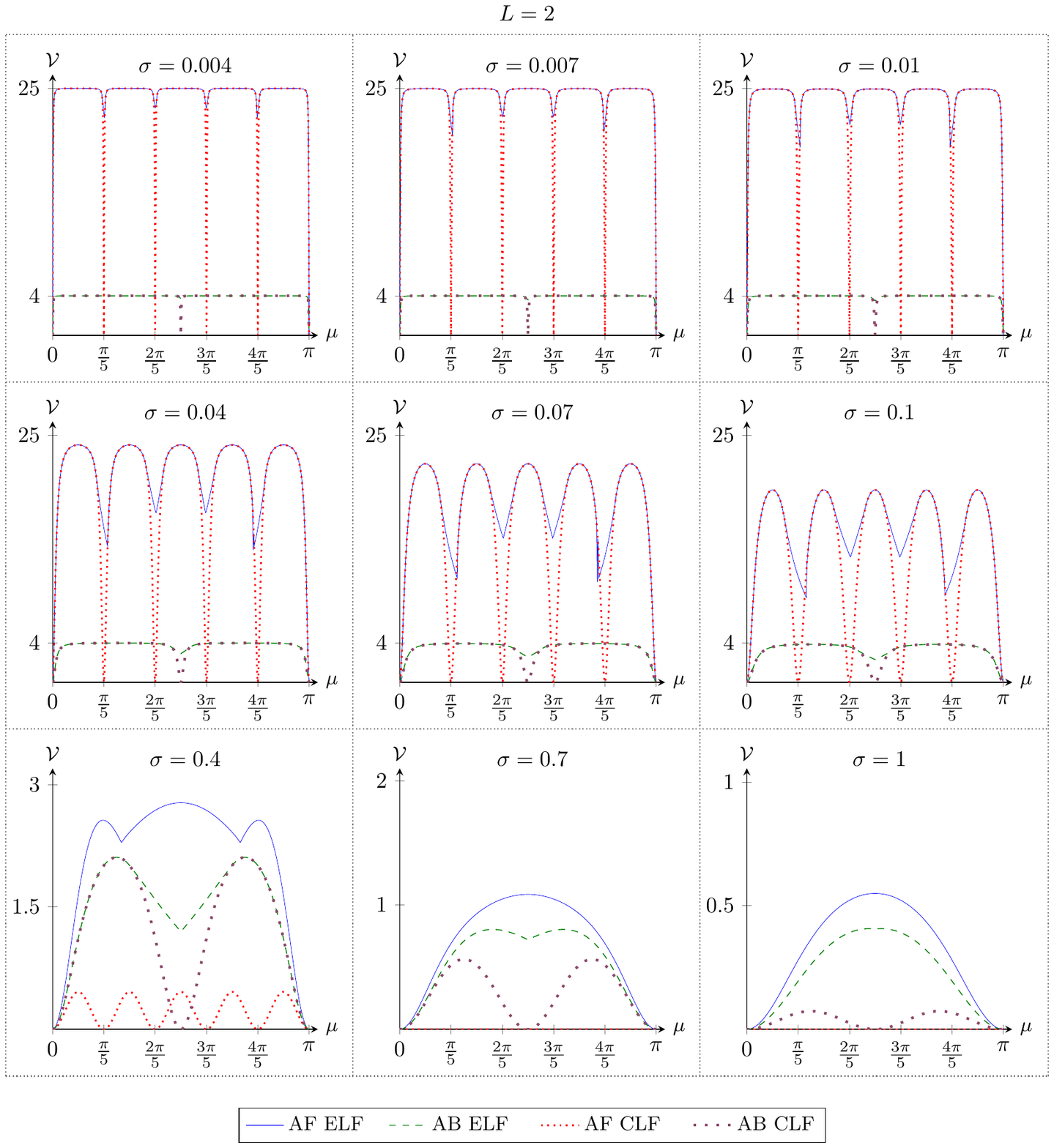}
  
    \caption{Plots of the variance reduction factor versus the prior mean $\mu$ for $L=2$ for various prior variances $\sigma$. In each plot, the (optimized) engineered likelihood function (ELF) is compared with the Chebyshev likelihood function (CLF) for both the ancilla-free (AF) and ancilla-based (AB) schemes.
}
  \label{fig:grid_L2}
\end{figure}

\begin{figure}[htbp]
    \centering
    \includegraphics[trim={2.94cm 4.9cm 2.07cm 4.75cm},clip]{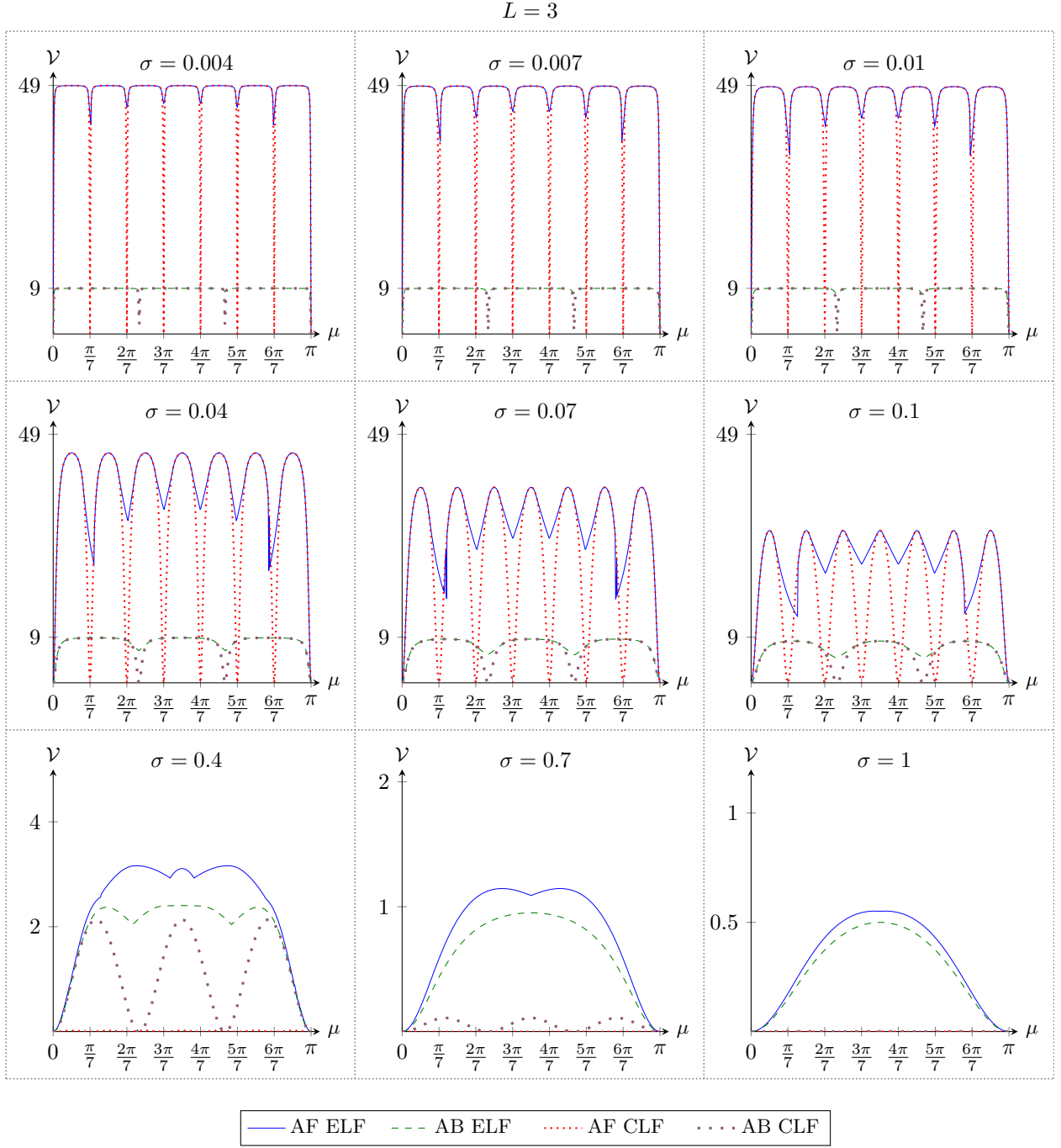}
  
    \caption{Plots of the variance reduction factor versus the prior mean $\mu$ for $L=3$ for various prior variances $\sigma$. In each plot, the (optimized) engineered likelihood function (ELF) is compared with the Chebyshev likelihood function (CLF) for both the ancilla-free (AF) and ancilla-based (AB) schemes.
}
  \label{fig:grid_L3}
\end{figure}

The goal here is to solve the optimization problem \eqref{eq:optProblem}. The objective function, which we seek to maximize, is the variance reduction factor, whose analytical expression is given by Eq.~\eqref{eq:varRedFactorelf}. To solve this optimization problem, we performed the following steps using Wolfram Mathematica \cite{Mathematica}:
\begin{enumerate}[label=(\roman*)]
    \item compute the Fourier coefficients $\mu_l^{\mathcal A}(\vec x)$'s using the expansion formulas given in Eq.~\eqref{eq:cosPolyCoef} and \eqref{eq:cosPolyCoefAB}.
    \item compute the bias \eqref{eq:bFourierSeries} and the chi function \eqref{eq:chiFourierSeries} and plug the resulting expressions into Eq.~\eqref{eq:varRedFactorelf} for the variance reduction factor.
    \item Feed the expression for the variance reduction factor into Mathematica's built-in optimization function \texttt{NMaximize} to find the parameters $(x_1,\ldots, x_{2L})$
    that maximize $\mathcal V^{\mathcal A}(\mu,\sigma; \vec x = x_1,\ldots, x_{2L})$. 
\end{enumerate}

We present the results of this optimization in Figures \ref{fig:grid_L1}--\ref{fig:grid_L3}, where we plot graphs of the variance reduction factor $\mathcal V^{\mathcal A}(\mu,\sigma; \vec x)$ versus the prior mean $\mu$ for different values of $\sigma$ for $L=1,2,3$. In each graph, we compare the engineered likelihood function (ELF) with the Chebyshev likelihood function (CLF) for both the ancilla-free (AF) and ancilla-based (AB) schemes. As seen from the plots, in both the AF and AB schemes, there exist values of $\mu$ for which the ELFs outperform the CLFs (i.e. the ELF variance reduction factor is larger than the CLF variance reduction factor). This gap in performance decreases as $\sigma$ goes to zero, suggesting that using ELFs would be most beneficial when the prior variance $\sigma$ is large.

Another property of the plots is that the ancilla-free schemes yield much larger values for the expected posterior variance than for the ancilla-based case. This can be explained by the fact that as cosine polynomials, the ancilla-free bias has degree $2L+1$ (see Eq.~\eqref{eq:biasAsCosineSeriesAndTrigonoquadraticseries}), while the ancilla-based bias has only degree $L$ (See Eq.~\eqref{eq:biasABSeries}).

An interesting observation from Figure \ref{fig:grid_L1} is that for the $L=1$ ancilla-based scheme, the performance of the optimized engineered likelihood function is identical to that of the Chebyshev likelihood function. In other words, the Chebyshev likelihood function is optimal, i.e.
\begin{align}
    \displaystyle \argmax_{x_1, x_2 \in (-\pi,\pi] }\mathcal V^{\mathrm{AB}}(\mu,\sigma; x_1, x_2)
    \ni \left(\tfrac \pi 2,\tfrac \pi 2\right).
    \label{eq:ABchebyoptimal}
\end{align}
Using the tools developed in Appendix \ref{sec:seriesExpansions} (see Eq.~\eqref{eq:cosPolyCoefAB}), we shall give an analytical proof of the observation \eqref{eq:ABchebyoptimal} in Appendix \ref{app:optimalityL1}.

\section{Concluding remarks}

In this paper, we developed tools for characterizing and analyzing ELFs, focusing on the ancilla-based and ancilla-free schemes. Both these schemes involve alternate applications of generalized reflection operators, which may be visualized as rotations in a two-dimensional subspace. This visualization can be used to show that each of these schemes produces likelihood functions that can be written as cosine polynomials whose degree scales with the number of alternations. We showed that these polynomials can be used to derive analytical expressions for the expected posterior variance describing the parameter of interest. Finally, we presented simulation results to compare the performance of various ELFs with each other and to CLFs.

The results in this paper may be extended in a number of ways. Firstly, while we have limited the scope of this paper to only noiseless ELFs, the results here can be generalized to the case where the ELFs are noisy. This case---which arises when the states, transformations and measurements in circuits
in Figure \ref{fig:elfcircuit} are replaced by imperfect noisy versions of themselves, and which is arguably more relevant in this noisy intermediate-scale quantum \cite{Preskill2018quantumcomputingin} era,
where near-term quantum devices are subject to high levels of noise---is treated in detail in a separate paper \cite{ELFPaper}, which builds on the foundations laid here (see Appendix \ref{sec:noise} for a brief discussion of incorporating noise into the ELF framework).

Secondly, while we have modeled the prior distribution by a Gaussian random variable \eqref{eq:gaussianDistribution}, it might be appropriate in certain cases to model it using other distributions. We leave the sensitivity analysis on the prior distribution chosen to future work. Thirdly, while we have focused our attention on two specific schemes, namely the ancilla-based scheme and the ancilla-free scheme, we note that this framework can be extended to variants or extensions of these schemes. We leave this consideration for future work.

\section*{Acknowledgments}

We thank Peter J.~Love and Siong Thye Goh for helpful discussions.

\appendix
\renewcommand{\theequation}{\thesection \arabic{equation}}
\counterwithin{equation}{section}

\section{List of mathematical symbols}
\label{sec:notation}

We list in this appendix some of the mathematical symbols that appear in this paper.

\vspace{0.2cm}
\begin{tabular}{ll}
     $!!$ & double factorial 
     \\
     $2 \mathbb Z$ & set of even integers
     \\
     $A^\dag$ & conjugate transpose of the linear operator $A$
     \\
     $S^*$ &
     Kleene closure of the set $S$
     \\
     $[n]$ &
     set of integers from 1 to $n$, for $n\in \mathbb Z^+$. Denotes the set $ \{1,2,\ldots,n \}$
     \\
     $\bar z$ &
     complex conjugate of $z\in\mathbb C$
     \\
     $\delta_{ij}$ &
     Kronecker delta. Equals 1 when $i=j$ and 0 otherwise
     \\
     $\delta_{i}$ &
     single-argument Kronecker delta. Equals 1 when $i=0$ and 0 otherwise
     \\
     $\emptyset$ &
     empty set
     \\
     $\equiv_n$ &
     is congruent to (modulo $n$). Write $a \equiv_n b$ to mean $a\equiv b \pmod{n}$, where $a,b,n\in \mathbb Z$
     \\
     $\i$ &
     imaginary unit $\sqrt{-1}$
     \\
     $\lceil z \rceil$
     &
     ceiling of $x\in \mathbb R$
     \\
     $\lfloor z \rfloor$
     &
     floor of $x\in \mathbb R$
     \\
     $\mathbb C$ &
     set of complex numbers
     \\
     $\mathbb E(\cdot)$
     &
     expectation value (of a random variable)
     \\
     $\mathbb F_2^n$
     &
     set of $n$-bit strings over the alphabet $\{0,1\}$, where $n\in \mathbb N$
     \\
     $\mathbb N$
     &
     set of natural numbers (including 0)
     \\
     $\mathbb P(\cdot)$
     &
     probability (of an event)
     \\
     $\mathbb R$
     &
     set of real numbers
     \\
     $\mathbb Z$
     &
     set of integers
     \\
     $\mathbb Z^+$
     &
     set of positive integers
     \\
     $\mathbb Z_{\textup{odd}}$
     &
     set of odd integers
     \\
     $\mathds 1_{\texttt S}$
     &
     indicator function of the statement $\texttt S$. Equals 1 if $\texttt S$ is true and $0$ otherwise
     \\
     $\mathrm{Im}(z)$
     &
     imaginary part of $z\in\mathbb C$. Equals $\tfrac 1{2\i} (z-\bar z)$
     \\
     $\mathrm{Re}(z)$
     &
     real part of $z\in\mathbb C$. Equals $\tfrac 1{2} (z+\bar z)$
     \\
     $\mathrm{e}$
     &
     Euler's number $ 2.718...$
     \\
     $\mathrm{wt}$
     &
     weight (of a string). For $\vec x = x_1 x_2 \ldots x_n \in \{0,1\}^n$, $\wt(\vec x) = |\{ i \in [n]:x_i=1 \}|$
     \\
     $\operatorname{Var}$
     &
     variance (of a random variable)
     \\
     $\tr$
     &
     trace (of a linear operator)
     \\
     $\vec x^R$
     &
     reverse of the string $\vec x$. If $\vec x = x_1 x_2\ldots x_n$, then $\vec x^R = x_n\ldots x_2 x_1$
     \\
     $f'$
     &
     derivative of the function $f$. For example, $f'(x) = \tfrac{\d f(x)}{\d x}$
\end{tabular}


\section{Series expansions of the ancilla-free and ancilla-based biases}
\label{sec:seriesExpansions}

In this appendix, we will derive series expansion formulas\footnote{For the series expansion formulas, see Theorems \ref{thm:expanpansionOfDeltaAF} and \ref{thm:expanpansionOfDeltaAB}, which this appendix will culminate in.} for both the ancilla-free and ancilla-based biases, which will be useful in Section \ref{sec:optimization}. In order to state and prove the theorems, we first present some necessary preliminary definitions and lemmas in Appendices \ref{sec:sumOfProductsExpansion} and \ref{sec:tmlf}.
The main goal of \ref{sec:sumOfProductsExpansion} and \ref{sec:tmlf} is to develop the tools and notation needed to give expressions for the coefficients in the series expansions.

\subsection{Mathematical preliminaries I: Expansion formulas
}
\label{sec:sumOfProductsExpansion}

Let $n,k\in \mathbb N$ and $u,v \in \mathbb F_2$. The central object that we introduce here is the set $\Theta_{ukv}^n$, which is (informally) defined as follows:
\begin{itemize}[label={},itemindent=0em,leftmargin=4.5em]
\item[$\Theta_{ukv}^n =$] set of strings $\vec x = x_1 x_2\ldots x_n \in \mathbb F_2^n$ for which the string $r^{x_n}\ldots p^{x_4}q^{x_3}p^{x_2}q^{x_1}$ can be converted to $p^u(qp)^k q^v$ by repeatedly applying the replacement rules $pp\rightarrow \varepsilon$ and $qq \rightarrow \varepsilon$, where $r =p$ if $n$ is even and $r=q$ if $n$ is odd and $\varepsilon$ denotes the empty string. Here, $p^x = \varepsilon$ if $x=0$ and $p^x=p$ if $x=1$.
\end{itemize}

\noindent For example, the string $101101011 \in \Theta_{010}^9$ since
\begin{align}
    q^1 p^1 q^0 p^1 q^0 p^1 q^1 p^0 q^1 = qpppqq \rightarrow q\cancel{pp}p\cancel{qq} = qp = p^0(qp)^1 q^0 .
    \label{eq:exampleTheta}
\end{align}

For a formal definition of $\Theta_{ukv}^n$, see Appendix \ref{sec:expansionFormulasAppendix}. By convention, we take 
$\Theta_{ukv}^n =\emptyset$ whenever $k\notin \mathbb N$. The following lemma, which we prove in Appendix \ref{sec:expansionFormulasAppendix}, characterizes the set of $k$ values for which $\Theta_{ukv}^n$ is nonempty. Note that this strengthens the trivial upper bound $k \leq (n - u - v)/2$.

\begin{lemma} (also Lemma \ref{lem:nonZeroTheta})
    Let $n\in \mathbb N$, $u,v \in \mathbb F_2$ and $k\in\mathbb N$. Then,
    \begin{align}
        \Theta_{ukv}^n \neq \emptyset \iff
        k &\leq
        \begin{cases}
            \tfrac n2-1 & n \textup{ even}, \\
            \tfrac {n-1}2-u & n \textup{ odd}.
        \end{cases}
    \end{align}
    \label{lem:re:nonZeroTheta}
\end{lemma}
We now introduce our key lemma, which we prove in Appendix \ref{sec:expansionFormulasAppendix}:
\begin{lemma} (also Lemma \ref{lem:keyExpansionLemma})
Let $n \in \mathbb Z^+$ and $P^2 = Q^2 = I$.
Let $\{a_x^y:x \in \mathbb F_2$, $y \in [n] \} \subset \mathbb C$. Then,
\begin{align}
    &\left(a_0^n+a_1^n R\right)\ldots
    \left(a_0^4+a_1^4 P\right)
    \left(a_0^3+a_1^3 Q\right)
    \left(a_0^2+a_1^2 P\right)
    \left(a_0^1+a_1^1 Q\right) \nonumber\\
    &\qquad =
    \sum_{k=0}^\infty \sum_{u,v\in \mathbb F_2}\left( 
    \sum_{\vec x \in \Theta_{ukv}^n} a_{x_1}^1 a_{x_2}^2 \ldots a_{x_n}^n
    \right)
    P^u (QP)^k Q^v ,
    \label{eq:expansionFormula}
\end{align}
where $R = P$ if $n$ is even and $R=Q$ if $n$ is odd.
\label{lem:keyExpansionLemma2}
\end{lemma}

Keeping only the nonzero terms in the expansion gives
\begin{align}
    &\left(a_0^n+a_1^n R\right)\ldots
    \left(a_0^4+a_1^4 P\right)
    \left(a_0^3+a_1^3 Q\right)
    \left(a_0^2+a_1^2 P\right)
    \left(a_0^1+a_1^1 Q\right) \nonumber\\
    &\qquad =
    \sum_{u,v\in \mathbb F_2}
    \sum_{k=0}^{\left\lfloor\frac{n-1}2\right\rfloor - u \mathds 1_{n\in 2\mathbb Z+1}}
    \left( 
    \sum_{\vec x \in \Theta_{ukv}^n} a_{x_1}^1 a_{x_2}^2 \ldots a_{x_n}^n
    \right)
    P^u (QP)^k Q^v.
\end{align}

The next important object that we introduce is $\Xi$, which is defined as follows: for $\alpha \in \mathbb Z^+$ and $l \in \mathbb N$, let
\begin{align}
    \Xi_l^\alpha &= \bigcup\limits_{u,v \in \mathbb F_2} \Theta^\alpha_{u,l-v,v} \nonumber\\
    &=
    \Theta^\alpha_{0l0} \cup
    \Theta^\alpha_{1l0} \cup
    \Theta^\alpha_{0,l-1,1} \cup
    \Theta^\alpha_{1,l-1,1}.
    \label{eq:XiAsUnionThetas}
\end{align}

It is straightforward to check that
\begin{align}
    \Theta^\alpha_{ukv} \subseteq \Xi^\alpha_{k+v}. 
\end{align}

Next, we use Lemma 
\ref{lem:re:nonZeroTheta} to characterize the set of $k$ values for which $\Xi_l^n$ is nonempty.

\begin{lemma}
    Let $\alpha \in \mathbb Z^+$ and $l\in \mathbb N$. Then,
    \begin{align}
        \Xi_l^\alpha \neq \emptyset \iff l \leq \lceil \alpha/2 \rceil.
    \end{align}
\label{lem:whenIsXiEmpty}
\end{lemma}
\begin{proof}
We first prove the forward direction. Assume that $\Xi_l^\alpha \neq \emptyset$. Then at least one of the following holds: (i) $\Theta_{0l0}^\alpha \neq \emptyset$, (ii) $\Theta_{1l0}^\alpha \neq \emptyset$, (iii) $\Theta_{0,l-1,1}^\alpha \neq \emptyset$, (iv) $\Theta_{1,l-1,1}^\alpha \neq \emptyset$. Making use of Lemma \ref{lem:re:nonZeroTheta}, we find that
\begin{enumerate}
    \item If $\Theta_{0l0}^\alpha \neq \emptyset$, then $l \leq \left\lfloor\frac{\alpha-1}2 \right\rfloor \leq \left\lfloor\frac{\alpha+1}2 \right\rfloor = \lceil \alpha/2 \rceil $.
    \item If $\Theta_{1l0}^\alpha \neq \emptyset$, then $l \leq \left\lfloor\frac{\alpha-1}2 \right\rfloor 
    -\mathds 1_{n\in 2\mathbb Z+1}
    \leq  \left\lfloor\frac{\alpha-1}2 \right\rfloor
    \leq \left\lfloor\frac{\alpha+1}2 \right\rfloor = \lceil \alpha/2 \rceil $.
    \item If $\Theta_{0,l-1,1}^\alpha \neq \emptyset$, then $l-1 \leq \left\lfloor\frac{\alpha-1}2 \right\rfloor 
    \implies l \leq 1+\left\lfloor\frac{\alpha-1}2 \right\rfloor =\left\lfloor\frac{\alpha+1}2 \right\rfloor
    = \lceil \alpha/2 \rceil $.
    \item If $\Theta_{1,l-1,1}^\alpha \neq \emptyset$, then $l-1 \leq \left\lfloor\frac{\alpha-1}2 \right\rfloor -\mathds 1_{n\in 2\mathbb Z+1}
    \leq \left\lfloor\frac{\alpha-1}2 \right\rfloor
    \implies l \leq 1+\left\lfloor\frac{\alpha-1}2 \right\rfloor =\left\lfloor\frac{\alpha+1}2 \right\rfloor
    = \lceil \alpha/2 \rceil $.
\end{enumerate}
In all these cases, $l\leq \lceil \alpha/2 \rceil$.

Next, we prove the reverse direction. Assume that $l \leq \lceil \alpha/2 \rceil = \left\lfloor\frac{\alpha+1}{2}\right\rfloor
=
1+\left\lfloor\frac{\alpha-1}{2}\right\rfloor
$. which implies that $l-1\leq \left\lfloor\frac{\alpha-1}{2}\right\rfloor$. By Lemma \ref{lem:re:nonZeroTheta}, $\Theta_{0,l-1,1}^\alpha \neq \emptyset$, which implies that
$\Xi^\alpha_l \neq \emptyset$.
\end{proof}


In Appendix \ref{sec:propertiesOfXiAppendix}, we strengthen both Lemmas \ref{lem:re:nonZeroTheta} and \ref{lem:whenIsXiEmpty} by finding the cardinalities of the sets $\Xi_l^\alpha$ and $\Theta_{ukv}^n$. These cardinalities will be useful for bounding the space complexity of computing various sums related to Eq.~\eqref{eq:expansionFormula}.

\subsection{Mathematical preliminaries II: Trigono-multilinear and trigono-multiquadratic functions
}
\label{sec:tmlf}

For $k \in \mathbb Z^+$, let $\vec x = (x_1,\ldots, x_k) \in \mathbb R^k$ and $\vec y = y_1\ldots y_k \in \{0,1\}^k$.
Define
\begin{align}
    \zeta_{\vec y}(\vec x):=\prod_{a :y_a = 0} \cosp{x_{a}} \prod_{b :y_b=1} \sinp{x_{b}}  .
    \label{eq:zetayx}
\end{align}
For example,
\begin{align*}
\zeta_{00101}(x_1, x_2, x_3,x_4, x_5) =
\cos(x_1) \cos(x_2) \sin(x_3) \cos(x_4) \sin(x_5).   
\end{align*}
When $k=1$, each $\zeta_y(\cdot)$ is a trigonometric function: $\zeta_0(x) = \cos (x)$ and $\zeta_1(x) = \sin (x)$, i.e.
\begin{align}
    \zeta_y(x) = (\sin x)^y (\cos x)^{1-y}.
\end{align}
It is easy to see from the definition that
\begin{align}
    \zeta_{y_1 y_2 \ldots y_n}(x_1,x_2,\ldots,x_k) = \zeta_{y_1}(x_1)\zeta_{y_2}(x_2)\ldots \zeta_{y_k}(x_k).
\end{align}

The functions $\zeta_{\vec y}(\vec x)$ can be used to define the notions of trigono-multilinearity and trigono-multiquadraticity.

\begin{definition}
Let $k\in \mathbb Z^+$. A $k$-ary function $f:\mathbb R^k \rightarrow \mathbb C$ is \textit{trigono-multilinear} if for all $\vec y \in \{0,1\}^k$, there exists $\xi_{\vec y} \in \mathbb C$ such that for all $\vec x \in \mathbb R^k$,
\begin{align}
    f(\vec x) = \sum_{\vec y\in\{0,1\}^k} \xi_{\vec y} \zeta_{\vec y}(\vec x).
    \label{eq:tmlf}
\end{align}
\label{def:tmlf}
\end{definition}
\begin{definition}
Let $k\in \mathbb Z^+$. A $k$-ary function $f:\mathbb R^k \rightarrow \mathbb C$ is \textit{trigono-multiquadratic} if for all $\vec y, \vec z \in \{0,1\}^k$, there exists $\xi_{\vec y \vec z} \in \mathbb C$ such that for all $\vec x \in \mathbb R^k$,
\begin{align}
    f(\vec x) = \sum_{\vec y,\vec z\in\{0,1\}^k} \xi_{\vec y\vec z} \zeta_{\vec y\vec z}(\vec x,\vec x),
    \label{eq:tmqf}
\end{align}
where $\vec y \vec z = y_1\ldots y_k z_1\ldots z_k \in \{0,1\}^{2k}$ is the string formed from concatenating $\vec y$ and $\vec z$.
\label{def:tmqf}
\end{definition}

Equivalently, the trigono-multilinear functions are those that can be written as
\begin{align}
    f(x_1,\ldots, x_k) = \sum_{y_1\ldots y_k \in\{0,1\}^k} \xi_{y_1\ldots y_k} \prod_{a:y_a=0} \cos(x_a) \prod_{b:y_b=1} \sin(x_b)
\end{align}
and the trigono-multiquadratic functions are those that can be written as
\begin{align}
    f(x_1,\ldots, x_k) = \sum_{y_1\ldots y_k \in\{0, 1, 2\}^k} \eta_{y_1\ldots y_k} \prod_{a :y_a=0} \cospt{x_a} \prod_{b:y_b=1} \sinpt{x_b} \prod_{c: y_c=2} \csp{x_c},
\end{align}
where $\xi_{y_1\ldots y_k}, \eta_{y_1\ldots y_k} \in \mathbb C$.

When $k=1$, the above expressions simplify as follows. A unary trigono-multilinear function is of the form
\begin{align}
    f(x) = \xi_0 \cos(x) + \xi_1 \sin(x)
\end{align}
and a unary trigono-multiquadratic function is of the form
\begin{align}
    f(x) = \eta_0 \cos^2(x) + \eta_1 \sin^2(x)
    + \eta_2 \sin(x) \cos(x) ,
\end{align}
where $\xi_0, \xi_1, \eta_0, \eta_1, \eta_2 \in \mathbb C$.

For a larger example, consider $k=4$. A 4-ary
trigono-multilinear function takes the form
\begin{align}
f(x_1, x_2, x_3, x_4)&=\xi_{0000} \cosp{x_1} \cosp{x_2} \cosp{x_3} \cosp{x_4} \nonumber \\
&\quad +\xi_{0001} \cosp{x_1} \cosp{x_2} \cosp{x_3} \sinp{x_4} \nonumber \\
&\quad+\dots \nonumber \\
&\quad+ \xi_{1110} \sinp{x_1}\sinp{x_2} \sinp{x_3} \cosp{x_4} \nonumber \\
&\quad+ \xi_{1111} \sinp{x_1}\sinp{x_2} \sinp{x_3} \sinp{x_4}
\label{eq:tmlfexample}
\end{align}
and a 4-ary trigono-multiquadratic function takes the form
\begin{align}
f(x_1, x_2, x_3, x_4) &= \eta_{0000} \cospt{x_1} \cospt{x_2} \cospt{x_3} \cospt{x_4} \nonumber \\
&\quad+\eta_{0001} \cospt{x_1} \cospt{x_2} \cospt{x_3} \sinpt{x_4} \nonumber \\
&\quad+\eta_{0002} \cospt{x_1}  \cospt{x_2} \cospt{x_3} \csp{x_4}  \nonumber \\
&\quad+\dots \nonumber \\
&\quad+ \eta_{2220} \csp{x_1}  \csp{x_2}  \csp{x_3}  \cospt{x_4}  \nonumber \\
&\quad+ \eta_{2221} \csp{x_1} \csp{x_2} \csp{x_3} \sinpt{x_4}  \nonumber \\
&\quad+ \eta_{2222} \csp{x_1} \csp{x_2} \csp{x_3} \csp{x_4}.
\label{eq:tmqfexample}
\end{align}

Trigono-multilinear and trigono-multiquadratic functions have various nice properties that are useful and of independent interest. We explore some of these properties in Appendix \ref{sec:trigonofunctions}.

\subsection{Applying the expansion formulas to the quantum-generated biases
}
\label{sec:expansion_elf_bias}

We will now apply the results of Sections \ref{sec:sumOfProductsExpansion} and \ref{sec:tmlf}
to the quantum-generated biases $\Lambda^{\mathcal A}(\theta;\vec x)$. Specifically, we will use the expansion formula of Eq.~\eqref{eq:expansionFormula} to expand the functions $Q(\theta,\vec z)$, $Q_{00}[\vec z](\theta)$ and the biases $\Lambda^{\mathcal A}(\theta;\vec x)$ and show that  $\Lambda^{\AF}(\theta;\vec x)$ and $\Lambda^{\AB}(\theta;\vec x)$ are trigono-multiquadratic and trigono-multilinear functions (of $\theta$) respectively. For an example, we refer the reader to Appendix \ref{sec:L1example}, where we work out explicit expressions for the expansion formulas in the case when $L=1$.

First, the expansion formula when applied to $Q(\theta; \vec z)$ gives the following expression.
\begin{thm}
Let $\vec z\in\mathbb R^\alpha$. Then, written in the basis $\{\ket{\bar{0}}, \ket{\bar{1}}\}$, 
\begin{align}
Q(\theta; \vec z)=\sum_{k=0}^{\infty} \sum_{u, v \in \mathbb{F}_{2}} \sum_{\vec y \in \Theta_{ukv}^{\alpha} }(-\i)^{\mathrm{wt}(\vec y)} \zeta_{\vec y}(\vec z)\left(\begin{array}{cc}
\cos [(k+v) \theta] & -(-1)^{v} \sin [(k+v) \theta] \\
(-1)^{u} \sin [(k+v) \theta] & (-1)^{u+v} \cos [(k+v) \theta]
\end{array}\right),
\label{eq:ExpansionExpressionForQ}
\end{align}
where $\wt(y)$ denotes the Hamming weight of $y$ (see Appendix \ref{sec:notation}).
\label{thm:ExpansionExpressionForQ}
\end{thm}
\begin{proof}

Substituting Eqs.~\eqref{eq:ualpha} and \eqref{eq:vbeta} into Eq.~\eqref{eq:Q_as_a_product_of_Us_and_Vs} allows us to express $Q(\theta;\vec z)$ in the following form:
\begin{align}
    Q(\theta;\vec z) &=
   (\cos z_\alpha-\i\sin z_\alpha W) \ldots 
   (\cos z_4-\i \sin z_4 \bar Z)
   (\cos z_3-\i \sin z_3 P) \nonumber\\
   &\qquad \times(\cos z_2-\i \sin z_2 \bar Z)
   (\cos z_1-\i \sin z_1 P)
\nonumber\\
&=    \left(a_0^\alpha+a_1^\alpha W\right)\ldots
    \left(a_0^4+a_1^4 \bar Z\right)
    \left(a_0^3+a_1^3 P\right)
    \left(a_0^2+a_1^2 \bar Z\right)
    \left(a_0^1+a_1^1 P\right)
    \label{eq:productOfSumsForQ}
\end{align}
where $P$ denotes $P(\theta)$,
\begin{align}
    W &= \begin{cases}
        P(\theta), & \alpha \mbox{ odd} \\
        \bar Z, & \alpha \mbox{ even}
    \end{cases}
\end{align}
and for $j\in \{0,1\}$ and $k \in \{1,2,\ldots, \alpha\}$,
\begin{align}
    a_j^k 
    &= \begin{cases}
    \cos z_k, & j=0 \nonumber\\
   -\i \sin z_k, & j=1
    \end{cases} \\
    &= (-\i \sin z_k)^j (\cos z_k)^{1-j} \nonumber\\
    &= (-\i)^j \zeta_j(z_k),
\end{align}
where $\zeta_j(z_k)$ was defined in Eq.~\eqref{eq:zetayx}.
Since $\bar Z^2 = P^2 = I$, we can expand 
Eq.~\eqref{eq:productOfSumsForQ} according to the Lemma \ref{lem:keyExpansionLemma2}.

This gives
\begin{align}
    Q(\theta;\vec z) = \sum_{k=0}^\infty \sum_{u,v\in \mathbb F_2}\Bigg( 
    \sum_{\vec y \in \Theta_{ukv}^n} \underbrace{a_{y_1}^1 a_{y_2}^2 \ldots a_{y_\alpha}^\alpha}_{\Circled{1}}
    \Bigg)
    \underbrace{\bar Z^u (P\bar Z)^k P^v}_{\Circled{2}} .
    \label{eq:expansionFormulaForQ}
\end{align}

Now,
\begin{align}
    \Circled{1} &=  \prod_{j=1}^\alpha a_{y_j}^k 
    = \prod_{j=1}^\alpha (-\i)^{y_j} \zeta_{y_j}(z_k) 
    = (-\i)^{\wt(\vec y)} \prod_{j=1}^\alpha \zeta_{y_j} (z_k) 
    = (-\i)^{\wt(\vec y)} \zeta_{\vec y}(\vec z).
    \label{eq:expansionFormula1}
\end{align}

And,
\begin{align}
    \Circled{2} &= {\bar Z}^{u}(P\bar Z)^{k} P^{v} \nonumber\\
    &=
    \left(\begin{array}{cc}1 & 0 \\ 0 & -1\end{array}\right)^{u}\left(\begin{array}{cc}\cos k \theta & -\sin k \theta \\ \sin k \theta & \cos k \theta\end{array}\right)\left(\begin{array}{cc}\cos \theta & \sin \theta \\ \sin \theta & -\cos \theta\end{array}\right)^{v}\nonumber\\
    &=
    \left(\begin{array}{cc}1 & 0 \\ 0 & (-1)^u \end{array}\right)\left(\begin{array}{cc}\cos k \theta & -\sin k \theta \\ \sin k \theta & \cos k \theta\end{array}\right)\left(\begin{array}{cc}\cos v\theta & \sin v\theta \\ \sin v\theta & (-1)^v \cos v\theta\end{array}\right) \nonumber\\
    &= \left(\begin{array}{cc}
\cos [(k+v) \theta] & -(-1)^{v} \sin [(k+v) \theta] \\
(-1)^{u} \sin [(k+v) \theta] & (-1)^{u+v} \cos [(k+v) \theta]
\end{array}\right).
\label{eq:expansionFormula2}
\end{align}

Substituting Eqs.~\eqref{eq:expansionFormula1} and \eqref{eq:expansionFormula2} into
Eq.~\eqref{eq:expansionFormulaForQ} gives Eq.~\eqref{eq:ExpansionExpressionForQ}.

\end{proof}

Theorem \ref{thm:ExpansionExpressionForQ} can be used to expand $Q_{00}$ as follows.
\begin{thm}
    Let $\vec z \in \mathbb R^\alpha$. Then,
    \begin{align}
        Q_{00}[\vec z](\theta) &=
        \sum_{l=0}^{\lceil \alpha/2\rceil} \left\{ \sum_{y\in \Xi_l^\alpha} (-\i)^{\wt(\vec y)} \zeta_{\vec y}(\vec z) \right\} \cos(l\theta) 
       \label{eq:Q00cosineFourier}
        \\
        &= \sum_{\vec y \in \{0,1\}^{\alpha}} \left\{ (-\i)^{\wt(\vec y)} \cos(l_{\vec y} \theta)\right\} \zeta_{\vec y}(\vec z),
    \end{align}
    where $l_y$ is the unique $l$ for which $\vec y \in \Xi^\alpha_l$.
    
In other words,
\begin{align}
    Q_{00}[\vec z](\theta) = \sum_{l=0}^{\lceil \alpha/2\rceil} a_l(\vec z) \cos(l \theta) =  \sum_{\vec y \in \{0,1\}^{\alpha}} \xi_{\vec y}(\theta) \zeta_{\vec y}(\vec z)
\end{align}
is a 
\begin{enumerate}
    \item[(i)] cosine polynomial in $\theta$ of degree $\alpha$ with Fourier coefficients
    \begin{align}
        a_l(\vec z) = \sum_{\vec y \in \Xi^\alpha_l} (-\i)^{\wt(\vec y)} \zeta_{\vec y}(\vec z).
       \label{eq:alzFourier}
    \end{align}
    \item[(ii)]  trigono-multilinear function in $\vec z$ of arity $2\alpha$ with coefficients
    \begin{align}
        \xi_{\vec y}(\theta) = (-\i)^{\wt(\vec y)} \cos(l_{\vec y}\theta).
    \end{align}
\end{enumerate}
\label{thm:expansionOfQ00}
\end{thm}
Note that uniqueness of $l_y$ follows from the fact that $\{\Xi_l^\alpha:l \leq \lceil \alpha/2\rceil \}$ partitions the set of $\alpha$-bit strings.
\begin{proof} \hfill
\begin{enumerate}
    \item[(i)] Using Eq.~\eqref{eq:ExpansionExpressionForQ}, we find that
    \begin{align} Q_{00}[\vec z](\theta) &=\langle \bar 0|Q(\theta ; \vec z)| \bar 0\rangle \nonumber\\ &=\sum_{k=0}^{\infty} \sum_{u, v \in \mathbb F_{2}} \sum_{y \in \Theta_{ukv}^\alpha}(-\i)^{ \wt(\vec y)} \zeta_{\vec y}(\vec z)\bra{\bar 0} \left(\begin{array}{cc}\cos [(k+v) \theta] & -(-1)^{v} \sin [(k+v) \theta] \\(-1)^{u} \sin [(k+v) \theta] & (-1)^{u+v} \cos [(k+v) \theta]\end{array}\right) \ket{\bar 0}\nonumber\\ &=\sum_{k=0}^{\infty} \sum_{u, v \in \mathbb F_{2}} \sum_{\vec y \in \Theta_{ukv}^\alpha}(-\i)^{ \wt(\vec y)} \zeta_{\vec y}(\vec z) \cos [(k+v) \theta].
    \end{align}              
Denoting $\nu_{\vec y}(\vec z) = (-\i)^{\wt(\vec y)}\zeta_{\vec y}(\vec z)$, the above expression simplifies to
\begin{align} Q_{00}[\vec z](\theta) &=
\sum_{k=0}^{\infty} \sum_{u, v \in \mathbb F_{2}} \left(\sum_{\vec y \in \Theta_{ukv}^\alpha} \nu_{\vec y}(\vec z) \right) \cos [(k+v) \theta] \nonumber\\
&=
\sum_{u \in \mathbb F_{2}}\left[\sum_{k=0}^{\infty} \sum_{\vec y \in \Theta^{\alpha}_{ul0}} \nu_{\vec y}(\vec z) \cos (k \theta)+\sum_{k=0}^{\infty} \sum_{\vec y \in \Theta^{\alpha}_{uk1}} \nu_{\vec y}(\vec z) \cos [(k+1) \theta]\right]
\nonumber \\ 
&=
\sum_{u \in \mathbb F_{2}}\left[\sum_{l=0}^{\infty} \sum_{\vec y \in \Theta^{\alpha}_{ul0}} \nu_{\vec y}(\vec z) \cos (k \theta)+\sum_{l=1}^{\infty} \sum_{y \in \Theta^{\alpha}_{u,l-1,1}} \nu_{\vec y}(\vec z) \cos (l \theta)\right]
\nonumber \\ 
&= 
\sum_{u \in \mathbb F_2}
\left(\sum_{l=0}^{\infty} \sum_{\vec y \in \Theta_{ul0}^{\alpha}}+\sum_{l=1}^{\infty} \sum_{y \in \Theta_{u,l-1,1}^{\alpha}}\right) \nu_{\vec y}(\vec z) \cos (l \theta)
\nonumber\\
&=
\sum_{l=0}^{\infty}\left\{\sum_{u \in \mathbb F_{2}}\left(
\sum_{\vec y\in \Theta^\alpha_{ul0}}+\sum_{\vec y\in \Theta^\alpha_{u,l-1,1}}
\right) \nu_{\vec y}(\vec z)
\right\} \cos(l\theta), \quad \mbox{since $\Theta^\alpha_{u\alpha 0}=\Theta^\alpha_{u,-1,0}=\emptyset$}
\nonumber\\
&= \sum_{l=0}^\infty \left\{ \sum_{\vec y\in \Xi_l^\alpha} \nu_{\vec y}(\vec z) \right\}\cos(l\theta)
\nonumber\\
&= \sum_{l=0}^{\lceil \alpha/2\rceil} \left\{ \sum_{\vec y\in \Xi_l^\alpha} (-\i)^{\wt(\vec y)} \zeta_{\vec y}(\vec z) \right\} \cos(l\theta),
\end{align}
where the last line follows from the fact that 
$\Xi_{l}^{\alpha} \neq \emptyset \iff l \leqslant\lceil\alpha / 2\rceil$ (by Lemma \ref{lem:whenIsXiEmpty}).
\item[(ii)] Let $l_{\vec y}$ be the unique element in $\{l \in \mathbb N:\vec y\in \Xi_l^\alpha$\}. Then,
\begin{align}
    Q_{00}[\vec z](\theta) &= \sum_{l=0}^{\lceil\alpha / 2\rceil} \left\{ 
    \sum_{\vec y\in \Xi_l^\alpha} (-\i)^{\wt(\vec y)} \zeta_{\vec y}(\vec z) \cos(l_{\vec y}\theta)
    \right\}  \nonumber\\
    &= \sum_{\vec y \in \{0,1\}^{\alpha}} \left\{ (-\i)^{\wt(\vec y)} \cos(l_{\vec y} \theta)\right\} \zeta_{\vec y}(\vec z),
\end{align}
where the last line follows from the fact that $$\bigcup_{l=0}^{\lceil\alpha / 2\rceil} \Xi_l^\alpha = \{0,1\}^\alpha.$$
\end{enumerate}
\end{proof}

Next, we will use Theorem \ref{thm:expansionOfQ00} to expand the biases $\Lambda^{\mathcal A}(\theta;\vec x)$. But before we do so, we first prove the following lemma:
\begin{lemma}[Closure under reversal]
    Let $\vec a,\vec c \in \{0,1\}^{2L}$. Then,
    \begin{align}
        \vec c1\vec a^R \in \Xi_l^{4L+1} \iff
        \vec a1\vec c^R \in \Xi_l^{4L+1}.
    \label{eq:a1cR_closure}
    \end{align}
    \label{lem:a1cR_closure}
\end{lemma}
\begin{proof}
    We shall prove the forward direction. 
    By symmetry, 
    the proof of the backward direction is obtained by switching the roles of $a$ and $c$ in the proof.
    
    Consider the case when $l>0$. 
    Let 
    \begin{align}
        &
        \vec c 1\vec a^R \in \Xi_l^{4L+1}=\Theta_{0l0}^{4L+1} \cup \Theta_{1l0}^{4L+1} \cup \Theta_{0,l-1,1}^{4L+1} \cup \Theta_{1,l-1,1}^{4L+1}
        \\
        &\implies
        q^{a_1}p^{a_2}\ldots
        q^{a_{2L+1}}p^{a_{2L}}
        q
        p^{c_{2L}}q^{c_{2L-1}}\ldots
        p^{c_2}q^{c_1} \sim
        (qp)^l \mbox{ or }
        p(qp)^l\mbox{ or }
        (qp)^{l-1}q\mbox{ or }
        p(qp)^{l-1}q ,
    \end{align}
where $\sim$ is the equivalence relation defined in Eq.~\eqref{eq:deftilde}.
    
Since string reversal preserves $\sim$ (see Eq.~\eqref{eq:closureStringReversal}), 
\begin{align}
    \left(q^{a_1}p^{a_2}\ldots
    q^{a_{2L+1}}p^{a_{2L}}
        q
        p^{c_{2L}}q^{c_{2L-1}}\ldots
        p^{c_2}q^{c_1}
        \right)^R
        \sim
        ((qp)^l)^R \mbox{ or }
        (p(qp)^l
        )^R
        \mbox{ or }
       ((qp)^{l-1}q)^R\mbox{ or }
        (p(qp)^{l-1}q)^R.
    \label{eq:stringReversqp}
    \end{align}
    
Now, $p(qp)^l$ and $(qp)^{l-1}q$ are palindromes and mapped to themselves under $(\cdot)^R$; and $(qp)^l$ and $p(qp)^{l-1}q$ are mapped to each other under $(\cdot)^R$. Hence, Eq.~\eqref{eq:stringReversqp} becomes
    \begin{align}
&    q^{c_1}p^{c_2}\ldots
    q^{c_{2L+1}}p^{c_{2L}}
        q
        p^{a_{2L}}q^{a_{2L-1}}\ldots
        p^{a_2}q^{a_1} \sim
        (qp)^l \mbox{ or }
        p(qp)^l\mbox{ or }
        (qp)^{l-1}q\mbox{ or }
        p(qp)^{l-1}q \nonumber\\
&\implies
        \vec a 1\vec c^R \in \Theta_{0l0}^{4L+1} \cup \Theta_{1l0}^{4L+1} \cup \Theta_{0,l-1,1}^{4L+1} \cup \Theta_{1,l-1,1}^{4L+1} =\Xi_l^{4L+1}.
    \end{align}

    The proof for the case $l=0$ proceeds similarly, with the last two clauses of each `or' statement above deleted. 
\end{proof}

This lemma may be used to expand the ancilla-free bias.
\begin{thm}
\label{thm:expanpansionOfDeltaAF}
    Let $\vec x\in \mathbb R^{2L}$. Then,
     \begin{align}
        \Lambda^{\AF}(\theta;\vec x)
        &=
        \sum_{l=0}^{2L+1} \left[\left(
        \raisebox{3mm}{
        $\displaystyle\sum_{
        \underset{
        \wt(\vec c)-\wt(\vec a) \equiv_4 \, 0
        }{
        \vec a1\vec c^R \in \Xi_l^{4L+1}
        }
        }
        -
        \displaystyle\sum_{
        \underset{
        \wt(\vec c)-\wt(\vec a) \equiv_4 \, 2
        }{
        \vec a1\vec c^R \in \Xi_l^{4L+1}
        }
        }$
        }
        \right) \zeta_{\vec a\vec c}(\vec x,\vec x) \right]
        \cos(l\theta) 
        \\
        &=
        \sum_{\vec a,\vec c \in \{0,1\}^{2L}} \left[ \nu_{\wt(\vec c)-\wt(\vec a)} \cos(l_{\vec a\vec c}\theta) \right] \zeta_{\vec a \vec c}(\vec x,\vec x),
        \label{eq:DeltaTrigonomultiquadratic}
    \end{align}
    where 
    $l_{\vec a\vec c}$ is the unique $l$ for which $\vec a 1\vec c^R \in \Xi^{4L+1}_l$ and
    \begin{align}
    \nu_s = \operatorname{Re}(\i^s) = \begin{cases}
    1 & s\equiv 0 \bmod 4,\\
    0 & s\equiv 1 \textup{ or } 3 \bmod 4, \\
    -1 & s\equiv 2 \bmod 4.
    \end{cases}
    \label{eq:nureia0}
\end{align}

In other words,
\begin{align}
    \Lambda^{\AF}(\theta;\vec x)
        =
        \sum_{l=0}^{2L+1} \mu_l^{\AF}(\vec x) \cos(l\theta) =  \sum_{\vec a, \vec c \in \{0,1\}^{2L}} \xi^{\AF}_{\vec a\vec c}(\theta) \zeta_{\vec a \vec c}(\vec x,\vec x)
        \label{eq:biasAsCosineSeriesAndTrigonoquadraticseries}
\end{align}
is a
\begin{enumerate}
    \item[(i)] cosine polynomial in $\theta$ of degree $2L+1$ with Fourier coefficients
    \begin{align}
        \mu_l^{\AF}(\vec x) = 
        \left(
        \raisebox{3mm}{
        $\displaystyle\sum_{
        \underset{
        \wt(\vec c)-\wt(\vec a) \equiv_4 \, 0
        }{
        \vec a1\vec c^R \in \Xi_l^{4L+1}
        }
        }
        -
        \displaystyle\sum_{
        \underset{
        \wt(\vec c)-\wt(\vec a) \equiv_4 \, 2
        }{
        \vec a1\vec c^R \in \Xi_l^{4L+1}
        }
        }$
        }
        \right) \zeta_{\vec a\vec c}(\vec x,\vec x).
        \label{eq:cosPolyCoef}
    \end{align}
    \item[(ii)]  trigono-multiquadratic function in $\vec x$ of arity $2L$ with coefficients
    \begin{align}
        \xi^{\AF}_{\vec a\vec c}(\theta) = \nu_{\wt(\vec c)-\wt(\vec a)}  \cos(l_{\vec a\vec c}\theta).
    \end{align}
\end{enumerate}

In the above, $s \equiv_n t$ means $s \equiv t \pmod{n}$.
\end{thm}

\begin{proof}\hfill
\begin{enumerate}
    \item[(i)] 
By Eq.~\eqref{eq:delta_in_terms_of_Q00_AF} and Theorem \ref{thm:expansionOfQ00},
\begin{align}
\Lambda^{\AF}(\theta; \vec x) 
&= \i Q_{00}\left[\vec x,\tfrac \pi 2, -\vec x^R\right](\theta) \\
&=
\sum_{l=0}^{\lceil (4L+1)/2\rceil} \left\{ \i
\sum_{\vec y\in \Xi_l^{4L+1}}  (-\i)^{\wt(\vec y)} \zeta_{\vec y}\left(\vec x,\tfrac \pi 2, -\vec x^R\right) \right\} \cos(l\theta) \\
&=
\sum_{l=0}^{2L+1} \mu^{\AF}_l(\vec x) \cos(l\theta),
\end{align}
where
\begin{align}
    \mu_l^{\AF}(\vec x) &= 
    \i
\sum_{\vec y\in \Xi_l^{4L+1}}  (-\i)^{\wt(\vec y)} \underbrace{\zeta_{\vec y}\left(\vec x,\tfrac \pi 2, -\vec x^R\right)}_{\Circled{*}}.
\label{eq:muxone}
\end{align}

Writing $\vec y = \vec a b \vec c^R = a_1\ldots a_{2L}b c_{2L}\ldots c_1$, 
where $\vec a, \vec c \in \{0,1\}^{2L}$ and $b \in \{0,1\}$,
\begin{align}
    \Circled{*} &= \zeta_{a_1,\ldots,a_{2L},b,c_{2L},\ldots,c_1} \left(x_1,\ldots,x_{2L},\tfrac\pi 2,-x_{2L},\ldots,x_1\right) \nonumber\\
    &=
    \zeta_{a_1}(x_1)\ldots \zeta_{a_{2L}}(x_{2L})
    \zeta_b(\pi/2)
    \zeta_{c_{2L}}(-x_{2L})
    \ldots
    \zeta_{c_1}(-x_1).
\end{align}

Since $\zeta_b(\pi/2) = b$ and $\zeta_{c_i}(-x)=(-1)^{c_i} \zeta_{c_i}(x)$ for all $i$,
\begin{align}
    \Circled{*} &= (-1)^{c_1+\ldots c_{2L}} b \zeta_{a_1}(x_1) \ldots
    \zeta_{a_{2L}}(x_{2L})
    \zeta_{c_1}(x_1)\ldots
    \zeta_{c_{2L}}(x_{2L}) \nonumber
    \\
    &= (-1)^{\wt(\vec c)} b \zeta_{\vec a}(\vec x)\zeta_{\vec c}(\vec x).
\end{align}

Substituting this back into Eq.~\eqref{eq:muxone} gives
\begin{align}
    \mu^{\AF}_l(\vec x) &= \i
\sum_{\vec ab \vec c^R\in \Xi_l^{4L+1}}  (-\i)^{\wt(\vec a)+\wt(b)+\wt(\vec c)} (-1)^{\wt(\vec c)} b
\zeta_{\vec a}(\vec x)\zeta_{\vec c}(\vec x)
\nonumber\\
&=
\sum_{\vec a 1 \vec c^R\in \Xi_l^{4L+1}}  (-\i)^{\wt(\vec a)} \i^{\wt(\vec c)} \zeta_{\vec a}(\vec x)\zeta_{\vec c}(\vec x)
\nonumber\\
&=
\left( \vphantom{\displaystyle\sum_{
        \underset{
        a<c
        }{
        \vec a1\vec c^R \in \Xi_l^{4L+1}
        }
        }
        }\right.
        \raisebox{3mm}{
        $\underbrace{\displaystyle\sum_{
        \underset{
        a=c
        }{
        \vec a  1\vec c^R \in \Xi_l^{4L+1}
        }
        }
        }_{\Circled{1}}
        +
        \underbrace{\displaystyle\sum_{
        \underset{
        \vec a<\vec c
        }{
        \vec a1\vec c^R \in \Xi_l^{4L+1}
        }
        }
        }_{\Circled{2}}
        +
        \underbrace{\displaystyle\sum_{
        \underset{
        \vec c<\vec a
        }{
        \vec a1\vec c^R \in \Xi_l^{4L+1}
        }
        }
        }_{\Circled{3}}
        $
        }
        \left.
        \vphantom{\displaystyle\sum_{
        \underset{
        \vec a<\vec c
        }{
        \vec a1\vec c^R \in \Xi_l^{4L+1}
        }
        }
        }
        \right)  \underbrace{(-1)^{\wt(\vec a)} \i^{\wt(\vec a)+\wt(\vec c)} \zeta_{\vec a}(\vec x)\zeta_{\vec c}(\vec x)
        }_{(**)},
        \label{eq:mul123}
\end{align}
where $<$ denotes any lexicographical ordering of strings.

First, we calculate
\begin{align}
    \Circled{3}(**)
    &=
    \displaystyle\sum_{
        \underset{
        \vec c<\vec a
        }{
        \vec a1\vec c^R \in \Xi_l^{4L+1}
        }
        }
    (-1)^{\wt(\vec a)} \i^{\wt(\vec a)+\wt(\vec c)} \zeta_{\vec a}(\vec x)\zeta_{\vec c}(\vec x)
    \nonumber\\
&=
    \displaystyle\sum_{
        \underset{
        \vec a<\vec c
        }{
        \vec c1\vec a^R \in \Xi_l^{4L+1}
        }
        }
    (-1)^{\wt(\vec c)} \i^{\wt(\vec a)+\wt(\vec c)} \zeta_{\vec a}(\vec x)\zeta_{\vec c}(\vec x) \qquad\mbox{by interchanging $\vec a$ and $\vec c$}
    \nonumber\\
&=
    \displaystyle\sum_{
        \underset{
        \vec a<\vec c
        }{
        \vec a1\vec c^R \in \Xi_l^{4L+1}
        }
        }
    (-1)^{\wt(\vec c)} \i^{\wt(\vec a)+\wt(\vec c)} \zeta_{\vec a}(\vec x)\zeta_{\vec c}(\vec x) \qquad\mbox{by Lemma \ref{lem:a1cR_closure}.}
\end{align}

Hence,
\begin{align}
    \left(\Circled{2}+\Circled{3}\right)(**)
    &=
    \displaystyle\sum_{
        \underset{
        \vec a<\vec c
        }{
        \vec a1\vec c^R \in \Xi_l^{4L+1}
        }
        }
        \left[(-1)^{\wt(\vec c)+\wt(\vec a)}\right] 
        \i^{\wt(\vec a)+\wt(\vec c)}
        \zeta_{\vec a}(\vec x)
        \zeta_{\vec c}(\vec x)
        \nonumber\\
       &= \displaystyle\sum_{
        \underset{
        \vec a<\vec c
        }{
        \vec a1\vec c^R \in \Xi_l^{4L+1}
        }
        }
\nu_{\wt(\vec c)-\wt(\vec a)}
        \zeta_{\vec a\vec c}(\vec x,\vec x),
\end{align}
where we used the fact that for all $\alpha,\beta \in\mathbb Z$, 
\begin{align}
    [(-1)^\alpha+(-1)^\beta] \i^{\alpha+\beta} =
    2 \nu_{\beta-\alpha},
\end{align}
where $\nu_a$ is given by Eq.~\eqref{eq:nureia0}.

Therefore,
\begin{align}
    \left(\Circled{2}+\Circled{3}\right)(**)
    &=
 2 \left(
 \raisebox{4mm}{$
 \displaystyle\sum_{
        \underset{
       \substack{\vec a<\vec c \\
        \wt(\vec c)-\wt(\vec a) \equiv_4 \, 0
        }
        }{
        \vec a1\vec c^R \in \Xi_l^{4L+1}
        }
        }
        +
     \displaystyle\sum_{
        \underset{
       \substack{\vec a<\vec c \\
        \wt(\vec c)-\wt(\vec a) \equiv_4 \, 1\textup{ or } 3
        }
        }{
        \vec a1\vec c^R \in \Xi_l^{4L+1}
        }
        }
        +
         \displaystyle\sum_{
        \underset{
       \substack{\vec a<\vec c \\
        \wt(\vec c)-\wt(\vec a) \equiv_4 \, 2
        }
        }{
        \vec a1\vec c^R \in \Xi_l^{4L+1}
        }
        }
    $
    }
        \right)
\nu_{\wt(\vec c)-\wt(\vec a)}
        \zeta_{\vec a\vec c}(\vec x,\vec x)
        \nonumber\\
        &=
 2 \left(
 \raisebox{4mm}{$
 \displaystyle\sum_{
        \underset{
       \substack{\vec a<\vec c \\
        \wt(\vec c)-\wt(\vec a) \equiv_4 \, 0
        }
        }{
        \vec a1\vec c^R \in \Xi_l^{4L+1}
        }
        }
        -
     \displaystyle\sum_{
        \underset{
       \substack{\vec a<\vec c \\
        \wt(\vec c)-\wt(\vec a) \equiv_4 \, 2
        }
        }{
        \vec a1\vec c^R \in \Xi_l^{4L+1}
        }
        }
    $
    }
        \right)
        \zeta_{\vec a\vec c}(\vec x,\vec x).
        \label{eq:mul23}
\end{align}

Note that replacing $\vec a<\vec c$ with $\vec c<\vec a$ does not change the value of $\left(\Circled{2}+\Circled{3}\right)(**)$.

Indeed, 
\begin{align}
    \left(\Circled{2}+\Circled{3}\right)(**) &=
     2 \left(
 \raisebox{4mm}{$
 \displaystyle\sum_{
        \underset{
       \substack{\vec a<\vec c \\
        \wt(\vec c)-\wt(\vec a) \equiv_4 \, 0
        }
        }{
        \vec a1\vec c^R \in \Xi_l^{4L+1}
        }
        }
        -
     \displaystyle\sum_{
        \underset{
       \substack{\vec a<\vec c \\
        \wt(\vec c)-\wt(\vec a) \equiv_4 \, 2
        }
        }{
        \vec a1\vec c^R \in \Xi_l^{4L+1}
        }
        }
    $
    }
        \right)
        \zeta_{\vec a\vec c}(\vec x,\vec x) \nonumber\\
        &=
     2 \left(
 \raisebox{4mm}{$
 \displaystyle\sum_{
        \underset{
       \substack{\vec a<\vec c \\
        \wt(\vec c)-\wt(\vec a) \equiv_4 \, 0
        }
        }{
        \vec c1\vec a^R \in \Xi_l^{4L+1}
        }
        }
        -
     \displaystyle\sum_{
        \underset{
       \substack{\vec a<\vec c \\
        \wt(\vec c)-\wt(\vec a) \equiv_4 \, 2
        }
        }{
        \vec c1\vec a^R \in \Xi_l^{4L+1}
        }
        }
    $
    }
        \right)
        \zeta_{\vec a\vec c}(\vec x,\vec x),\quad\mbox{by Lemma \ref{lem:a1cR_closure}}
        \nonumber\\
        &=
             2 \left(
 \raisebox{4mm}{$
 \displaystyle\sum_{
        \underset{
       \substack{\vec c<\vec a \\
        \wt(\vec a)-\wt(\vec c) \equiv_4 \, 0
        }
        }{
        \vec a1\vec c^R \in \Xi_l^{4L+1}
        }
        }
        -
     \displaystyle\sum_{
        \underset{
       \substack{\vec c<\vec a \\
        \wt(\vec a)-\wt(\vec c) \equiv_4 \, 2
        }
        }{
        \vec a1\vec c^R \in \Xi_l^{4L+1}
        }
        }
    $
    }
        \right)
        \zeta_{\vec c\vec a}(\vec x,\vec x),\quad\mbox{by switching the labels $a$ and $c$}
        \nonumber\\
&=                     2 \left(
 \raisebox{4mm}{$
 \displaystyle\sum_{
        \underset{
       \substack{\vec c<\vec a \\
        \wt(\vec c)-\wt(\vec a) \equiv_4 \, 0
        }
        }{
        \vec a1\vec c^R \in \Xi_l^{4L+1}
        }
        }
        -
     \displaystyle\sum_{
        \underset{
       \substack{\vec c<\vec a \\
        \wt(\vec c)-\wt(\vec a) \equiv_4 \, 2
        }
        }{
        \vec a1\vec c^R \in \Xi_l^{4L+1}
        }
        }
    $
    }
        \right)
        \zeta_{\vec a\vec c}(\vec x,\vec x),
        \label{eq:mul32}
\end{align}
where the last line follows from the facts that (i) $t\equiv_4 0\Leftrightarrow -t\equiv_4 0$, (ii) $t\equiv_4 2\Leftrightarrow -t\equiv_4 2$, and (iii) $\zeta_{\vec a\vec c}(\vec x,\vec x)=\zeta_{\vec c\vec a}(\vec x,\vec x)$.

By taking the average of Eqs.~\eqref{eq:mul23} and \eqref{eq:mul32}, we obtain

\begin{align}
    \left(\Circled{2}+\Circled{3}\right)(**)
    =
    \left(
 \raisebox{4mm}{$
 \displaystyle\sum_{
        \underset{
       \substack{\vec c\neq\vec a \\
        \wt(\vec c)-\wt(\vec a) \equiv_4 \, 0
        }
        }{
        \vec a1\vec c^R \in \Xi_l^{4L+1}
        }
        }
        -
     \displaystyle\sum_{
        \underset{
       \substack{\vec c\neq\vec a \\
        \wt(\vec c)-\wt(\vec a) \equiv_4 \, 2
        }
        }{
        \vec a1\vec c^R \in \Xi_l^{4L+1}
        }
        }
    $
    }
        \right)
        \zeta_{\vec a\vec c}(\vec x,\vec x).
        \label{eq:mu23}
\end{align}

Next, we calculate
\begin{align}
    \Circled{1}(**) = \sum_{\vec a 1\vec a^R \in \Xi_l^{4L+1}} \zeta_{\vec a}(\vec x)^2,
\end{align}
which can be expressed as
\begin{align}
    \Circled{1}(**) =     \left(
 \raisebox{4mm}{$
 \displaystyle\sum_{
        \underset{
       \substack{\vec c=\vec a \\
        \wt(\vec c)-\wt(\vec a) \equiv_4 \, 0
        }
        }{
        \vec a1\vec c^R \in \Xi_l^{4L+1}
        }
        }
        -
     \displaystyle\sum_{
        \underset{
       \substack{\vec c=\vec a \\
        \wt(\vec c)-\wt(\vec a) \equiv_4 \, 2
        }
        }{
        \vec a1\vec c^R \in \Xi_l^{4L+1}
        }
        }
    $
    }
        \right)
        \zeta_{\vec a\vec c}(\vec x,\vec x).
         \label{eq:mul1}
\end{align}

Substituting Eq.~\eqref{eq:mu23} and \eqref{eq:mul1} into Eq.~\eqref{eq:mul123} gives
    \begin{align}
        \mu^{\AF}_l(\vec x) = 
        \left(
        \raisebox{3mm}{
        $\displaystyle\sum_{
        \underset{
        \wt(\vec c)-\wt(\vec a) \equiv_4 \, 0
        }{
        \vec a1\vec c^R \in \Xi_l^{4L+1}
        }
        }
        -
        \displaystyle\sum_{
        \underset{
        \wt(\vec c)-\wt(\vec a) \equiv_4 \, 2
        }{
        \vec a1\vec c^R \in \Xi_l^{4L+1}
        }
        }$
        }
        \right) \zeta_{\vec a \vec c}(\vec x,\vec x).
    \end{align}

\item[(ii)] By arranging the terms in the sum differently, we obtain
     \begin{align}
        \Lambda^{\AF}(\theta;\vec x)
        &=
        \sum_{l=0}^{2L+1} \left[\left(
        \raisebox{3mm}{
        $\displaystyle\sum_{
        \underset{
        \wt(\vec c)-\wt(\vec a) \equiv_4 \, 0
        }{
        \vec a1\vec c^R \in \Xi_l^{4L+1}
        }
        }
        -
        \displaystyle\sum_{
        \underset{
        \wt(\vec c)-\wt(\vec a) \equiv_4 \, 2
        }{
        \vec a1\vec c^R \in \Xi_l^{4L+1}
        }
        }$
        }
        \right) \zeta_{\vec a\vec c}(\vec x,\vec x) \right]
        \cos(l\theta) \nonumber\\
        &= \sum_{l=0}^{2L+1} \sum_{\vec a 1 \vec c^R \in \Xi_l^{4L+1}}  \nu_{\wt(\vec c)-\wt(\vec a)}  \zeta_{\vec a \vec c}(\vec x,\vec x) \cos(l_{\vec a \vec c}\theta)
        \nonumber\\
        &= 
        \sum_{\vec a,\vec c \in \{0,1\}^{2L}} \left[ \nu_{\wt(\vec c)-\wt(\vec a)} \cos(l_{\vec a\vec c}\theta) \right] \zeta_{\vec a \vec c}(\vec x,\vec x),
    \end{align}
which completes the proof of the theorem.
\end{enumerate}
\end{proof}

Note that the Fourier coefficients \eqref{eq:cosPolyCoef} can also be expressed as
\begin{align}
    \mu^{\AF}_l(\vec x) = 
    \sum_{a \in \Gamma_l^{2L}} \zeta_{\vec a}(\vec x)^2 + 2\left(\sum_{\vec b \in \Omega_{l,0}^{4L}} - \sum_{\vec b \in \Omega_{l,2}^{4L}}\right) \zeta_{\vec b}(\vec x,\vec x),
    \label{eq:FourierCoefficients2}
\end{align}
where
\begin{align}
    \Omega_{l,K}^{4L} &= \{
    \vec a \vec c \in \{0,1\}^{2L}\times \{0,1\}^{2L}: \vec a 1 \vec c^R \in \Xi_l^{4L+1},\ a<c,\ \wt(c) - \wt(a) \equiv K \bmod 4
    \},
    \\
    \Gamma_l^{2L} &= \{
    \vec a \in \{0,1\}^{2L}: \vec a1\vec a^R \in \Xi_l^{4L+1}
    \}.
    \label{eq:OmegaGammaSets}
\end{align}

Next, we expand the ancilla-based bias.
\begin{thm}
\label{thm:expanpansionOfDeltaAB}
    Let $\vec x\in \mathbb R^{2L}$. Then,
     \begin{align}
        \Lambda^{\AB}(\theta;\vec x)
        &=
        \sum_{l=0}^{L} \left[\left(
        \raisebox{3mm}{
        $\displaystyle\sum_{
        \underset{
        \wt(\vec y) \equiv_4 \, 0
        }{
        \vec y \in \Xi_l^{2L}
        }
        }
        -
        \displaystyle\sum_{
        \underset{
        \wt(\vec y) \equiv_4 \, 2
        }{
        \vec y \in \Xi_l^{2L}
        }
        }$
        }
        \right) \zeta_{\vec y}(\vec x) \right]
        \cos(l\theta) 
        \\
        &=
        \sum_{\vec y \in \{0,1\}^{2L}} \left[ \nu_{\wt(\vec y)} \cos(l_{\vec y}\theta) \right] \zeta_{\vec y}(\vec x)
        \label{eq:DeltaTrigonomultiquadraticAB}
    \end{align}
    where $l_{\vec y}$ is the unique $l$ for which $\vec y \in \Xi^{2L}_l$ and $\nu_s$ is given by Eq.~\eqref{eq:nureia0}.

In other words,
\begin{align}
    \Lambda^{\AB}(\theta;\vec x)
        =
        \sum_{l=0}^{L} \mu_l^{\mathrm{AB}}(\vec x) \cos(l\theta) =  \sum_{\vec y \in \{0,1\}^{2L}} \xi_{\vec y}^{\mathrm{AB}}(\theta) \zeta_{\vec y}(\vec x)
        \label{eq:biasABSeries}
\end{align}
is a 
\begin{enumerate}
    \item[(i)] cosine polynomial in $\theta$ of degree $L$ with Fourier coefficients
    \begin{align}
    \mu_l^{\mathrm{AB}}(\vec x) =
\left(
        \raisebox{3mm}{
        $\displaystyle\sum_{
        \underset{
        \wt(\vec y) \equiv_4 \, 0
        }{
        \vec y \in \Xi_l^{2L}
        }
        }
        -
        \displaystyle\sum_{
        \underset{
        \wt(\vec y) \equiv_4 \, 2
        }{
        \vec y \in \Xi_l^{2L}
        }
        }$
        }
        \right) \zeta_{\vec y}(\vec x).
        \label{eq:cosPolyCoefAB}
    \end{align}
    \item[(ii)]  trigono-multilinear function in $\vec x$ of arity $2L$ with coefficients
    \begin{align}
        \xi^{\mathrm{AB}}_{\vec y}(\theta) = \nu_{\wt(\vec y)}  \cos(l_{\vec y}\theta).
    \end{align}
\end{enumerate}
\end{thm}
\begin{proof}\hfill
\begin{enumerate}
    \item[(i)] 
From Eq.~\eqref{eq:defQ00}, the ancilla-based bias may be written as
\begin{align}
    \Lambda^{\AB}(\theta;\vec x) &= \operatorname{Re} Q_{00}[\vec x](\theta) 
    = \sum_{l=0}^{L} \left\{ \sum_{y\in \Xi_l^{2L}} \operatorname{Re} (-\i)^{\wt(\vec y)} \zeta_{\vec y}(\vec x) \right\} \cos(l\theta),
    \label{eq:Lambdathetax}
\end{align}
which follows from Eq.~\eqref{eq:Q00cosineFourier}. 

Therefore, the Fourier coefficients \eqref{eq:cosPolyCoefAB} can be simplified as
\begin{align}
    \mu_l^{\mathrm{AB}}(\vec x) &=
    \operatorname{Re}  (-\i)^{\wt(\vec y)} \zeta_{\vec y}(\vec z) 
    =
    \sum_{\vec y \in \Xi^{2L}_l} \nu_{\wt(\vec y)} \zeta_{\vec y}(\vec x) \label{eq:muABvwty}
    \\
    &=
    \left(
        \raisebox{3mm}{
        $
    \displaystyle\sum_{
        \underset{
        \wt(\vec y) \equiv_4 \, 0
        }{
        \vec y \in \Xi_l^{2L}
        }
        }
    +
    \displaystyle\sum_{
        \underset{
        \wt(\vec y) \equiv_4 \, 1 \, \mathrm{or} \,  3
        }{
        \vec y \in \Xi_l^{2L}
        }
        }
    +
        \displaystyle\sum_{
        \underset{
        \wt(\vec y) \equiv_4 \, 2
        }{
        \vec y \in \Xi_l^{2L}
        }
        }$
        }
        \right) \nu_{\wt(\vec y)} \zeta_{\vec y}(\vec x) 
        =\left(
        \raisebox{3mm}{
        $\displaystyle\sum_{
        \underset{
        \wt(\vec y) \equiv_4 \, 0
        }{
        \vec y \in \Xi_l^{2L}
        }
        }
        -
        \displaystyle\sum_{
        \underset{
        \wt(\vec y) \equiv_4 \, 2
        }{
        \vec y \in \Xi_l^{2L}
        }
        }$
        }
        \right) \zeta_{\vec y}(\vec x).
\end{align}

\item[(ii)] From Eqs.~\eqref{eq:muABvwty} and \eqref{eq:Lambdathetax},
\begin{align}
    \Lambda^{\AB}(\theta;\vec x) &= \sum_{l=0}^{L}     \sum_{\vec y \in \Xi^{2L}_l} \nu_{\wt(\vec y)} \zeta_{\vec y}(\vec x) \cos(l\theta)
    = \sum_{l=0}^{L}     \sum_{\vec y \in \Xi^{2L}_l} \left\{\nu_{\wt(\vec y)} \zeta_{\vec y}(\vec x) \cos(l_{\vec y} \theta)\right\}
    \nonumber \\
    &= \sum_{\vec y \in \{0,1\}^{2L}} \left[ \nu_{\wt(\vec y)} \cos(l_{\vec y}\theta) \right] \zeta_{\vec y}(\vec x).
\end{align}
\end{enumerate}
\end{proof}

In summary, Theorems \ref{thm:expanpansionOfDeltaAF} and \ref{thm:expanpansionOfDeltaAB} give the expansion formulas for the ancilla-free and ancilla-based biases respectively, which we will use in Section \ref{sec:optimization}. At least two special cases of these expansion formulas yield nice simple expressions and are of interest: (i) when the angles $\vec x$ are chosen to be $(\pi/2)^{2L}$ and (ii) when $L=1$. We study Case (i) in Appendix \ref{sec:Chebyshev} and as mentioned above, Case (ii) in Appendix \ref{sec:L1example}. Finally, in Appendix \ref{sec:leading}, as an application of the expansion formulas, we derive expressions for the leading terms in the cosine expansions \eqref{eq:cosPolyCoef} and \eqref{eq:cosPolyCoefAB}.

\subsection{Special case: Chebyshev likelihood functions}
\label{sec:Chebyshev}

Here we will consider the special case when the tunable parameters $\vec x \in \mathbb R^{2L}$ in Eq.~\eqref{eq:likelihoodInTermsOfBiasSchemes} are chosen to be $\vec x = (\pi/2)^{2L} = (\pi/2,\ldots,\pi/2)$ (i.e.~all the entries of $\vec x$ are chosen to be $\pi/2$) and show
how the results in Sections \ref{sec:quantumGeneratedLikelihoodFunctions} and Appendix \ref{sec:seriesExpansions} simplify in this case. In particular, we shall show that the biases of the likelihood function in both the ancilla-free and ancilla-based schemes can be written in terms of Chebyshev polynomials in $\bra {\bar 0} P(\theta) \ket {\bar 0}$.
Due to this, we refer to the likelihood functions \eqref{eq:likelihoodInTermsOfBiasSchemes} with parameters $\vec x = (\pi/2)^{2L}$ as \textit{Chebyshev likelihood functions} (CLFs).

Setting $\alpha = \beta = \pi/2$ in Eqs.~\eqref{eq:ualpha} and \eqref{eq:vbeta} gives
\begin{align}
    U(\theta; \tfrac \pi 2) &= -\i P(\theta), \nonumber\\
    V(\tfrac \pi 2) &= -\i \bar Z.
\end{align}

Hence, Eq.~\eqref{eq:qalphabeta} evaluates to
\begin{align}
    Q\left(\theta;\left(\tfrac \pi 2\right)^{2L} \right)
    &=
    \left(
    V(\tfrac \pi 2)U(\theta; \tfrac \pi 2)
    \right)^L \nonumber\\
    &= (-1)^L (\bar ZP(\theta))^L \nonumber\\
    &=
    (-1)^L  \begin{pmatrix} \cos L\theta & \sin L\theta \\ -\sin L \theta & \cos L\theta \end{pmatrix},
    \label{eq:Qthetapi2L}
\end{align}
where the matrix above is written in the basis $\{\ket{\bar 0}, \ket{\bar 1} \}$.

By making use of Proposition \ref{prop:expressionsForBiases}, for example, the ancilla-free and the ancilla-based biases may be expressed as
\begin{align}
    \Lambda^\AF\left(\theta;(\tfrac \pi 2)^{2L}
    \right) &= \cos[(2L+1)\theta],
    \nonumber\\
    \Lambda^\AB\left(\theta;(\tfrac \pi 2)^{2L}
    \right) &= (-1)^L \cos(L\theta).
    \label{eq:biasAsCosine}
\end{align}

Hence, by Eq.~\eqref{eq:likelihoodInTermsOfBiasSchemes}
the likelihood functions corresponding to these biases are
\begin{align}
    \mathcal L^\AF(\theta;d,\vec x) &= \frac 12 \left[ 
1+(-1)^d \cos[(2L+1)\theta] \right], \nonumber\\
\mathcal L^\AB(\theta;d,\vec x) &= \frac 12 \left[ 
1+(-1)^{d+L} \cos(L\theta) \right],
\label{eq:chebyLikelihood}
\end{align}

By rewriting $\theta$ in terms of the expectation value $\bra {\bar 0} P(\theta) \ket {\bar 0}$ using Eq.~\eqref{eq:thetaDef}, the above biases can be expressed in terms of Chebyshev polynomials:
\begin{align}
    \Lambda^\AF\left(\theta;(\tfrac \pi 2)^{2L}
    \right) &= \cos[(2L+1) \arccos\bra {\bar 0} P(\theta) \ket {\bar 0}]
    =
    \mathcal T_{2L+1}(\bra {\bar 0} P(\theta) \ket {\bar 0}),
    \nonumber\\
    \Lambda^\AB\left(\theta;(\tfrac \pi 2)^{2L}
    \right) &= (-1)^L \cos(L\arccos\bra {\bar 0} P(\theta) \ket {\bar 0})
    =
    (-1)^L \mathcal T_L(\bra {\bar 0} P(\theta) \ket {\bar 0}),
\end{align}
where 
$\mathcal T_m(x) = \cos(m \arccos x)$, for $|x|\leq 1$, is the $n$th Chebyshev polynomial of the first kind. As explained above, it is for this reason that the likelihood functions in \eqref{eq:chebyLikelihood} are called Chebyshev likelihood functions.

Note that the biases \eqref{eq:biasAsCosine} can trivially be expanded as Fourier series:
\begin{align}
\Lambda^\AF\left(\theta;(\tfrac \pi 2)^{2L}
    \right) &=
\sum_{l=0}^{2L+1} \delta_{l,2L+1} \cos(l \theta), \nonumber\\
\Lambda^\AB\left(\theta;(\tfrac \pi 2)^{2L}
    \right) &=
\sum_{l=0}^L (-1)^L \delta_{l,L} \cos(l \theta).
\end{align}

Hence, the Fourier coefficients given by Eqs.~\eqref{eq:cosPolyCoef} and \eqref{eq:cosPolyCoefAB} can be read off to be
\begin{align}
    \mu_l^\AF\left((\tfrac \pi 2)^{2L}\right) &= \delta_{l,2L+1},
    \nonumber\\
    \mu_l^\AB\left((\tfrac \pi 2)^{2L}\right) &= (-1)^L \delta_{l,L}.
    \label{eq:mu_as_delta_function}
\end{align}

\section{A product-of-sums expansion}
\label{sec:expansionFormulasAppendix}
\subsection{String reductions}

Let $s,t \in \{p,q\}^*$ be strings. We say that $s$ $1$-\textit{reduces} to $t$ if there exist $t_1, t_2 \in \{p,q\}^*$ such that
\begin{enumerate}
    \item $t= t_1 t_2$,
    \item $s = t_1 pp t_2$ or $s = t_1 qq t_2$.
\end{enumerate}
where juxtaposition of strings in the above notation indicates string concatenation.

For $k\geq 2$, we say that s $k$-\textit{reduces} to $t$ if there exist $u_1,u_2,\ldots, u_{k-1} \in \{p,q\}^*$ such that
\begin{enumerate}
    \item $s$ 1-reduces to $u_1$,
    \item $u_i$ 1-reduces to $u_{i+1}$ for all $i=1,\ldots,k-2$,
    \item $u_{k-1}$ 1-reduces to $t$.
\end{enumerate}
Also, we say (trivially) that $s$ $0$-\textit{reduces} to $t$ if $s=t$. 

Say that $s$ \textit{reduces} to $t$ if there exists $k\in \mathbb N$ such that $s$ $k$-reduces to $t$. In other words, $s$ reduces to $t$ if $t$ can be obtained from $s$ by repeatedly deleting substrings $pp$ and $qq$. We write $s \rightarrow t$ if $s$ reduces to $t$. For example, $pppqqq \rightarrow pq$.

We say that $s$ is \textit{irreducible} if there does not exist $r \in \{p,q\}^*$ such that $s\rightarrow r$ and $|r|<|s|$. We will make use of the following notation: for $u\in \mathbb F_2$, let
\begin{align}
    s^u =
    \begin{cases}
        \varepsilon & u=0 \\
        s & u=1.
    \end{cases}
\end{align}
where $\varepsilon$ is the empty string.

Define \begin{align}
    t^{(ukv)}:= p^u (qp)^k q.
\end{align}

It is straightforward to check that $t^{(ukv)}$ satisfies the following properties:
\begin{enumerate}
    \item For all $u, v\in \mathbb F_2$ and for all $k\in \mathbb N$, the string $t^{(ukv)}$ is irreducible.
\item For all $s\in \{p,q\}^*$, there exist unique $u,v \in \mathbb F_2$ and $k\in \mathbb N$ such that $
    s \rightarrow t^{(ukv)}$.
\end{enumerate}
Note that $|t^{(ukv)}| = u+v+2k$. Write 
\begin{align}
\mbox{$s\sim t$, if there exists $u\in \{p,q\}^*$ such that $s\rightarrow u$ and $t\rightarrow u$.     }
\label{eq:deftilde}
\end{align}
It is easy to see that for each $n\in \mathbb N$, the relation $\sim$ is an equivalence relation on the set of strings $\{p,q\}^n$. For example, $pqqpq \sim ppqpp$ since they both reduce to $q$. Note that $\sim$ is preserved by string reversal, i.e.
\begin{align}
    s \sim t \iff s^R \sim t^R.
    \label{eq:closureStringReversal}
\end{align}

Let
\begin{align}
\mathcal M_{ukv} = 
\{
    s \in \{p,q\}^*: s \sim t^{(ukv)}
\}
\end{align}
denote the set of strings $s$ that reduce to $t^{(ukv)}$. Note that $\mathcal M_{ukv}$ forms a partition of $\{p,q\}^*$.

Let
\begin{align}
    r^{(pq)}: \mathbb F_2^n &\rightarrow \{p,q\}^* \nonumber\\
    (x_1,\ldots, x_n) &\mapsto
    r_n r_{n-1} \ldots r_2 r_1
\end{align}
where
\begin{align}
    r_i &= \begin{cases}
    \varepsilon & x_i = 0 \\
    p & x_i=1, i \ \mathrm{even}\\
    q & x_i=1, i \ \mathrm{odd}
    \end{cases}\\
    &=
    \begin{cases}
    p^{x_i} & i \ \mathrm{even},\\
    q^{x_i} & i \ 
    \mathrm{odd}.
    \end{cases}
\end{align}
In other words,
\begin{align}
    r^{(pq)}(x_1,\ldots, x_n) = r^{x_n}\ldots p^{x_2}q^{x_1} ,
\end{align}
where $r=p$ if $k$ is even, and $r=q$ otherwise. Note that $|r^{(pq)}(x)| = \mathrm{wt}(x)$.

\subsection{The set \texorpdfstring{$\Theta_{ukv}^n$}{Thetaukvn}}

We are now ready to define the set $\Theta_{ukv}^n$. Let $n \in \mathbb Z^+$, $k\in\mathbb N$ and $u,v\in\mathbb F_2$. Define 
\begin{align}
    \Theta_{ukv}^n &=   
    \{\vec x \in \mathbb F_2^n: r^{(pq)}(x) \in \mathcal M_{ukv}\} \\
    &=
    \{\vec x\in \mathbb F_2^n:
    r^{(pq)}(x) \rightarrow t^{(ukv)}\}.
\end{align}
In other words, $\Theta_{ukv}^n$ is the set of strings $\vec x = x_1 x_2\ldots x_n \in \mathbb F_2^n$ for which the string $r^{x_n}\ldots p^{x_4}q^{x_3}p^{x_2}q^{x_1}$ reduces to $p^u(qp)^k q^v$, where $r =p$ if $n$ is even and $r=q$ if $n$ is odd.
By convention, for $k\notin \mathbb N$, we take $\Theta^n_{ukv} = \emptyset$.

The following lemma characterizes the set of $k$ values for which $\Theta^n_{ukv}$ is nonempty.
\begin{lemma}
    Let $n\in \mathbb Z^+$, $u,v \in \mathbb F_2$ and $k\in\mathbb N$. Then,
    \begin{align}
        \Theta_{ukv}^n \neq \emptyset \iff
        k &\leq \left\lfloor\frac{n-1}2\right\rfloor - u \mathds 1_{n\in 2\mathbb Z+1}\nonumber\\
        &=
        \begin{cases}
            \tfrac n2-1 & n \textup{ even}, \\
            \tfrac {n-1}2-u & n \textup{ odd}.
        \end{cases}
        \label{eq:evenoddTheta}
    \end{align}
    \label{lem:nonZeroTheta}
\end{lemma}
\begin{proof} We will consider the even and odd cases of $n$ separately.

\vspace{0.2cm}

\noindent\underline{Case 1: $n=2m$ is even, where $m\in \mathbb Z^+$}

\begin{enumerate}
    \item[($\Rightarrow$)]
    
    Assume that $\Theta^n_{ukv} \neq \emptyset$. Then let $\vec x\in \Theta^{2m}_{ukv}$, i.e.
    \begin{align}
      p^{x_{2m}}q^{x_{2m-1}}\ldots p^{x_2}q^{x_1} \rightarrow p^u(qp)^k q^v ,
    \end{align}
which implies that
\begin{align}
    p^{x_{2m}+u}q^{x_{2m-1}}\ldots p^{x_2}q^{x_1+v}\rightarrow (qp)^k.
    \label{eq:case1reduction}
\end{align}
 
\begin{itemize}
    \item First, we show that $k\leq m$. By Eq.~\eqref{eq:case1reduction},
\begin{align}
    & |(qp)^k| \leq  |p^{x_{2m}+u}q^{x_{2,-1}}\ldots p^{x_2}q^{x_1+v}| \\
    & \implies 
    2k \leq 2m \\
    &\implies  
    k \leq m.
\end{align}
\item Second, we show that $k\neq m$. Suppose, for the sake of contradiction, that $k=m$. By Eq.~\eqref{eq:case1reduction},
\begin{align}
    p^{x_2m+u}q^{x_{2m-1}}\ldots p^{x_2}q^{x_1+v} \rightarrow (qp)^m,
\end{align}
which implies that
\begin{align}
    |p^{x_{2m}+u}q^{x_{2m-1}}\ldots p^{x_2}q^{x_1+v}| \geq 2m.
    \label{eq:morethan2m}
\end{align}
But
\begin{align}
    & |p^{x_{2m}+u} q^{x_{2m-1}}\ldots p^{x_2}q^{x_1+v}| 
    \nonumber\\
    &=
    |\{ i \in [2m]: x_{2m}+u=x_1+v = 1, x_{2m=1}=\ldots = x_2 = 1 \}| \nonumber\\
    &\leq 2m.
    \label{eq:lessthan2m}
\end{align}

By Eqs.~\eqref{eq:morethan2m} and \eqref{eq:lessthan2m}, we get
\begin{align}
    |p^{x_{2m}+u} q^{x_{2m-1}}\ldots p^{x_2}q^{x_1+v}|  = 2m,
\end{align}
which implies that
\begin{align}
    x_{2m}\neq u, x_1\neq v, x_{2m-1} = \ldots = x_2 = 1.
\end{align}
Substituting this into 
Eq.~\eqref{eq:case1reduction} gives
\begin{align}
    (pq)^m \rightarrow (qp)^m,
\end{align}
which is a contradiction. Hence, $k\neq m$.
\end{itemize} 
The above two bullet points imply that
\begin{align}
    k\leq m-1 = \frac n2-1.
\end{align}

\item[$(\Leftarrow)$] Let $k \leq \frac n2-1 = m-1 $, and consider
$\vec x = v 1^{2k} u 0^{2(m-k-1)} \in \{0,1\}^{2m}$. Note that $\vec x$ is a well-defined string since $m-k-1 \geq 0$. Also, note that
\begin{align}
    r^{(pq)}(\vec x) = p^u (qp)^k q^v = t^{(ukv)}
\end{align}
is irreducible. Hence, $\vec x \in \Theta_{ukv}^{2m} = \Theta_{ukv}^n$, which implies that
\begin{align}
    \Theta_{ukv}^n \neq \emptyset.
\end{align}
\end{enumerate}

\noindent\underline{Case 2: $n=2m+1$ is odd, where $m\in \mathbb N$}

\begin{enumerate}
    \item[($\Rightarrow$)] Assume that   $\Theta_{ukv}^n \neq \emptyset$. Then there exists some $x \in \Theta^{2m+1}_{ukv}$, i.e.
    \begin{align}
       q^{x_{2m+1}}p^{x_{2m}}\ldots p^{x_2}q^{x_1} \rightarrow
        p^u(qp)^k q^v.
        \label{eq:noddqpqp}
    \end{align}
    
    \begin{itemize}
        \item \underline{Case: $u=0$}
        
        By Eq.~\eqref{eq:noddqpqp},
        \begin{align}
    &\quad q^{x_{2m+1}}p^{x_{2m}}\ldots p^{x_2}q^{x_1} \rightarrow
        (qp)^k q^v \nonumber\\
        &\implies
        q^{x_{2m+1}}p^{x_{2m}}\ldots p^{x_2}q^{x_1+v} \rightarrow
        (qp)^k \nonumber\\
        &\implies
        |q^{x_{2m+1}}p^{x_{2m}}\ldots p^{x_2}q^{x_1+v}| \rightarrow
        |(qp)^k| \nonumber\\
        &\implies
        2m+1 \geq 2k \nonumber\\
        &\implies
        k \leq m+\frac 12.
        \end{align}
        
        Since $m\in \mathbb Z^+$, $k\leq m = \frac{n-1}2$.
        
    \item \underline{Case: $u=1$}
           
        \begin{itemize}
            \item  First, we show that $k\leq m$. By Eq.~\eqref{eq:noddqpqp},
     \begin{align}
    & \phantom{\implies} q^{x_{2m+1}}p^{x_{2m}}\ldots p^{x_2}q^{x_1} \rightarrow
        p(qp)^k q^v \nonumber\\
        &\implies q^{x_{2m+1}}p^{x_{2m}}\ldots p^{x_2}q^{x_1+v} \rightarrow
        p(qp)^k
        \nonumber\\
        &\implies |q^{x_{2m+1}}p^{x_{2m}}\ldots p^{x_2}q^{x_1+v}| \geq
        |p(qp)^k|        \nonumber\\
        &\implies 
        2m+1 \geq
        2k+1\nonumber\\
        &\implies 
        k \leq
        m.
            \end{align}
            
        \item     Second, we show that $k\neq m$. Suppose, for the sake of contradiction, $k=m$. Then,
             By Eq.~\eqref{eq:noddqpqp},
     \begin{align}
    & \phantom{\implies} q^{x_{2m+1}}p^{x_{2m}}\ldots p^{x_2}q^{x_1+v} \rightarrow
        p(qp)^m \label{eq:qpqpqpqkneqm}\\
        &\implies  
    |q^{x_{2m+1}}p^{x_{2m}}\ldots p^{x_2}q^{x_1+v}| = 2m+1,
        \end{align}
        
    which implies that
    \begin{align}
        x_{2m+1}=x_{2m}=\ldots = x_2 = x_1+v = 1.
    \end{align}
    Substituting this into Eq.~\eqref{eq:qpqpqpqkneqm}, we obtain
    \begin{align}
        q(pq)^m \rightarrow p(qp)^m,
    \end{align}
    which is a contradiction. Hence, $k\leq m-1 = \frac{n-1}2 -1$.
        \end{itemize}   
    \end{itemize}
 
    \item[($\Leftarrow$)]

    Let $k \leq \frac {n-1}2 -u = m-u$. Define the string
    \begin{align}
        \vec x = \begin{cases}
        v1^{2k} 0^{2(m-k)} & u = 0 \\
        v1^{2k} u 0^{2(m-k)-1} & u = 1.
        \end{cases}
        \label{eq:vecxstringdef}
    \end{align}
    
    Note that $\vec x$ is a well-defined string in $\{0,1\}^{2m+1}$, i.e. each of the substrings in the definition \eqref{eq:vecxstringdef} have non-negative length. To see this, note that when $u=0$, $m-k \geq m-(m-u) = 0$; and when $u=1$, $2(m-k)-1 \geq 2(m-(m-u))-1 = 2u-1 =1$.
    
    Therefore,
    \begin{align}
        r^{(pq)}(\vec x) &= \begin{cases}
        (qp)^k q^v & u=0 \\
        p(qp)^k q^v & u=1
        \end{cases}
        \\
        &= p^u(qp)^k q^v
        \\
        &\rightarrow p^u (qp)^k q^v.
    \end{align}
    
Hence,
\begin{align}
    \vec x \in \Theta^{2m+1}_{ukv} = \Theta^n_{ukv}.
\end{align}
This implies that $\Theta^n_{ukv} \neq \emptyset$.
\end{enumerate}

\end{proof}

\subsection{Product-of-sums expansion formula}

We start by proving the following lemma.
\begin{lemma}
Let $n \in \mathbb Z^+$ and $P^2 = Q^2 = I$. Let
\begin{align}
    R_y = \begin{cases}
    P & y \ \mathrm{even} \\
    Q & y \ \mathrm{odd}.
    \end{cases}
\end{align}
Then for all $\vec x\in \Theta_{ukv}^n$,
\begin{align}
    R_n^{x_n}R_{n-1}^{x_{n-1}} \ldots R_1^{x_1}= P^u (QP)^k Q^v.
\end{align}
\end{lemma}
\begin{proof}
Consider
\begin{align}
    R_n^{x_n}R_{n-1}^{x_{n-1}} \ldots R_1^{x_1} =  R_n^{x_n}\ldots P^{x_4}Q^{x_3}P^{x_2}Q^{x_1} = r^{(PQ)}(\vec x).
\end{align}
\end{proof}

\begin{lemma}
Let $n \in \mathbb Z^+$ and $P^2 = Q^2 = I$.
Let $\{a_x^y:x \in \mathbb F_2$, $y \in [n] \} \subset \mathbb C$. Then,
\begin{align}
    &\left(a_0^n+a_1^n R\right)\ldots
    \left(a_0^4+a_1^4 P\right)
    \left(a_0^3+a_1^3 Q\right)
    \left(a_0^2+a_1^2 P\right)
    \left(a_0^1+a_1^1 Q\right) \nonumber\\
    &\qquad =
    \sum_{k=0}^\infty \sum_{u,v\in \mathbb F_2}\left( 
    \sum_{\vec x \in \Theta_{ukv}^n} a_{x_1}^1 a_{x_2}^2 \ldots a_{x_n}^n
    \right)
    P^u (QP)^k Q^v
    \label{eq:keyExpansionLemma}
\end{align}
where $R = P$ if $n$ is even and $R=Q$ if $n$ is odd.
\label{lem:keyExpansionLemma}
\end{lemma}
\begin{proof}
Consider
\begin{align}
    &\left(a_0^n+a_1^n R\right)\ldots
    \left(a_0^4+a_1^4 P\right)
    \left(a_0^3+a_1^3 Q\right)
    \left(a_0^2+a_1^2 P\right)
    \left(a_0^1+a_1^1 Q\right) \nonumber\\
    &=
    \prod_{y=n}^1(a_0^y+a_1^y R_y), \quad R_y = \begin{cases}
    P & y\mbox{ even}\\
    Q & y\mbox{ odd}
    \end{cases}
    \nonumber\\
    &= 
    \prod_{y=n}^1 \sum_{x\in \mathbb F_2} a_x^y R_y^x \nonumber\\
    &= \left(\sum_{x_n \in \mathbb F_2} a_{x_n}^n R^{x_n}_n\right)\ldots \left(\sum_{x_1 \in \mathbb F_2} a_{x_1}^1 R^{x_1}_1\right) \nonumber\\
    &=  \sum_{\vec x\in \mathbb F_2^n} a_{x_1}^1 \ldots a_{x_n}^n R_n^{x_n}\ldots R_x^{x_1} \nonumber\\
    &= \sum_{k=0}^\infty \sum_{u,v\in \mathbb F_2} 
    \sum_{\vec x \in \Theta_{ukv}^n} a_{x_1}^1 a_{x_2}^2 \ldots a_{x_n}^n
    R_n^{x_n} \ldots R_1^{x_1} \nonumber\\
    &= \sum_{k=0}^\infty \sum_{u,v\in \mathbb F_2}\left( 
    \sum_{\vec x \in \Theta_{ukv}^n} a_{x_1}^1 a_{x_2}^2 \ldots a_{x_n}^n
    \right)
    P^u (QP)^k Q^v.
\end{align}
\end{proof}

Note that by Lemma \ref{lem:nonZeroTheta}, Eq.~\eqref{eq:keyExpansionLemma} can also be written as
\begin{align}
    &\left(a_0^n+a_1^n R\right)\ldots
    \left(a_0^4+a_1^4 P\right)
    \left(a_0^3+a_1^3 Q\right)
    \left(a_0^2+a_1^2 P\right)
    \left(a_0^1+a_1^1 Q\right) \nonumber\\
    &\qquad =
    \sum_{u,v\in \mathbb F_2}
    \sum_{k=0}^{
    \lfloor (n-1)/2 \rfloor - u \mathds 1_{n\in 2\mathbb Z +1}
    }
    \left( 
    \sum_{\vec x \in \Theta_{ukv}^n} a_{x_1}^1 a_{x_2}^2 \ldots a_{x_n}^n
    \right)
    P^u (QP)^k Q^v.
    \label{eq:keyExpansionLemma2}
\end{align}

\section{Cardinalities of
\texorpdfstring{$\Theta^n_{ukv}$}{Thetaukv} and \texorpdfstring{$\Xi_l^\alpha$}{Xi}}
\label{sec:propertiesOfXiAppendix}

What is the space complexity of storing each of the Fourier coefficients, of say, Eq.~\eqref{eq:alzFourier} and \eqref{eq:cosPolyCoef}? To address this question, it suffices to find the cardinality of the sets $\Xi_l^\alpha$.

To begin, we prove the following theorem, which gives the cardinality of the set $\Theta^n_{ukv}$.
\begin{thm}
    Let $n \in \mathbb Z^+$, $u,v \in \mathbb F_2$, $k \in\{
    0,1,2,\ldots, \left\lfloor (n-1)/2\right\rfloor - u \mathds 1_{n \in 2\mathbb Z+1}
    \}$. Then,
    \begin{align}
    \left|\Theta_{ukv}^{n}\right|
    =\left(\begin{array}{c} n-1 \\ \left\lfloor \frac{n-1}2\right\rfloor - u \mathds 1_{n \in 2\mathbb Z+1} - k
    \end{array}\right),
    \end{align}
i.e. for $m\in \mathbb Z^+$,
\begin{align}
    \left|\Theta_{ukv}^{2 m}\right|
&=\binom{2m-1}{m-1-k}= \binom{2m-1}{m+k}
,
\\
    \left|\Theta_{ukv}^{2 m+1}\right|
&= \binom{2m}{m-u-k}= \binom{2m}{m+u+k}.
\end{align}
\end{thm}

\begin{proof} We consider the odd and even cases separately.

\begin{itemize}
        \item \underline{Case: $n=2m$ even}

We first prove the following identity: for $m \in \mathbb Z^+$, if $P^2 = Q^2 = I$, then
\begin{align}
    \left[
        (I+P)(I+Q)
    \right]^m
    =
    \sum_{k=0}^{m-1} \binom{2m-1}{m+k}\left[
        \sum_{u,v\in \mathbb F_2} P^u(QP)^k Q^v
    \right].
    \label{eq:mathindeven}
\end{align}

To prove Eq.~\eqref{eq:mathindeven}, we make use of mathematical induction. The base case $m=1$ is clearly true. Assume that Eq.~\eqref{eq:mathindeven} holds for $m$. Then, by the induction hypothesis,
\begin{align}
    \left[
        (I+P)(I+Q)
    \right]^{m+1} &= (I+P+Q+PQ)\sum_{k=0}^{m-1} \binom{2m-1}{m+k}\left[
        \sum_{u,v\in \mathbb F_2} P^u(QP)^k Q^v
    \right]
    \nonumber\\
    &=
    \sum_{k=0}^{m-1}\binom{2m-1}{m+k} (I+P+Q+PQ)\left[(QP)^k+P(QP)^k+(QP)^k Q + P(QP)^k Q\right] \nonumber\\
    &=
    \sum_{k=1}^{m-1} 2 \binom{2m-1}{m+k} \left((QP)^k + P(QP)^k + (QP)^k Q + P(QP)^k Q\right) \nonumber\\
    &\quad +
    \sum_{k=1}^{m-1}  \binom{2m-1}{m+k} \left((QP)^{k-1} + P(QP)^{k-1} + (QP)^{k-1} Q + P(QP)^{k-1} Q\right) \nonumber\\
    &\quad +
    \sum_{k=1}^{m-1}  \binom{2m-1}{m+k} \left((QP)^{k+1} + P(QP)^{k+1} + (QP)^{k+1} Q + P(QP)^{k+1} Q\right) \nonumber\\
    &\quad +
    \binom{2m-1}{m} \left[
    3(I+P+Q+PQ)
    +
    (QP+PQP+QPQ+PQPQ)
    \right]\nonumber\\
    &=
    \sum_{k=0}^{m-2} \left[
    2\binom{2m-1}{m+k}+
    \binom{2m-1}{m+k+1}
    +
    \binom{2m-1}{m+k-1}
    \right]
    \nonumber\\
    &\qquad\times\left((QP)^k+P(QP)^k + (QP)^k Q + P(QP)^k Q\right)
    \nonumber \\
    &\quad +
    (2n+3)\left((QP)^{m-1}+P(QP)^{m-1}+(QP)^{m-1}Q + P(QP)^{m-1}Q\right) \nonumber\\
    &\quad + (QP)^m + P(QP)^m + (QP)^m Q + P(QP)^m Q.
\end{align}
By Pascal's rule,
\begin{align}
2\binom{2m-1}{m+k}+
    \binom{2m-1}{m+k+1}
    +
    \binom{2m-1}{m+k-1}
    =\binom{2m+1}{m+k+1},    
\end{align}
and hence,
\begin{align}
    \left[
        (I+P)(I+Q)
    \right]^{m+1} &=
    \sum_{k=0}^m \binom{2m+1}{m+k+1} ((QP)^k + P(QP)^k + (QP)^k Q + P(QP)^k Q) \nonumber\\
    &=
    \sum_{k=0}^m \binom{2m+1}{m+1+k}\left[
        \sum_{u,v\in \mathbb F_2} P^u(QP)^k Q^v
    \right],
\end{align}
which completes the inductive step and the proof of Eq.~\eqref{eq:mathindeven}.

Next, by setting $n=2m$ and $a_y^x=1$ for all $x,y$ in Eq.~\eqref{eq:keyExpansionLemma2}, we get
\begin{align}
    \left[
        (I+P)(I+Q)
    \right]^m
    =
    \sum_{k=0}^{m-1} 
    \left[
        \sum_{u,v\in \mathbb F_2} \left|
    \Theta^{2m}_{ukv} \right| P^u(QP)^k Q^v
    \right].
    \label{eq:mathindeven1}
\end{align}

Comparing Eqs.~\eqref{eq:mathindeven}
and \eqref{eq:mathindeven1} gives

\begin{align}
    \left|\Theta_{ukv}^{2 m}\right|
= \binom{2m-1}{m+k}.
\end{align}

 \item \underline{Case: $n=2m+1$ odd}

We first prove the following identity: for $m \in \mathbb N$, if $P^2 = Q^2 = I$, then
\begin{align}
    (I+Q)\left[
        (I+P)(I+Q)
    \right]^m
    =
    \sum_{u,v\in \mathbb F_2}
    \sum_{k=0}^{m-u} \binom{2m}{m+u+k}
         P^u(QP)^k Q^v
    .
    \label{eq:mathindodd}
\end{align}

To prove Eq.~\eqref{eq:mathindodd}, our starting point is Eq.~\eqref{eq:mathindeven}. By multiplying $I+Q$ on the left of Eq.~\eqref{eq:mathindeven}, we obtain
\begin{align}
    & (I+Q)\left[
        (I+P)(I+Q)
    \right]^m \nonumber\\
    &=
    (I+Q)\sum_{k=0}^{m-1} \binom{2m-1}{m+k}\left[
        \sum_{u,v\in \mathbb F_2} P^u(QP)^k Q^v
    \right] \nonumber\\
    &= \sum_{k=0}^{m-1} \binom{2m-1}{m+k}\left[
        \sum_{u,v\in \mathbb F_2} P^u(QP)^k Q^v
    \right] +
    \underbrace{\sum_{k=0}^{m-1} \binom{2m-1}{m+k}
        \underbrace{\sum_{u,v\in \mathbb F_2} QP^u(QP)^k Q^v}_{\mathscr{J}_k}
        }_{
        \Circled{1}}
       . \label{eq:IQIPIQm}
\end{align}

Now, by expanding the sum $\mathscr{J}_k$, we obtain
\begin{align}
    \mathscr{J}_k=
    \sum_{v \in \mathbb F_2} Q(QP)^k Q^v + (QP)^{k+1} Q^v.
    \label{eq:mathscrJk}
\end{align}

When $k=0$, Eq.~\eqref{eq:mathscrJk} evaluates to
\begin{align}
    \mathscr{J}_0=
    \sum_{v \in \mathbb F_2} Q^v + (QP)Q^v.
\end{align}

When $k\geq 1$, Eq.~\eqref{eq:mathscrJk} evaluates to
\begin{align}
    \mathscr{J}_k=
    \sum_{v \in \mathbb F_2} P(QP)^{k-1}Q^v + (QP)^{k+1}Q^v.
\end{align}

Substituting these expressions for $\mathscr{J}_k$ into $\Circled{1}$ gives
\begin{align}
    \Circled{1} &= \sum_{k=0}^{m-1}\binom{2m-1}{m+k} \mathscr{J}_k \nonumber\\
    &=
    \binom{2m-1}{m} \mathscr{J}_0 + \sum_{k=1}^{m-1}\binom{2m-1}{m+k} \mathscr{J}_k \nonumber\\
    &= \binom{2m-1}{m} \sum_{v \in \mathbb F_2} [Q^v + (QP)Q^v] + \sum_{k=1}^{m-1}\binom{2m-1}{m+k} \sum_{v \in \mathbb F_2} [P(QP)^{k-1}Q^v + (QP)^{k+1}Q^v] \nonumber\\
    &=
    \sum_{v \in \mathbb F_2}
    \left\{
    \binom{2m-1}{m}(Q^v + QPQ^v)
    +\sum_{k=0}^{m-2}
    \binom{2m-1}{m+k+1} P(QP)^k Q^v
    +
    \sum_{k=2}^m \binom{2m-1}{m+k-1} (QP)^k Q^v
    \right\} \nonumber\\
    &=
    \sum_{v \in \mathbb F_2}
    \left[
    \sum_{k=0}^{m-2}
    \binom{2m-1}{m+k+1} P(QP)^k Q^v
    +
    \sum_{k=0}^m \binom{2m-1}{m+k-1} (QP)^k Q^v
    \right],
    \label{eq:circledones}
\end{align}
where the last line follows from the fact that
\begin{align}
    \binom{2m-1}{m}(Q^v + QPQ^v) = \binom{2m-1}{m-1} (QP)^0 Q^v + 
    \binom{2m-1}{m} (QP) Q^v.
\end{align}

Substituting Eq.~\eqref{eq:circledones} into Eq.~\eqref{eq:IQIPIQm} gives
\begin{align}
    & (I+Q)\left[
        (I+P)(I+Q)
    \right]^m \nonumber\\
    &= 
    \sum_{v \in \mathbb F_2} \left\{
    \sum_{k=0}^{m-1} \binom{2m-1}{m+k}\left[
        (QP)^k Q^v + P(QP)^k Q^v
    \right] \right. \nonumber\\
    &\quad + \left.
    \sum_{k=0}^{m-2}
    \binom{2m-1}{m+k+1} P(QP)^k Q^v
    +
    \sum_{k=0}^m \binom{2m-1}{m+k-1} (QP)^k Q^v
    \right\} \nonumber\\
    &= 
    \sum_{v \in \mathbb F_2} \left\{
    \sum_{k=0}^{m-1}
    \left[\binom{2m-1}{m+k}
    + \binom{2m-1}{m+k-1}
    \right]P(QP)^k Q^v
     \right. \nonumber\\
     &\quad +
    \left. \sum_{k=0}^{m}
    \left[\binom{2m-1}{m+k}
    + \binom{2m-1}{m+k+1}
    \right]P(QP)^k Q^v
     \right\} \nonumber\\
    &= \sum_{v \in \mathbb F_2} \left\{
    \sum_{k=0}^{m-1}
    \binom{2m}{m+k}
    P(QP)^k Q^v
     +
    \sum_{k=0}^{m}
    \left[\binom{2m}{m+k+1}
    \right]P(QP)^k Q^v
     \right\} \nonumber\\
    &= \sum_{u,v \in \mathbb F_2}
    \sum_{k=0}^{m-u}
    \binom{2m}{m+u+k}
    P^u(QP)^k Q^v,
\end{align}
which completes the proof of Eq.~\eqref{eq:mathindodd}.

Next, by setting $n=2m+1$ and $a_y^x=1$ for all $x,y$ in Eq.~\eqref{eq:keyExpansionLemma2}, we get
\begin{align}
    (I+Q)\left[
        (I+P)(I+Q)
    \right]^m
    =
    \sum_{u,v\in \mathbb F_2}
    \sum_{k=0}^{m-u}\left| 
    \Theta_{ukv}^{2m+1}
    \right|
         P^u(QP)^k Q^v
    .
    \label{eq:mathindodd1}
\end{align}

\end{itemize}

Comparing Eqs.~\eqref{eq:mathindodd}
and \eqref{eq:mathindodd1} gives

\begin{align}
    \left|\Theta_{ukv}^{2 m+1}\right|
= \binom{2m}{m+u+k}.
\end{align}

\end{proof}

The following theorem gives the cardinality of the set $\Xi_l^\alpha$.
\begin{thm}
    Let $\alpha \in \mathbb Z^+$ and $l\in \mathbb N$. Then,
    \begin{align}
    \left|\Xi_{l}^{2 m}\right|
    =\left\{\begin{array}{ll}2\left(\begin{array}{c}\alpha-\delta_l \\ \frac{\alpha}2 - l - \delta_l \end{array}\right) & \alpha \ \mathrm{ even}, \\[10pt] \left(\begin{array}{c}\alpha+1-\delta_l \\ \frac{\alpha+1}2 - l-\delta_l
    \end{array}\right) & \alpha \ \mathrm{ odd,}
    \end{array}\right.
\end{align}
i.e. for $m\in \mathbb Z^+$,
\begin{align}
    \left|\Xi_{l}^{2 m}\right|
    =\left\{\begin{array}{ll}2\left(\begin{array}{c}2 m-1 \\ m-1\end{array}\right) & l=0, \\[12pt] 2\left(\begin{array}{c}2 m \\ m-l\end{array}\right) & l \geq 1.
    \end{array}\right.
\end{align}
\begin{align}
    \left|\Xi_{l}^{2 m+1}\right|
    =\left\{\begin{array}{ll}\left(\begin{array}{c}2 m+1 \\ m\end{array}\right) & l=0, \\[10pt] 2\left(\begin{array}{c}2 m+2 \\ m-l+1\end{array}\right) & l \geq 1
    .
    \end{array}\right.
\end{align}
\end{thm}
\begin{proof}
We first consider the case when $\alpha = 2m$ is even.
Since the sets in the union represented by Eq.~\eqref{eq:XiAsUnionThetas} are disjoint, it follows that
\begin{align}
    \left|\Xi_l^{2m}\right|=
    \underbrace{\left|\Theta_{0l0}^{2 m}\right|+\left|\Theta_{1l0}^{2 m}\right|}_{\Circled{1}}+
    \underbrace{\left|\Theta_{0,l-1,1}^{2 m}\right|+\left|\Theta_{1,l-1,1}^{2 m}\right|}_{\Circled{2}}.
\end{align}

Here,
\begin{align}
    \Circled{1} = \binom{2m-1}{m-1-l}+
    \binom{2m-1}{m-1-l} = 2\binom{2m-1}{m-1-l} 
\end{align}
and
\begin{align}
    \Circled{2} = \left[\binom{2m-1}{m-l}+
    \binom{2m-1}{m-l}
    \right] \mathds 1_{l\geq 1}
    = 2\binom{2m-1}{m-l} \mathds 1_{l\geq 1} .
\end{align}

Therefore, if $l=0$,
\begin{align}
    \left|\Xi_l^{2m}\right|=2\binom{2m-1}{m-1}.
\end{align}

And if $l \geq 1$,
\begin{align}
    \left|\Xi_l^{2m}\right|
    &=2\left[\binom{2m-1}{m-1-l}+
    \binom{2m-1}{m-l}
    \right] \nonumber\\
    &=
    2 \binom{2m}{m-l},
\end{align}
by Pascal's rule. Next, we consider the case when $\alpha = 2m+1$ is odd.
As before, since the sets in the union represented by Eq.~\eqref{eq:XiAsUnionThetas} are disjoint, it follows that
\begin{align}
    \left|\Xi_l^{2m+1}\right|=
    \underbrace{\left|\Theta_{0l0}^{2 m+1}\right|+\left|\Theta_{1l0}^{2 m+1}\right|}_{\Circled{3}}+
    \underbrace{\left|\Theta_{0,l-1,1}^{2 m+1}\right|+\left|\Theta_{1,l-1,1}^{2 m+1}\right|}_{\Circled{4}}.
\end{align}

Here,
\begin{align}
    \Circled{3} = \binom{2m}{m-l}+
    \binom{2m}{m-1-l} = \binom{2m+1}{m-l} 
\end{align}
and
\begin{align}
    \Circled{4} = \left[\binom{2m}{m-l+1}+
    \binom{2m}{m-l}
    \right] \mathds 1_{l\geq 1}
    = \binom{2m+1}{m-l+1} \mathds 1_{l\geq 1}.
\end{align}

Therefore, if $l=0$,
\begin{align}
    \left|\Xi_l^{2m}\right|=\binom{2m+1}{m}.
\end{align}
And if $l \geq 1$,
\begin{align}
    \left|\Xi_l^{2m}\right|
    &=\binom{2m+1}{m-l}+
    \binom{2m+1}{m-l+1}
     \nonumber\\
    &=
    \binom{2m+2}{m-l+1},
\end{align}
which completes the proof of the theorem.
\end{proof}

\section{Trigono-multivariate polynomial functions}
\label{sec:trigonofunctions}

Let $k,d\in \mathbb Z^+$. A $k$-ary function $f:\mathbb R^k \rightarrow \mathbb C$ is a \textit{trigono-multivariate polynomial function} of degree $d$ if for all $\vec y\in \{0,1\}^{kd}$, there exists $\xi_{\vec y} \in \mathbb C$ such that for all $\vec x \in \mathbb R^k$,
\begin{align}
    f(\vec x) = \sum_{\vec y \in \{0,1\}^{kd}} \xi_{\vec y} \zeta_{\vec y}(\underbrace{\vec x, \vec x, \ldots, \vec x}_{d \textup{ times}}).
    \label{eq:trigonomultivariate}
\end{align}
Denote the set of $k$-ary trigono-multivariate polynomial functions of degree $d$ by $\mathscr{T}^k_d$.

Note that the definition \eqref{eq:trigonomultivariate} generalizes the notions of trigono-multilinearity and trigono-multiquadraticity:
from Eqs.~\eqref{eq:tmlf} and \eqref{eq:tmqf}, $\mathscr{T}^k_1$ and $\mathscr{T}^k_2$ are the sets of $k$-ary trigono-multilinear and trigono-multiquadratic functions respectively. 

It is easy to check that the sets $\mathscr{T}^k_d$ satisfy the following simple closure properties.
\begin{proposition}
Let $k,d,e\in \mathbb Z^+$. The sets $\mathscr{T}^k_d$ satisfy the following properties.
\begin{enumerate}
    \item If $f,g \in \mathscr{T}^k_d$, then $f+g \in \mathscr{T}^k_d$. 
    \item If $f \in \mathscr{T}^k_d$ and $g \in \mathscr{T}^k_e$, then $fg \in \mathscr{T}^k_{d+e}$.
\end{enumerate}
In particular, the product of two trigono-multilinear functions is a trigono-multiquadratic function.
\label{prop:closureProp}
\end{proposition}

Next we describe some properties of trigono-multilinear and trigono-multiquadratic functions. Recall that we have defined these functions by expressing them as the sum of exponentially many terms as in Eqs.~\eqref{eq:tmlf} and \eqref{eq:tmqf}, respectively. But sometimes it is more convenient to work with the following equivalent definitions of these functions.

\begin{proposition}
Let $k \in \mathbb{Z}^+$. A $k$-ary function $f: \mathbb{R}^k \to \mathbb{C}$ is trigono-multilinear if and only if for all $j \in [k]$, there exist $(k-1)$-ary functions $C_j, S_j: \mathbb{R}^{k-1} \to \mathbb{C}$ such that
\begin{align}
f(\vec x) = C_j(\vec x_{\neg j}) \cosp{x_j} + S_j(\vec x_{\neg j}) \sinp{x_j},
\label{eq:def_tmlf_2}
\end{align}
where $\vec x=(x_1, \dots, x_k)$ and $\vec x_{\neg j}=(x_1,\dots, x_{j-1}, x_{j+1}, \dots, x_k)$. We call $C_j$ and $S_j$ the cosine-sine-decomposition (CSD) functions of $f$ with respect to $x_j$.
\label{prop:tmlf}
\end{proposition}
\begin{proof}
The necessity of the given condition is easy to prove. Suppose $f(\vec x)=\sum_{\vec y \in \{0, 1\}^k} \xi_{\vec y} \zeta_{\vec y}(\vec x)$. Then we set 
\begin{align}
C_j(\vec x_{\neg j}) &= \sum_{\vec z \in \{0, 1\}^{k-1}} \xi_{\vec z \circ 0} \zeta_{\vec z}(\vec x_{\neg j}), \\
S_j(\vec x_{\neg j}) &= \sum_{\vec z \in \{0, 1\}^{k-1}} \xi_{\vec z \circ 1} \zeta_{\vec z}(\vec x_{\neg j}),    
\end{align}
where $\vec z \circ 0=z_1\dots z_{k-1} 0$ and $\vec z \circ 1=z_1\dots z_{k-1} 1$, and obtain Eq.~\eqref{eq:def_tmlf_2}.

Next, we prove the sufficiency of the given condition by induction on $k$. It is obvious for $k=1$. Suppose $f$ satisfies the given condition, i.e. $f(\vec x) = C_j(\vec x_{\neg j}) \cosp{x_j} + S_j(\vec x_{\neg j}) \sinp{x_j}$ for some $C_j$ and $S_j$, for all $j \in [k]$. Let $\vec w=(w_1,\dots,w_{k-1}) \in \mathbb{R}^{k-1}$ be arbitrary. Let $\vec a=\vec w \circ 0=(w_1,\dots,w_{k-1},0)$ and $\vec b=\vec w \circ \pi/2=(w_1,\dots,w_{k-1},\pi/2)$. Then we
have
\begin{align}
C_k(\vec w) = f(\vec a) = \hat{C}_j(\vec w_{\neg j}) \cosp{w_j} + \hat{S}_j(\vec w_{\neg j}) \sinp{w_j},~~~\forall j \in [k-1],
\end{align}
where $\hat{C}_j(\vec w_{\neg j})=C_j(\vec a_{\neg j})$ and $\hat{S}_j(\vec w_{\neg j})=S_j(\vec a_{\neg j})$, and 
\begin{align}
S_k(\vec w) = f(\vec b) = \bar{C}_j(\vec w_{\neg j}) \cosp{w_j} + \bar{S}_j(\vec w_{\neg j}) \sinp{w_j},~~~\forall j \in [k-1].
\end{align}
where $\bar{C}_j(\vec w_{\neg j})=C_j(\vec b_{\neg j})$ and $\bar{S}_j(\vec w_{\neg j})=S_j(\vec b_{\neg j})$. This means that both $C_k$ and $S_k$ satisfy the given condition for $(k-1)$-ary functions. So by induction hypothesis, we know that both $C_k$ and $S_k$ are $(k-1)$-ary trigono-multilinear functions, i.e. they can be expressed as in Eq.~\eqref{eq:tmlf}. It follows that $f(\vec x) = C_k(\vec x_{\neg k}) \cosp{x_k} + S_k(\vec x_{\neg k}) \sinp{x_k}$ can be also expressed as in Eq.~\eqref{eq:tmlf}, i.e.~it is a $k$-ary trigono-multilinear function.  
\end{proof}

\begin{proposition}
Let $k \in \mathbb{Z}^+$. A $k$-ary function $f: \mathbb{R}^k \to \mathbb{C}$ is trigono-multiquadratic if and only if for all $j \in [k]$, there exist $(k-1)$-ary functions $C_j, S_j, B_j: \mathbb{R}^{k-1} \to \mathbb{C}$ such that
\begin{align}
f(\vec x) = C_j(\vec x_{\neg j}) \cosp{2x_j} + S_j(\vec x_{\neg j}) \sinp{2x_j} + B_j(\vec x_{\neg j}),
\end{align}
where $\vec x=(x_1, \dots, x_k)$ and $\vec x_{\neg j}=(x_1,\dots, x_{j-1}, x_{j+1}, \dots, x_k)$. We call $C_j$, $S_j$ and $B_j$ the cosine-sine-bias-decomposition (CSBD) functions of $f$ with respect to $x_j$.
\label{prop:tmqf}
\end{proposition}
\begin{proof}
Since $\cosp{2x}=\cospt{x}-\sinpt{x}$, $\sinp{2x}=2 \cosp{x}\sinp{x}$ and $\cospt{x}+\sinpt{x}=1$, it suffices to show that $f$ is trigono-multilquadratic if and only if for all $j \in [k]$, there exist $(k-1)$-ary functions $E_j, F_j, G_j: \mathbb{R}^{k-1} \to \mathbb{C}$ such that
\begin{align}
f(\vec x) = E_j(\vec x_{\neg j}) \cospt{x_j} + F_j(\vec x_{\neg j}) \sinpt{x_j} + G_j(\vec x_{\neg j}) \cosp{x_j} \sinp{x_j},
\label{eq:def_tmqf_2}
\end{align}
where $\vec x=(x_1, \dots, x_k)$ and $\vec x_{\neg j}=(x_1,\dots, x_{j-1}, x_{j+1}, \dots, x_k)$. 

The necessity of this condition is easy to prove. Suppose $f(\vec x)=\sum_{\vec y, \vec z \in \{0, 1\}^k} \xi_{\vec y \vec z} \zeta_{\vec y \vec z}(\vec x, \vec x)$. Then we set 
\begin{align}
E_j(\vec x_{\neg j}) &= \sum_{\vec u, \vec v \in \{0, 1\}^{k-1}} \xi_{\vec u \circ 0, \vec v \circ 0} \zeta_{\vec u, \vec v}(\vec x_{\neg j}, \vec x_{\neg j}), \\
F_j(\vec x_{\neg j}) &= \sum_{\vec u, \vec v \in \{0, 1\}^{k-1}} \xi_{\vec u \circ 1, \vec v \circ 1} \zeta_{\vec u, \vec v}(\vec x_{\neg j}, \vec x_{\neg j}), \\ 
G_j(\vec x_{\neg j}) &= \sum_{\vec u, \vec v \in \{0, 1\}^{k-1}} \xi_{\vec u \circ 1, \vec v \circ 0} \zeta_{\vec u, \vec v}(\vec x_{\neg j}, \vec x_{\neg j})
+
\sum_{\vec u, \vec v \in \{0, 1\}^{k-1}} \xi_{\vec u \circ 0, \vec v \circ 1} \zeta_{\vec u, \vec v}(\vec x_{\neg j}, \vec x_{\neg j}),    
\end{align}
where $\vec u \circ 0=u_1\dots u_{k-1} 0$ and $\vec u \circ 1=u_1\dots u_{k-1} 1$, and similarly for $\vec v \circ 0$ and $\vec v \circ 1$, and obtain Eq.~\eqref{eq:def_tmqf_2}.

Next, we prove the sufficiency of the above condition by induction on $k$. It is obvious for $k=1$. Suppose $f$ satisfies the given condition, i.e. $f(\vec x) = E_j(\vec x_{\neg j}) \cospt{x_j} + F_j(\vec x_{\neg j}) \sinpt{x_j} + G_j(\vec x_{\neg j}) \cosp{x_j} \sinp{x_j}$, for some $E_j$, $F_j$ and $G_j$, for all $j \in [k]$. Let $\vec w=(w_1,\dots,w_{k-1}) \in \mathbb{R}^{k-1}$ be arbitrary. Let $\vec a=\vec w \circ 0=(w_1,\dots,w_{k-1},0)$, $\vec b=\vec w \circ \pi/2=(w_1,\dots,w_{k-1},\pi/2)$ and $\vec c=\vec w \circ \pi/4=(w_1,\dots,w_{k-1},\pi/4)$. Then we have 
\begin{align}
E_k(\vec w) = f(\vec a) 
= \hat{E}_j(\vec w_{\neg j}) \cospt{w_j} + \hat{F}_j(\vec w_{\neg j}) \sinpt{w_j} + \hat{G}_j(\vec w_{\neg j}) \cosp{w_j} \sinp{w_j},~~~\forall j \in [k-1],
\end{align}
where $\hat{E}_j(\vec w_{\neg j})=E_j(\vec a_{\neg j})$, $\hat{F}_j(\vec w_{\neg j})=F_j(\vec a_{\neg j})$, and $\hat{G}_j(\vec w_{\neg j})=G_j(\vec a_{\neg j})$, and
\begin{align}
F_k(\vec w) = f(\vec b) = \bar{E}_j(\vec w_{\neg j}) \cospt{w_j} + \bar{F}_j(\vec w_{\neg j}) \sinpt{w_j} + \bar{G}_j(\vec w_{\neg j}) \cosp{w_j} \sinp{w_j},~~~\forall j \in [k-1],
\end{align}
where $\bar{E}_j(\vec w_{\neg j})=E_j(\vec b_{\neg j})$, $\bar{F}_j(\vec w_{\neg j})=F_j(\vec b_{\neg j})$, and $\bar{G}_j(\vec w_{\neg j})=G_j(\vec b_{\neg j})$, and
\begin{align}
G_k(\vec w) &= 2f(\vec c)-f(\vec a)-f(\vec b) \\
&= \tilde{E}_j(\vec w_{\neg j}) \cospt{w_j} + \tilde{F}_j(\vec w_{\neg j}) \sinpt{w_j} + \tilde{G}_j(\vec w_{\neg j}) \cosp{w_j} \sinp{w_j},~~~\forall j \in [k-1],
\end{align}
where $\tilde{E}_j(\vec w_{\neg j})=2E_j(\vec c_{\neg j})-E_j(\vec a_{\neg j})-E_j(\vec b_{\neg j})$, $\tilde{F}_j(\vec w_{\neg j})=2F_j(\vec c_{\neg j})-F_j(\vec a_{\neg j})-F_j(\vec b_{\neg j})$, and $\tilde{G}_j(\vec w_{\neg j})=2G_j(\vec c_{\neg j})-G_j(\vec a_{\neg j})-G_j(\vec b_{\neg j})$. This means that $E_k$, $F_k$ and $G_k$ all satisfy the above condition for $(k-1)$-ary functions. So by induction hypothesis, we know that $E_k$, $F_k$ and $G_k$ are all $(k-1)$-ary trigono-multiquadratic functions, i.e.~they can be expressed as in Eq.~\eqref{eq:tmqf}. It follows that $f(\vec x) = E_k(\vec x_{\neg k}) \cospt{x_k} + F_k(\vec x_{\neg k}) \sinpt{x_k} + G_k(\vec x_{\neg k}) \cosp{x_k} \sinp{x_k}$ can be also expressed as in Eq.~\eqref{eq:tmqf}, i.e. it is a $k$-ary trigono-multiquadratic function.  
\end{proof}

We say that $f\in \mathscr{T}^k_d$ is \textit{real} if its range is contained in $\mathbb R$, i.e.\ $f(\vec x) \in \mathbb R$ for all $\vec x \in \mathbb R^k$. It turns out that for real trigono-multilinear and trigono-multiquadratic functions, if we fix the values of all variables except $x_j$, then we can easily determine the value of $x_j$ that maximizes (or minimizes) the function, provided that we can efficiently evaluate the CSD or CSBD coefficient functions of the function with respect to $x_j$. 

Specifically, suppose $f: \mathbb{R}^k \to \mathbb{C}$ satisfies the condition in Proposition \ref{prop:tmlf}. Then
\begin{equation}
\argmax\limits_{y} f(x_1, \dots, x_{j-1}, y, x_{j+1}, \dots, x_k) = \mathrm{Arg}[C_j(\vec x_{\neg j}) + \i S_j(\vec x_{\neg j})],
\label{eq:optisol1}
\end{equation}
and
\begin{equation}
\max\limits_{y} f(x_1, \dots, x_{j-1}, y, x_{j+1}, \dots, x_k) = \sqrt{C_j(\vec x_{\neg j})^2 + S_j(\vec x_{\neg j})^2},
\end{equation}
where $\mathrm{Arg}(x+iy) = \mathrm{atan2}(y,x)$ is the 2-argument arctangent defined by
\begin{align}
    \mathrm{atan2}(y,x) = 
    \begin{cases}
    \arctan(y/x), & x>0, \\
    \arctan(y/x)+\pi, & x<0, y\geq 0, \\
    \arctan(y/x)-\pi, & x>0, y<0, \\
    \pi/2, & x=0, y>0, \\
    -\pi/2, & x=0, y<0, \\
    \mbox{undefined}, & x=y=0. \\
    \end{cases}
    \label{eq:atan2def}
\end{align}
Note that if $C_j(\vec x_{\neg j})=S_j(\vec x_{\neg j})=0$, then $f(\vec x)=0$ regardless of the value of $x_j$.

Similarly, suppose $f: \mathbb{R}^k \to \mathbb{C}$ satisfies the condition in Proposition \ref{prop:tmqf}. Then
\begin{equation}
\argmax\limits_{y} |f(x_1, \dots, x_{j-1}, y, x_{j+1}, \dots, x_k)| = \dfrac{\mathrm{Arg}[\sgn{B_j} (C_j(\vec x_{\neg j} ) + \i S_j(\vec x_{\neg j}))]}{2},
\label{eq:optisol2}
\end{equation}
and
\begin{equation}
\max\limits_{y} |f(x_1, \dots, x_{j-1}, y, x_{j+1}, \dots, x_k)| = \sqrt{C_j(\vec x_{\neg j})^2 + S_j(\vec x_{\neg j})^2} + |B_j(\vec x_{\neg j})|,
\end{equation}
where $\sgn{x}=1$ if $x \ge 0$ and $-1$ otherwise. Note that if $C_j(\vec x_{\neg j})=S_j(\vec x_{\neg j})=0$, then $f(\vec x)=B_j(\vec x_{\neg j})$ does not depend on $x_j$. So we can pick arbitrary $x_j \in \mathbb{R}$ to maximize $f(\vec x)$ in this case.

\section{Example: \texorpdfstring{$L=1$}{L1}}
\label{sec:L1example}

In this appendix, we consider the special case when $L=1$. In particular, we show how the results in Sections \ref{sec:quantumGeneratedLikelihoodFunctions} and \ref{sec:optimization} specialize in this case.

When $L=1$, Eq.~\eqref{eq:qalphabeta} becomes
\begin{align}
    Q(\theta;x_1,x_2) = V(x_2) U(\theta;x_1).
\end{align}

It is straightforward to check that the only nonempty sets $\Omega_{l,0}^4$,
$\Omega_{l,2}^4$ and
$\Gamma_l^2$ defined by Eq.~\eqref{eq:OmegaGammaSets} are
\begin{align}
    \Omega_{0,0}^4 &= \{0110\},
    \\
    \Omega_{2,2}^4 &= \{0011\},
    \\
    \Gamma_1^2 &= \{00,01,10\},
    \\
    \Gamma_3^2 &= \{11\}.
\end{align}

Upon substituting these into Eq.~\eqref{eq:FourierCoefficients2}, we get the following Fourier series expansion of the $L=1$ ancilla-free bias:
\begin{align}
    \Lambda^\AF(\theta;x_1,x_2) = \sum_{l=0}^3 \mu_l^\AF(x_1,x_2) \cos(l\theta),
    \label{eq:slopePrime}
\end{align}
where
\begin{align}
    \mu^\AF_0(x_1,x_2) &=
    2 \cos(x_1) \sin(x_2) \sin(x_1) \cos(x_2),
    \nonumber\\
    \mu^\AF_1(x_1,x_2) &=
    \cos^2(x_1) \cos^2(x_2)
    +
    \cos^2(x_1)\sin^2(x_2)
    +
    \cos^2(x_1) \sin^2(x_2),
    \nonumber\\
    \mu^\AF_2(x_1,x_2) &= -2 \cos(x_1) \cos(x_2) \sin(x_1) \sin(x_2),
    \nonumber\\
    \mu^\AF_3(x_1,x_2) &=
    \sin^2(x_1) \sin^2(x_2).
\end{align}

Therefore, using Eq.~\eqref{eq:bFourierSeries} and \eqref{eq:chiFourierSeries}, the variance reduction factor \eqref{eq:varRedFactorelf}
becomes
\begin{align}
    \mathcal V^\AF(\mu,\sigma; x_1,x_2) = 
    \frac{ \chi^\AF(\mu,\sigma;x_1, x_2)^2}{1-b^\AF(\mu,\sigma;x_1, x_2)^2},
\end{align}
where
\begin{align}
    b^\AF(\mu,\sigma;x_1,x_2) &= 2 \cos(x_1) \sin(x_2) \sin(x_1) \cos(x_2) \nonumber \\ 
    &\quad+ (\cos^2(x_1) \cos^2(x_2)
    +
    \cos^2(x_1)\sin^2(x_2)
    +
    \cos^2(x_1) \sin^2(x_2)) \e^{-\sigma^2/2} \cos(\mu)
    \nonumber\\
    &\quad - 2 \cos(x_1) \cos(x_2) \sin(x_1) \sin(x_2)
    \e^{-2\sigma^2/2} \cos(2\mu)
    \nonumber\\
    &\quad +
    \sin^2(x_1) \sin^2(x_2) 
    \e^{-9\sigma^2/2} \cos(3\mu)
\end{align}
and
\begin{align}
    \chi^\AF(\mu,\sigma;x_1,x_2) &=  -(\cos^2(x_1) \cos^2(x_2)
    +
    \cos^2(x_1)\sin^2(x_2)
    +
    \cos^2(x_1) \sin^2(x_2)) \e^{-\sigma^2/2} \sin(\mu)
    \nonumber\\
    &\quad + 4 \cos(x_1) \cos(x_2) \sin(x_1) \sin(x_2)
    \e^{-2\sigma^2/2} \sin(2\mu)
    \nonumber\\
    &\quad -3
    \sin^2(x_1) \sin^2(x_2) 
    \e^{-9\sigma^2/2} \sin(3\mu).
\end{align}

In the ancilla-based case, Eq.~\eqref{eq:cosPolyCoefAB} evaluates to

\begin{align}
    \Lambda^\AB(\theta;x_1,x_2) = \sum_{l=0}^1 \mu_l^\AB(x_1,x_2) \cos(l\theta),
    \label{eq:slopePrimeAB}
\end{align}
where
\begin{align}
    \mu^\AB_0(x_1,x_2) &=
    \cos(x_1) \cos(x_2),
    \nonumber\\
    \mu^\AB_1(x_1,x_2) &=
    -\sin(x_1) \sin(x_2).
    \end{align}

Therefore, using Eq.~\eqref{eq:bFourierSeries} and \eqref{eq:chiFourierSeries}, the variance reduction factor \eqref{eq:varRedFactorelf}
becomes
\begin{align}
    \mathcal V^\AB(\mu,\sigma; x_1,x_2) = 
    \frac{ \chi^\AB(\mu,\sigma;x_1, x_2)^2}{1-b^\AB(\mu,\sigma;x_1, x_2)^2},
    \label{eq:ABL1VRF}
\end{align}
where
\begin{align}
    b^\AB(\mu,\sigma;x_1,x_2) &= \cos(x_1)\cos(x_2) - \sin(x_1)\sin(x_2)\e^{-\sigma^2/2}\cos(\mu)
    \label{eq:bABL1}
\end{align}
and
\begin{align}
    \chi^\AB(\mu,\sigma;x_1,x_2) &=   \sin(x_1) \sin(x_2)\e^{-\sigma^2/2}\sin(\mu).
    \label{eq:chiABL1}
\end{align}

\subsection{The \texorpdfstring{$L=1$}{L=1} ancilla-based Chebyshev likelihood function is optimal}
\label{app:optimalityL1}

In this section, we shall prove Eq.~\eqref{eq:ABchebyoptimal}, which says that for the $L=1$ ancilla-based scheme, the performance of the optimized engineered likelihood function is identical to that of the Chebyshev likelihood function, i.e.~the Chebyshev likelihood function is optimal. This optimality can be observed from the numerical results obtained in Figure \ref{fig:grid_L1}.

To prove Eq.~\eqref{eq:ABchebyoptimal}, we first rewrite the variance reduction factor
\eqref{eq:ABL1VRF} for the $L=1$ ancilla-based scheme as
\begin{align}
    \mathcal V^{\mathrm{AB}}(\mu,\sigma;x_1,x_2)=
    \frac{t^2 \sin^2 x_1 \sin^2 x_2}{1-(\cos x_1 \cos x_2 - s \sin x_1 \sin x_2)^2},
    \label{eq:varianceRFL1AB}
\end{align}
where
\begin{align}
    s &= \e^{-\sigma^2/2} \cos\mu , \\
    t &= \e^{-\sigma^2/2} \sin\mu .
    \label{eq:deftshort}
\end{align}

We first state a few observations about Eq.~\eqref{eq:varianceRFL1AB}--\eqref{eq:deftshort}. First, for $\sigma>0$ and $\mu \in \mathbb R$, we have $|s|,|t| <1$. Second, we note that the formula in \eqref{eq:varianceRFL1AB} is undefined whenever its denominator vanishes, which occurs whenever $\cos x_1 \cos x_2 - s \sin x_1 \sin x_2=\pm 1$. As the following lemma proves, this occurs if and only if
$x_1, x_2 \in \pi \mathbb Z$. 
\begin{lemma}
Let $x_1,x_2\in \mathbb R$ and $|s|<1$. Then,
\begin{align}
    |\cos x_1 \cos x_2 - s \sin x_1 \sin x_2|= 1 \iff x_1, x_2 \in \pi \mathbb Z.
\end{align} 
\label{lem:when_vanishes}
\end{lemma}
\begin{proof}
The backward implication follows directly from having $|\cos x_1| = |\cos x_2| = 1$ and $\sin x_1 = \sin x_2 = 0$ whenever $x_1, x_2 \in \pi \mathbb Z$. To prove the forward implication, assume that $|\cos x_1 \cos x_2 - s \sin x_1 \sin x_2|= 1$. To proceed, it will be convenient to consider the following 4 cases separately (see Table \ref{tab:cases_proof}): (i) $\cos x_1 \cos x_2 - s \sin x_1 \sin x_2= 1$ and $\sin x_1 \sin x_2 \geq 0$; 
(ii) $\cos x_1 \cos x_2 - s \sin x_1 \sin x_2= 1$ and $\sin x_1 \sin x_2 \leq 0$; 
(iii) $\cos x_1 \cos x_2 - s \sin x_1 \sin x_2= -1$ and $\sin x_1 \sin x_2 \geq 0$; 
(iv) $\cos x_1 \cos x_2 - s \sin x_1 \sin x_2= -1$ and $\sin x_1 \sin x_2 \leq 0$.

\begin{table}[h!]
\caption{4 different cases we consider in the proof of Lemma \ref{lem:when_vanishes}}
\label{tab:cases_proof}
\begin{center}
\begin{tabular}{cc|cc}
                      &          & \multicolumn{2}{c}{$\sin x_1 \sin x_2$}    \\ \cline{3-4} 
                      &          & \multicolumn{1}{c|}{$\geq 0$}       & $\leq 0$      \\ \hline
\multicolumn{1}{c|}{\multirow{2}{*}{$\cos x_1 \cos x_2 - s \sin x_1 \sin x_2$}} & $1$ & \multicolumn{1}{c|}{case (i)} & case (ii) \\ \cline{2-4} 
\multicolumn{1}{c|}{} & $-1$ & \multicolumn{1}{c|}{case (iii)} & case (iv)
\end{tabular}
\end{center}
\end{table}
\noindent \underline{Case 1}: Assume that $\cos x_1 \cos x_2 - s \sin x_1 \sin x_2= 1$ and $\sin x_1 \sin x_2 \geq 0$. Then,
\begin{align}
    1 &= \cos x_1 \cos x_2 - s \sin x_1 \sin x_2
    \label{eq:1cosx1costx2-s}
    \\
    &\leq 
    \cos x_1 \cos x_2 + \sin x_1 \sin x_2
    \\
    &=
    \cos(x_1 -x_2) \leq 1,
\end{align}
where the second line follows from our assumption that $s > -1$. Hence, $\cos(x_1-x_2) = 1$, which implies that $x_1 = x_2 + 2\pi k$ for some $k \in \mathbb Z$. Substituting this into Eq.~\eqref{eq:1cosx1costx2-s} gives
\begin{align}
    1 &= \cos(x_2+ 2\pi k) \cos x_2 - s \sin(x_2+2\pi k) \sin x_2 \nonumber\\
    &= \cos^2 x_2 - s \sin^2 x_2 \nonumber\\
    &=(1+s) \cos^2 x_2 -s,
\end{align}
which implies that
\begin{align}
    1+s = (1+s) \cos^2 x_2.
\end{align}
Since $s\neq -1$, it follows that $\cos^2 x_2 = 1$, which implies that $x_2 = \pi l \in \pi \mathbb Z$ for some $l\in \mathbb Z$. Hence, $x_1 = \pi l + 2\pi k \in \pi \mathbb Z$.

\noindent \underline{Case 2}: Assume that $\cos x_1 \cos x_2 - s \sin x_1 \sin x_2= 1$ and $\sin x_1 \sin x_2 \leq 0$. Then,
\begin{align}
    1 &= \cos x_1 \cos x_2 - s \sin x_1 \sin x_2
    \label{eq:1cosx1costx2-s2}
    \\
    &\leq 
    \cos x_1 \cos x_2 - \sin x_1 \sin x_2
    \\
    &=
    \cos(x_1 +x_2) \leq 1,
\end{align}
where the second line follows from our assumption that $s < 1$. Hence, $\cos(x_1+x_2) = 1$, which implies that $x_1 = -x_2 + 2\pi k$ for some $k \in \mathbb Z$. Substituting this into Eq.~\eqref{eq:1cosx1costx2-s2} gives
\begin{align}
    1 &= \cos(-x_2+ 2\pi k) \cos x_2 - s \sin(-x_2+2\pi k) \sin x_2 \nonumber\\
    &= \cos^2 x_2 + s \sin^2 x_2 \nonumber\\
    &=(1-s) \cos^2 x_2 +s,
\end{align}
which implies that
\begin{align}
    1-s = (1-s) \cos^2 x_2.
\end{align}
Since $s\neq 1$, it follows that $\cos^2 x_2 = 1$, which implies that $x_2 = \pi l \in \pi \mathbb Z$ for some $l\in \mathbb Z$. Hence, $x_1 = \pi l + 2\pi k \in \pi \mathbb Z$.

\noindent \underline{Case 3}: Assume that $\cos x_1 \cos x_2 - s \sin x_1 \sin x_2= -1$ and $\sin x_1 \sin x_2 \geq 0$. Then,
\begin{align}
    -1 &= \cos x_1 \cos x_2 - s \sin x_1 \sin x_2
    \label{eq:1cosx1costx2-s3}
    \\
    &\geq 
    \cos x_1 \cos x_2 - \sin x_1 \sin x_2
    \\
    &=
    \cos(x_1 +x_2) \geq -1,
\end{align}
where the second line follows from our assumption that $s < 1$. Hence, $\cos(x_1+x_2) = -1$, which implies that $x_1 = -x_2 + (2k+1)\pi$ for some $k \in \mathbb Z$. Substituting this into Eq.~\eqref{eq:1cosx1costx2-s3} gives
\begin{align}
    -1 &= \cos(-x_2+ (2k+1)\pi ) \cos x_2 - s \sin(-x_2+(2k+1)\pi) \sin x_2 \nonumber\\
    &= -\cos^2 x_2 - s \sin^2 x_2 \nonumber\\
    &=(s-1) \cos^2 x_2 -s,
\end{align}
which implies that
\begin{align}
    s-1 = (s-1) \cos^2 x_2.
\end{align}
Since $s\neq 1$, it follows that $\cos^2 x_2 = 1$, which implies that $x_2 = \pi l \in \pi \mathbb Z$ for some $l\in \mathbb Z$. Hence, $x_1 = -\pi l + (2k+1)\pi \in \pi \mathbb Z$.

\noindent \underline{Case 4}: Assume that $\cos x_1 \cos x_2 - s \sin x_1 \sin x_2= -1$ and $\sin x_1 \sin x_2 \leq 0$. Then,
\begin{align}
    -1 &= \cos x_1 \cos x_2 - s \sin x_1 \sin x_2
    \label{eq:1cosx1costx2-s4}
    \\
    &\geq 
    \cos x_1 \cos x_2 + \sin x_1 \sin x_2
    \\
    &=
    \cos(x_1 -x_2) \geq -1,
\end{align}
where the second line follows from our assumption that $s > -1$. Hence, $\cos(x_1-x_2) = -1$, which implies that $x_1 = x_2 + (2k+1)\pi $ for some $k \in \mathbb Z$. Substituting this into Eq.~\eqref{eq:1cosx1costx2-s4} gives
\begin{align}
    -1 &= \cos(x_2+ (2k+1)\pi ) \cos x_2 - s \sin(x_2+(2k+1)\pi) \sin x_2 \nonumber\\
    &= -\cos^2 x_2 + s \sin^2 x_2 \nonumber\\
    &=(-1-s) \cos^2 x_2 +s,
\end{align}
which implies that
\begin{align}
    -1-s = (-1-s) \cos^2 x_2.
\end{align}
Since $s\neq -1$, it follows that $\cos^2 x_2 = 1$, which implies that $x_2 = \pi l \in \pi \mathbb Z$ for some $l\in \mathbb Z$. Hence, $x_1 = \pi l + (2k+1)\pi \in \pi \mathbb Z$.
\end{proof}

For the pairs $(\hat x_1,\hat x_2)$ satisfying the condition $\hat x_1, \hat x_2 \in \pi \mathbb Z$, we have $\lim_{(x_1,x_2)\rightarrow (\hat x_1,\hat x_2)} V^{\mathrm{AB}}(\mu,\sigma;x_1,x_2)= 0$, which is less than the maximum of the variance reduction factor $V^{\mathrm{AB}}(\mu,\sigma;x_1,x_2)$ (to see this, it suffices to pick any $\hat x_1, \hat x_2 \notin \pi \mathbb Z$). For this reason, when characterizing the maximum points of $V^{\mathrm{AB}}(\mu,\sigma;x_1,x_2)$, it suffices to just consider pairs $(x_1,x_2)$ for which Eq.~\eqref{eq:varianceRFL1AB} is well-defined.

We now state and prove a lemma that can be used to provide a tight upper bound for Eq.~\eqref{eq:varianceRFL1AB}:
\begin{lemma}
For all $x_1, x_2 \in \mathbb R \backslash \pi \mathbb Z$ and for all $|s|<1$,
\begin{align}
   \frac{\sin^2 x_1 \sin^2 x_2}{1-(\cos x_1 \cos x_2 - s \sin x_1 \sin x_2)^2}  \leq \frac{1}{1-s^2}
   \label{eq:sin2inequality}
\end{align}
with equality obtained when $x_1, x_2 \in \tfrac{\pi}{2} \mathbb Z_{\mathrm{odd}}$ are odd multiples of $\pi/2$.
\label{lem:upperBoundFor}
\end{lemma}
\begin{proof}

Equality in \eqref{eq:sin2inequality} holds whenever $x_1, x_2 \in \tfrac{\pi}{2} \mathbb Z_{\mathrm{odd}}$ since $|\sin x_1| = |\sin x_2| = 1$ and $\cos x_1 = \cos x_2 = 0$. To prove the inequality \eqref{eq:sin2inequality},
we shall first show that
\begin{align}
    \cos^2 x_1 \cos^2 x_2 - 2s \cos x_1 \cos x_2 \sin x_1 \sin x_2 + \sin^2 x_1 \sin^2 x_2 \leq 1.
    \label{eq:firstprovethis}
\end{align}

To see this, we first consider the case when $\cos x_1 \cos x_2 \sin x_1 \sin x_2\geq 0$. Since $s\geq -1$, it follows that
\begin{align}
    & \cos^2 x_1 \cos^2 x_2 - 2s \cos x_1 \cos x_2 \sin x_1 \sin x_2 + \sin^2 x_1 \sin^2 x_2 \nonumber\\
    &\qquad\leq 
    \cos^2 x_1 \cos^2 x_2 + 2 \cos x_1 \cos x_2 \sin x_1 \sin x_2 + \sin^2 x_1 \sin^2 x_2 \nonumber\\
    &\qquad=
    (\cos x_1 \cos x_2 + \sin x_1 \sin x_2)^2 \nonumber\\
    &\qquad= \cos^2(x_1-x_2) \nonumber\\
    &\qquad\leq 1.
\end{align}

Next, consider the case when $\cos x_1 \cos x_2 \sin x_1 \sin x_2\leq 0$. Since $s\leq 1$, it follows that
\begin{align}
    & \cos^2 x_1 \cos^2 x_2 - 2s \cos x_1 \cos x_2 \sin x_1 \sin x_2 + \sin^2 x_1 \sin^2 x_2 \nonumber\\
    &\qquad\leq 
    \cos^2 x_1 \cos^2 x_2 - 2 \cos x_1 \cos x_2 \sin x_1 \sin x_2 + \sin^2 x_1 \sin^2 x_2 \nonumber\\
    &\qquad=
    (\cos x_1 \cos x_2 - \sin x_1 \sin x_2)^2 \nonumber\\
    &\qquad= \cos^2(x_1+x_2) \nonumber\\
    &\qquad\leq 1.
\end{align}
This completes the proof of Eq.~\eqref{eq:firstprovethis}. Now, rearranging the inequality 
\eqref{eq:firstprovethis} gives
\begin{align}
    (1-s^2) \sin^2 x_1 \sin^2 x_2 \leq 
    1- (\cos x_1 \cos x_2 - s \sin x_1 \sin x_2)^2.
    \label{eq:rearranged}
\end{align}
Since $s^2<1$, both sides of Eq.~\eqref{eq:rearranged} are nonnegative. Furthermore, since $x_1, x_2 \notin \pi\mathbb Z$, Lemma \ref{lem:when_vanishes} implies that the right-hand-side of Eq.~\eqref{eq:rearranged} is nonzero. Hence,
\begin{align}
   \frac{\sin^2 x_1 \sin^2 x_2}{1-(\cos x_1 \cos x_2 - s \sin x_1 \sin x_2)^2}  \leq \frac{1}{1-s^2}.
\end{align}

\end{proof}

Applying Lemma \ref{lem:upperBoundFor} to Eq.~\eqref{eq:varianceRFL1AB} gives 
\begin{align}
    \mathcal V^{\mathrm{AB}}(\mu,\sigma;x_1,x_2)=
    \frac{t^2 \sin^2 x_1 \sin^2 x_2}{1-(\cos x_1 \cos x_2 - s \sin x_1 \sin x_2)^2}
    \leq \frac{t^2}{1-s^2}.
\end{align}

Since $(x_1,x_2) = (\tfrac \pi 2,\tfrac \pi 2)$ saturates this bound (by Lemma \ref{lem:upperBoundFor}), it follows that
\begin{align}
    \displaystyle \argmax_{x_1, x_2 \in (-\pi,\pi] }\mathcal V^{\mathrm{AB}}(\mu,\sigma; x_1, x_2)
    \ni (\tfrac \pi 2,\tfrac \pi 2).
\end{align}
This completes the proof of Eq.~\eqref{eq:ABchebyoptimal}.

\section{Leading terms in the cosine series expansions of the biases}
\label{sec:leading}

In this appendix, we use the expansion formulas in Appendix \ref{sec:expansion_elf_bias} to show that the leading terms of the cosine expansions \eqref{eq:cosPolyCoef} and \eqref{eq:cosPolyCoefAB} can be written as products of sine functions.

\begin{proposition}
Let $\vec x \in \mathbb R^{2L}$. Then the leading terms of the cosine expansions \eqref{eq:cosPolyCoef} and \eqref{eq:cosPolyCoefAB} are given by
\begin{align}
    \mu_{2L+1}^\AF(\vec x) &= \prod_{i=1}^{2L} \sin^2(x_i),
    \label{sec:leadingAF}
  \\
    \mu_L^\AB(\vec x) &= (-1)^L \prod_{i=1}^{2L} \sin(x_i).
    \label{sec:leadingAB}
\end{align}
\end{proposition}
\begin{proof}
To prove Eq.~\eqref{sec:leadingAF}, we first recall the expression for the leading coefficient of Eq.~\eqref{eq:cosPolyCoef}
    \begin{align}
        \mu_{2L+1}^{\AF}(\vec x) = 
        \left(
        \raisebox{3mm}{
        $\displaystyle\sum_{
        \underset{
        \wt(\vec c)-\wt(\vec a) \equiv_4 \, 0
        }{
        \vec a1\vec c^R \in \Xi_{2L+1}^{4L+1}
        }
        }
        -
        \displaystyle\sum_{
        \underset{
        \wt(\vec c)-\wt(\vec a) \equiv_4 \, 2
        }{
        \vec a1\vec c^R \in \Xi_{2L+1}^{4L+1}
        }
        }$
        }
        \right) \zeta_{\vec a\vec c}(\vec x,\vec x).
    \end{align}

By Eq.~\eqref{eq:XiAsUnionThetas},
\begin{align}
    \Xi^{4L+1}_{2L+1}
    =
    \Theta^{4L+1}_{0,2L+1,0} \cup 
    \Theta^{4L+1}_{1,2L+1,0} \cup
    \Theta^{4L+1}_{0,2L,1} \cup 
    \Theta^{4L+1}_{1,2L,1}.
\end{align}

Now, setting $n=4L+1$ in Lemma \ref{lem:nonZeroTheta} gives
\begin{align}
    k > 2L-u \implies \Theta^{4L+1}_{ukv} = \emptyset.
    \label{eq:impl}
\end{align}

Since $2L+1>2L-u$ for $u\in \{0,1\}$, Eq.~\eqref{eq:impl} implies that
\begin{align}
    \Theta^{4L+1}_{0,2L+1,0} = 
    \Theta^{4L+1}_{1,2L+1,0} = \emptyset .
\end{align}

Since $2L>2L-1$, Eq.~\eqref{eq:impl} implies that
\begin{align}
    \Theta^{4L+1}_{1,2L,1} =  \emptyset .
\end{align}

Therefore,
\begin{align}
    \Xi^{4L+1}_{2L+1}
    &=
    \Theta^{4L+1}_{0,2L,1} \nonumber\\
    &=\{
    \vec x \in \{0,1\}^{4L+1}:
    q^{4L+1}p^{4L}\ldots q^{x_3}p^{x_2}q^{x_1} \sim (qp)^{2L}q^v
    \}
    \nonumber\\
    &=
    \{1^{4L+1}\} 
    \nonumber\\
    &= \{1^{2L}\cdot1\cdot1^{2L}\},
\end{align}
where $\cdot$ means string concatenation.

Hence, the set of strings $\vec a\vec c$ satisfying $\vec a 1 \vec c^R \in \Xi^{4L+1}_{2L+1}$ and $\wt(\vec c)-\wt(\vec a) \equiv_4 0$ consists of only the element $1^{4L}$ and the set of strings $\vec a\vec c$ satisfying $\vec a 1 \vec c^R \in \Xi^{4L+1}_{2L+1}$ and $\wt(\vec c)-\wt(\vec a) \equiv_4 2$ is the empty set. Therefore,
\begin{align}
    \mu_{2L+1}^\AF(\vec x)
    &=
    \zeta_{1^{4L}}(\vec x,\vec x)
    \nonumber\\
    &= \prod_{i=1}^{2L} \zeta_i(x_i)^2
    \nonumber\\
    &= \prod_{i=1}^{2L} \sin^2(x_i).
\end{align}

To prove Eq.~\eqref{sec:leadingAB}, we recall the expression for the leading coefficient of Eq.~\eqref{eq:cosPolyCoefAB} 
    \begin{align}
    \mu_L^{\mathrm{AB}}(\vec x) =
\left(
        \raisebox{3mm}{
        $\displaystyle\sum_{
        \underset{
        \wt(\vec y) \equiv_4 \, 0
        }{
        \vec y \in \Xi_L^{2L}
        }
        }
        -
        \displaystyle\sum_{
        \underset{
        \wt(\vec y) \equiv_4 \, 2
        }{
        \vec y \in \Xi_L^{2L}
        }
        }$
        }
        \right) \zeta_{\vec y}(\vec x).
        \label{eq:cosPolyCoefAB1}
    \end{align}

But
\begin{align}
    \Xi_L^{2L} &=
    \underbrace{\Theta_{0,L,0}^{2L}
    }_{=\emptyset}
    \cup
    \underbrace{\Theta_{1,L,0}^{2L}
    }_{=\emptyset}
    \cup
    \Theta_{0,L-1,1}^{2L}
        \cup
    \Theta_{1,L-1,1}^{2L}
    \nonumber\\
    &=
    \{
    \vec x \in \{0,1\}^{2L}:
    p^{x_{2L}} q^{x_{2L-1}} \ldots 
    p^{x_2}
    q^{x_1}
    \sim
    p^u (qp)^{L-1} q, 
    u \in \{0,1\}
    \}
    \nonumber\\
    &=
    \{ u 1^{2L-1}: u \in \{0,1\}\}
    \nonumber\\
    &=
    \{0 1^{2L-1},1^{2L}
    \}.
\end{align}

Now,
\begin{align}
    \wt(01^{2L-1} &= 2L-1 \neq 0 \mbox{ or } 2 \pmod{4} 
    \nonumber\\
    \wt(1^{2L}) = 2L
    &=
    \begin{cases}
    0\pmod{4} & L \mbox{ even},
    \\
    2\pmod{4} & L \mbox{ odd}.
    \end{cases}
\end{align}
Hence, the values of the sets $\{\vec y \in \Xi_L^{2L}:\wt(\vec y) \equiv_4 k\}$, for $k=0,2$, are described by the following table:
\begin{center}
\begin{tabular}{c|c|c}
 &
 $\{\vec y \in \Xi_L^{2L}:\wt(\vec y) \equiv_4 0$
 &
 $\{\vec y \in \Xi_L^{2L}:\wt(\vec y) \equiv_4 2$
 \\
 \hline
$L$ even & $\{1^{2L}\}$ & $\emptyset$
\\
$L$ odd & $\emptyset$ & $\{1^{2L}\}$
\end{tabular}
\end{center}

Consequently,
\begin{align}
    \mu_L^\AB(\vec x)
    &=
    \begin{cases}
    \zeta_{1^{2L}}(\vec x) & L\mbox{ even}
    \\
    -\zeta_{1^{2L}}(\vec x) & L\mbox{ odd}
    \end{cases}
    \nonumber\\
    &=(-1)^L
    \zeta_{1^{2L}}(\vec x)
    \nonumber\\
    &=
    (-1)^L \prod_{i=1}^{2L} \sin(x_i).
\end{align}

\end{proof}

\section{On noisy likelihood functions}
\label{sec:noise}

As we show in \cite{ELFPaper}, depolarizing noise that occurs after each rotation operator $V(\cdot)$ in 
the circuits in Figure \ref{fig:elfcircuit} and/or depolarizing noise during measurement lead to likelihood functions that are of the form\footnote{As described in \cite{ELFPaper}, in the context of randomized benchmarking, such noise could arise from state preparation and
measurement (SPAM) errors \cite{gambetta2012characterization, sun2018efficient}.}
\begin{align}
    \mathcal L_{\mathrm{noisy}}^\mathcal A(\theta;d,\vec x) =\frac 12\left[
    1+(-1)^d f \Lambda^\mathcal A(\theta;\vec x)
    \right]
    \label{eq:likelihoodInTermsOfBiasSchemesNoisy}
\end{align}
for some \textit{fidelity} parameter $f\in[0,1)$ \cite{ELFPaper}. 
In other words, the effect of noise transforms the bias as
\begin{align}
    \Lambda^\mathcal A \rightarrow f \Lambda^\mathcal A.
\end{align}
Since the bias \eqref{eq:bchiV1} and the chi function \eqref{eq:bchiV2} are linear in $\Lambda^\mathcal A$, they transform as
    \begin{align}
    b^{\mathcal A} &\rightarrow fb^{\mathcal A}, \\
    \chi^{\mathcal A} &\rightarrow f\chi^{\mathcal A}.
    \end{align}
Consequently, the variance reduction factor \eqref{eq:bchiV3}  takes the form
\begin{align}
   \mathcal V_{\mathrm{noisy}}^{\mathcal A}(\mu,\sigma;\vec x)= \frac{f^2 \chi^{\mathcal A}(\mu,\sigma;\vec x)^2}{1-f^2 b^{\mathcal A}(\mu,\sigma;\vec x)^2}.
\end{align}

\bibliographystyle{unsrt}
\bibliography{ref}

\end{document}